\documentclass[12pt,english]{article}

\usepackage{geometry}
\usepackage[protrusion=true,expansion=false]{microtype}
\usepackage{setspace}
\usepackage{etoolbox}
\usepackage{babel}
\usepackage{hyperref}
\usepackage{placeins}
\usepackage{comment}
\usepackage{subcaption}
\usepackage{csquotes}
\usepackage{pdflscape}
\usepackage{float}
\usepackage{algorithm}
\usepackage{algpseudocode}
\usepackage{xcolor}
\usepackage[normalem]{ulem}
\usepackage{manuscript-root}
\usepackage[
  backend=biber,
  style=apa,
  doi=false,
  eprint=true,
  isbn=false,
  sorting=nyt,
  url=false,
  natbib=true
]{biblatex}
\AtEveryBibitem{%
  \iffieldequalstr{eprinttype}{arxiv}
    {}%
    {%
      \clearfield{eprint}%
      \clearfield{eprinttype}%
      \clearfield{eprintclass}%
    }%
}
\DeclareFieldFormat{eprint:arxiv}{%
  \href{https://arxiv.org/abs/#1}{arXiv\addspace #1}%
}
\renewbibmacro*{location+publisher}{\printlist{publisher}}

\graphicspath{{figures/}}
\makeatletter

\renewenvironment{abstract}
  {\par\noindent\textbf{\abstractname.}\ \ignorespaces}
  {\par}

\patchcmd{\@maketitle}{\vskip 2em}{\vskip 0.5em}{}{}
\patchcmd{\@maketitle}{\vskip 1.5em}{\vskip 0.5em}{}{}
\patchcmd{\@maketitle}{\vskip 1em}{\vskip 0.5em}{}{}
\patchcmd{\@maketitle}{\vskip 1.5em}{\vskip 0.5em}{}{}

\makeatother

\begin{document}
\title{%
  An Adversarial Approach to Identification, Computation, and Inference in Models with a Linear-in-Measures Representation
  \thanks{Botosaru: Department of Economics, McMaster University, \texttt{botosari@mcmaster.ca}.
  Loh: Department of Economics, University of North Carolina Wilmington, \texttt{lohi@uncw.edu}.
  Muris: Department of Economics, McMaster University, \texttt{muerisc@mcmaster.ca}.
  We thank Andr\'es Aradillas-L\'opez, Tim Christensen, Inga Deimen, Jiaying Gu, Bo Honor\'e, Hide Ichimura, Hiro Kasahara, Vadim Marmer, Francesca Molinari, Adam Rosen, Rami Tabri, Alex Torgovitsky, and Victoria Zinde-Walsh for discussions and suggestions.
  We also thank audiences at various seminars and conferences for questions and comments.
  Botosaru gratefully acknowledges financial support from the Canada Research Chairs Program.
  Muris gratefully acknowledges financial support from the Social Sciences and Humanities Research Council of Canada (Insight Grant 435-2025-1380).
  }%
}
\author{Irene Botosaru \and Isaac Loh \and Chris Muris}
\date{September 25, 2026}
\maketitle

\begin{abstract}
We develop a framework for identification, computation, and inference in econometric models with a \emph{linear-in-measures representation}.
These models express maintained restrictions as moment conditions linear in the joint probability measure of observed and latent inputs, and map that measure linearly to the distribution of outputs, even with nonlinear outcome equations.
We construct an \emph{adversarial discrepancy function} whose zeros characterize the identified set for structural and counterfactual parameters.
With finite output support, finite linear programs compute the discrepancy function or provide certified bounds even when latent inputs have infinite support, and a penalized bootstrap yields confidence sets with uniform per-point coverage.
We apply the framework to two open cases in binary choice panels with fixed effects and discrete covariates: sequential exogeneity with unspecified conditional marginal error distributions, and known conditional marginal error distributions with unrestricted serial dependence.
In an entry game with multiple equilibria, the framework recovers the known sharp identification region.
\end{abstract}

\smallskip
\noindent\textbf{Keywords:} identification; latent-variable models; nonlinear panel models; sequential exogeneity; linear programming; uniform inference.\par
\noindent\textbf{JEL:} C12, C14, C23, C61.

\section{Introduction}
\label{sec:intro}

\subsection{Motivation and approach}

We develop a new approach to identification, computation, and inference for a broad class of econometric models. 
A model in this class specifies two objects: (i) moment conditions that restrict the joint probability measure of its \emph{inputs}, the set of observed and latent variables from which it generates its \emph{outputs};%
\footnote{``Input'' and ``output'' refer to roles in the representation. 
Observed inputs may be covariates, instruments, initial conditions, or the endpoints of an interval-censored covariate.
These variables are also outputs, which the model returns unchanged. 
Outcomes are outputs generated from the observed and latent inputs.}
and (ii) the conditional distribution of the outputs given the inputs.
Averaging this conditional distribution over the \emph{input measure} yields the output distribution.
Thus, at each parameter value, the model has a \emph{linear-in-measures representation} (LIMR) because both the moment restrictions and the mapping to the output distribution are linear in the input measure.
A parameter value belongs to the identified set if some input measure satisfies the moment restrictions at that value and reproduces the observed output distribution.

Using the LIMR and a separation argument in the space of distributions, we construct the \emph{adversarial discrepancy function} (ADF), a criterion function whose zeros characterize the identified set.
This places the adversarial approach in the criterion function tradition of \textcite{manskiInferenceRegressionsInterval2002} and \textcite{Chernozhukovetal2007}.
Here the ADF has a zero-sum-game form: a test function seeks to separate the observed output distribution from those generated by the admissible input measures, while the model seeks an input measure that minimizes the resulting discrepancy.
Identification requires neither an explicit characterization of the set of model-implied output distributions nor a separate characterization of its observable implications.

When the outputs take finitely many values, even when the latent inputs have infinite support, finite linear programs (LPs) either compute the ADF exactly or yield certified bounds.
The bounds are valid at every certified iterate, so no \textit{a priori} convergence guarantee is required.
The test statistic is the sample analog of the ADF, computed by the same LPs.
Inference is by test inversion, with critical values from a penalized bootstrap, and the resulting confidence sets have uniform per-point coverage.

Within the LIMR class, specifications differ in the choice of inputs, the restrictions on their joint probability measure, and the conditional distribution of the outputs given the inputs.
The ADF, and hence the LPs and the bootstrap procedure, is constructed from these objects.
Thus, changes in the outcome equation or maintained restrictions alter the ingredients of the same construction rather than requiring a new identification, computation, or inference method.

We demonstrate the scope of the LIMR in short nonlinear panel models with fixed effects.
We address open questions in binary choice panel models under sequential exogeneity without parametric restrictions on the error distributions, and under parametrically specified marginal error distributions with unrestricted serial dependence.
We analyze additional specifications that allow unrestricted dependence between the fixed effect and the error terms.
We also find that the framework recovers the point identification benchmark for the panel logit model with fixed effects.
The LIMR class also includes models with interval-censored covariates and dynamic discrete choice, and extends beyond panels to entry games with multiple equilibria and randomized trials with imperfect compliance.
Section~\ref{sec:examples} constructs a LIMR for each model mentioned above.

We call the map from input measures to output distributions the \emph{output operator}.
Linearity in the input measure restricts neither the outcome equation nor the moment functions, which may be nonlinear or nonseparable.
The moment conditions may be unconditional or conditional on any subset of the inputs and may form a continuum, as when a restriction holds at every value of a continuously distributed fixed effect.
They include zero-mean and median restrictions, parametric distributional restrictions, stationarity restrictions, and definitions of counterfactual parameters such as the average structural function (ASF).
Adding an equilibrium selector accommodates incomplete entry games, while a change of inputs yields a LIMR under sequential exogeneity in a nonlinear panel model with fixed effects.

In nonlinear panel models with fixed effects and few time periods, repeated observations do not in general reveal the individual fixed effect.
Hence, its unrestricted distribution given the covariates is an infinite-dimensional latent object.
The fixed effect also enters counterfactual parameters that average over its distribution, such as the ASF.
Under special parametric structure, the fixed effect can be eliminated through a sufficient statistic or a functional differencing transformation for the identification of some structural coefficients.
More generally, restrictions on the errors are often stated conditional on the fixed effect and must therefore hold at every value of its unknown support.
Structural and counterfactual parameters may consequently be only partially identified.

A substantial literature has developed powerful identification methods for important specifications of these models, both for point identification \parencite[e.g.,][]{cham1980,Manski1987,honorePanelDataDiscrete2000,Bonhomme2012} and for partial identification of structural coefficients and counterfactual parameters \parencite[e.g.,][]{honoreBoundsParametersPanel2006,ChernValHahnNewey2013,daveziesIdentificationEstimationAverage2022,botosaruMuris2025}.
These results exploit various restrictions on the errors and use constructions tailored to the outcome equation and the parameter of interest.
For example, stationarity restrictions on error distributions yield model-specific inequalities that depend only on observables \parencite[e.g.,][]{khanIdentificationDynamicBinary2023,gaoIdentificationNonlinearDynamic2024,pakesPorterMomentInequalities,mbakopIdentificationSomeDiscrete2023}, while with parametric errors functional differencing yields moment equalities free of the fixed effect \parencite{Bonhomme2012,honoreWeidnerMomentConditions2025}, which need not exhaust the observable implications \parencite{dobronyiGuKim2021}.

Our two binary choice applications depart from familiar benchmark restrictions in different ways.
Fully parametric specifications obtain a likelihood by specifying the joint error distribution \parencite{cham1980,Chamberlain2010}, whereas semiparametric approaches exploit restrictions such as equality of error distributions across periods \parencite{Manski1987,ChernValHahnNewey2013}.
Under sequential exogeneity, feedback from past outcomes to future covariates makes the conditioning information in the restrictions on the errors evolve over time.
This has long been a difficult case in nonlinear panel models \parencite{arellanoPanelDataModels2001,chamberlainFeedbackPanelData2022,bonhommeDanoGraham2025}.
With finitely supported covariates, we leave the conditional marginal error distributions unspecified and characterize the identified set of the coefficient and the ASF.%
\footnote{Existing identification results for the coefficient rely on a parametric error distribution \parencite{arellanoBinaryChoicePanel2003,piginiConditionalInferenceBinary2022,bonhommeDanoGraham2023} or a special regressor \parencite{honoreSemiparametricBinaryChoice2002}, while \textcite{ChernValHahnNewey2013} derive bounds for the ASF.}
Our second specification fixes the conditional marginal distribution of each error term given the fixed effect and covariates but leaves serial dependence unrestricted.
It therefore lies between the two benchmarks: the marginal distributions are specified, but the joint error distribution, and hence the likelihood, is not.%
\footnote{Specified serial dependence can instead be incorporated into a likelihood \parencite{heckmanStatisticalModelsDiscrete1981,Hyslop1999}.}
In both specifications, the LIMR retains the fixed effect among the inputs, so restrictions stated conditional on it and counterfactual parameters defined through its distribution can be imposed directly.

Random set, optimal transport, and entropic latent-variable methods also provide general identification procedures for partially identified models.
Section~\ref{sec:related_literature} compares those approaches with the adversarial approach.

\subsection{Contributions}
\label{sec:contributions}

Our first contribution is an adversarial characterization of the identified set.
Fix $\theta$ and let $\Gamma_\theta$ denote the set of admissible input measures.
Let $\Lgen$ denote the output operator, so that $\Lgen\gamma$ is the output measure, the output distribution generated by $\gamma\in\Gamma_\theta$.
Let $\mu^*$ denote the true output measure.
Under the conditions of Theorem~\ref{thm:main_sharpness}, we construct the ADF below
\begin{equation}
T(\theta)
\equiv
\sup_{\substack{\phi\ \mathrm{measurable}\\ 0\leq\phi\leq 1}} \;
\inf_{\gamma\in\Gamma_\theta} \;
\Big(
\EE{\mu^*}{\phi}
-
\EE{\Lgen\gamma}{\phi}
\Big),
\label{eq:discrepancy_function}
\end{equation}
whose zeros characterize the identified set. 
The outer problem searches over test functions $\phi$ and the inner problem chooses the admissible input measure that minimizes the resulting discrepancy.%
\footnote{\textcite{Kaji2023} use a generator--discriminator minimax construction for simulation-based estimation of parametric structural models.
The adversary here instead ranges over test functions that separate $\mu^*$ from the set of output measures generated by the model.}
Thus $T(\theta)>0$ whenever some $\phi$ separates $\mu^*$ from every output measure at $\theta$.
By the separation argument in Section~\ref{sec:identification}, $T(\theta)$ equals the total variation (TV) distance from $\mu^*$ to the TV closure of the set of output measures.\footnote{Hence $T(\theta)=0$ means that the model can approximate $\mu^*$ arbitrarily closely in TV. When the output measure set is TV closed, exact and approximate compatibility coincide. The set of output measures is convex by the LIMR. See Section~\ref{sec:discussion_discrepancy} for discussion.}
The max--min order in \eqref{eq:discrepancy_function} is intentional: by the LIMR, the inner problem is linear in $\gamma$, while the outer problem is linear in $\phi$. 
This linear structure yields the LP formulations used for computation and inference.

Our second contribution is computation: when the output space is finite, an auxiliary finite linear program computes $T(\theta)$ exactly under two checkable conditions and, under one of them together with a row bound, encloses it.
The auxiliary program keeps finitely many input values as rows and finitely many moment restrictions as columns.
Omitting input values weakly raises the auxiliary value, whereas omitting moment restrictions weakly lowers it, so the auxiliary value need not bound $T(\theta)$ in either direction.
\emph{Column certification} verifies that omitted moment restrictions leave the auxiliary value unchanged, making it an upper bound on $T(\theta)$.
\emph{Row certification} verifies that no omitted input value lowers the auxiliary value.
When both hold, the auxiliary value equals $T(\theta)$ (Theorem~\ref{thm:exact_computation}).
Column certification requires no optimization in any of our examples, whereas row certification requires a global optimization over the input space.
Deciding whether $\theta$ belongs to the identified set needs less than equality.
A \emph{row bound}, a certified bound on how much omitted input values can lower the auxiliary value, gives a lower bound, and with column certification it encloses $T(\theta)$ (Theorem~\ref{thm:enclosure}).
An upper bound of zero settles inclusion, and a positive lower bound settles exclusion.
If neither bound decides, column-and-row generation adds as a new row an input value that lowers the auxiliary value, constructs a column-certified set of moment restrictions for the enlarged support, and solves the program again.
We make no general convergence claim, and a decision does not need one: every column-certified iterate carries a valid enclosure.

Our third contribution is inference.
Replacing $\mu^*$ by the empirical output measure gives $T_n(\theta)$, the sample analog of \eqref{eq:discrepancy_function}.
With finite output support, the pointwise limiting distribution of $\sqrt{n}(T_n(\theta)-T(\theta))$ depends on the \emph{contact set}, the set of test functions attaining the population supremum.
We use a penalized bootstrap whose diverging penalty localizes the bootstrap supremum to the contact set.
Uniform validity does not require uniform estimation of the contact set.
For $\theta$ in the identified set, a finite-sample inequality reduces uniform size control to a uniform empirical-process approximation on the finite output space.
Test inversion then gives confidence sets that cover each point of the identified set with asymptotic probability at least $1-\alpha$, uniformly over $\mu^*$ and over $\theta$ in its identified set, including output measures with zero-probability cells.%
\footnote{The tests add a fixed $\varepsilon>0$ to the bootstrap critical value on the $\sqrt{n}T_n(\theta)$ scale, as in Corollary~\ref{C:inf_finite}.}
For alternatives satisfying $T(\theta)\geq\Delta/\sqrt n$, the worst-case asymptotic acceptance probability converges to zero uniformly over the sampling distribution and the parameter value as $\Delta\to\infty$.
The population ADF, its sample analog, and the penalized bootstrap statistic are computed with the same LP structure.

Our fourth contribution concerns binary choice panels with fixed effects.
Under the baseline specification, the error terms have a known distribution, are independent across periods, and are independent of the fixed effect and covariates.
In two-period designs based on six error distributions, logit reproduces the classical point identification result \parencite{cham1980,Chamberlain2010}, whereas the other five distributions produce interval-identified coefficients.
If serial dependence is unrestricted while each error term retains its known conditional distribution given the fixed effect and covariates, the coefficient set widens but its sign remains identified in all six designs.
If instead the fixed effect and the errors may be arbitrarily dependent while the errors remain serially independent conditional on the covariates, the identified sets contain zero in these designs.
For the probit design, we also compute joint identified sets for the coefficient and the ASF.
Additional periods contract these sets, but more slowly under either relaxation than under the baseline.
When the fixed effect and errors may be arbitrarily dependent, five periods are required in this design to identify the sign of the coefficient.
With two periods and interval-censored covariates, the coefficient sign remains identified even when each covariate is observed only through three equal-width bins; see Additional Appendix~\ref{app:interval_results}.
In the probit simulations, rejection rates at evaluated coefficients in the identified set do not exceed the nominal level, while rejection rates outside the set increase with sample size.

We also characterize the identified set under sequential exogeneity.
For binary choice panels with finite covariate support, this gives, to our knowledge, the first identified-set characterization for the coefficient without a parametric error distribution or a special regressor.
In our baseline numerical design, adding a third period produces an upper bound on the coefficient, although its sign remains unidentified with either two or three periods.
Feedback from past outcomes to future covariates may violate conditional stationarity \parencite{Manski1987}, which conditions on the complete covariate history, while leaving sequential exogeneity correctly specified.

Finally, the framework applies beyond panel models.
We study the two-player complete-information entry game of \textcite{Tamer2003} under the specification of \textcite{beresteanuSharpIdentificationRegions2011} (BMM).
Without covariates and with bivariate standard normal error terms, augmenting the inputs with an unrestricted equilibrium selector yields a LIMR whose ADF has BMM's identified set as its zero set.
Thus an incomplete model can be handled by completing it through the input measure rather than first deriving a random set characterization.
The two criteria agree on membership, and finite linear programs approximate the adversarial criterion to arbitrary accuracy; see Supplemental Appendix~\ref{app:entry_game}.

\subsection{Related literature}
\label{sec:related_literature}

\paragraph{Identification.}
A useful distinction among identification methods for latent-variable models is whether they eliminate the latent variables before characterizing compatibility or retain a latent object in the characterization.

Random set methods characterize compatibility in observable space \parencite{BeresteanuMolinari2008,beresteanuSharpIdentificationRegions2011}.
BMM represent the model-implied moment set as an Aumann expectation and characterize membership by support functions.
For their entry-game specification, which fixes the distribution of the error terms and leaves equilibrium selection unrestricted, the completed LIMR of Example~\ref{ex:entry_game_main} generates the same set of outcome probability vectors as the Aumann expectation of their equilibrium-outcome random set; see Supplemental Appendix~\ref{app:entry_game}.
This equivalence relies on unrestricted equilibrium selection.
Restrictions on equilibrium selection change the set of admissible selections, and BMM note that the resulting moment set need not remain convex \parencite[p.~1788]{beresteanuSharpIdentificationRegions2011}.
The input measure fixes the marginal distribution of the error terms.
A restriction that is linear in the conditional distribution of the equilibrium selector given the error terms is therefore linear in the input measure, and the LIMR can impose it.
The same applies to restrictions on errors conditional on other latent variables whenever those restrictions are linear in the input measure.
Additionally, the approach here does not require an explicit characterization of the set of model-implied output distributions or its observable implications.

\textcite{chesherRosenZhang2026} project out the fixed effect and then apply random set methods to characterize compatibility in the resulting incomplete model.
Projection permits unrestricted dependence between the fixed effect and the error terms.
The LIMR can allow unrestricted dependence while retaining the fixed effect among the inputs, as in Example~\ref{ex:binary_parametric}.
Restrictions or counterfactual parameters involving the fixed effect or its distribution can then be imposed directly.
After projection, they must instead admit an equivalent representation in terms of the projected model.
For discrete-outcome specifications, the projected characterization can require containment inequalities indexed by a core-determining collection of sets.
With finite observable support, the LIMR evaluates compatibility through the linear programming procedure of Section~\ref{sec:computation} without constructing such a collection.

A second class of methods retains a latent object in the compatibility problem.
Optimal transport methods use couplings of observed and latent variables, with the latent marginal specified or restricted through finitely many moments \parencite{ekelandOptimalTransportationFalsifiability2010,galichonSetIdentificationModels2011}.
\textcite{schennachEntropicLatentVariable2014} profiles out the latent distribution by entropic tilting when the model is represented by finitely many moment restrictions, and treats countably many restrictions through an increasing sequence of finite systems.
\textcite{li2026} uses a support-function criterion to characterize the moment closure of the identified set and treats structural and counterfactual parameters jointly.
The LIMR also accommodates models in which the distribution of the fixed effect is unrestricted and an uncountable family of restrictions is imposed conditional on it.
Those restrictions enter directly through the joint input measure, while the closure relevant for our identification result is taken in the space of output probability measures rather than in moment space.

Other methods retain type probabilities, cell probabilities, a subdistribution, or a latent function \parencite{balkePearlBoundsTreatmentEffects1997,laffersIdentificationModelsDiscrete2019,MogstadSantosTorgovitsky2018,torgovitskyPartialIdentificationExtending2019,Tebaldi2023,guCounterfactualIdentificationLatent2025}.
In nonlinear panels, related methods optimize over latent heterogeneity or exploit model-specific reductions to characterize coefficients or average effects \parencite{honoreBoundsParametersPanel2006,ChernValHahnNewey2013,daveziesIdentificationEstimationAverage2022,bonhommeDanoGraham2023}.
\textcite{guCounterfactualIdentificationLatent2025} instead construct an exact finite representation for discrete-outcome models with indices linear in latent variables and characterize additional restrictions that can be imposed in that representation.
These reductions depend on the outcome equation, restrictions on the latent variables, and the parameter of interest, and need not preserve arbitrary restrictions indexed by continuously distributed latent heterogeneity.

\textcite{christensenCounterfactualSensitivityRobustness2023} study sensitivity of counterfactuals to a parametric specification of the distribution of the unobservables, optimizing over a $\varphi$-divergence neighborhood subject to finitely many moment restrictions.
Their nonparametric case also yields a membership criterion whose dual contains an optimization over latent values; finite-radius neighborhoods replace this optimization by a convex expectation (their Section~2.5 and Remark~2.7).
With finite observable support, we instead bound the analogous optimization over input values and obtain certified lower and upper bounds on $T(\theta)$ (Section~\ref{sec:computation}).

\paragraph{Computation.}
Finite computation in latent-variable models typically follows either from an exact finite representation \parencite{balkePearlBoundsTreatmentEffects1997,honoreBoundsParametersPanel2006,KitamuraStoye2018,laffersIdentificationModelsDiscrete2019,Tebaldi2023,guCounterfactualIdentificationLatent2025} or from a finite approximation to an infinite-dimensional latent object \parencite{ChernValHahnNewey2013,MogstadSantosTorgovitsky2018,bonhommeDanoGraham2023}.
For their finite linear programming formulation, \textcite{honoreBoundsParametersPanel2006} impose finite support on the fixed effect, while \textcite{ChernValHahnNewey2013}, \textcite{bonhommeDanoGraham2023}, and \textcite{botosaru2024} implement continuous-support problems on finite grids without formal bounds on the discretization error.
\textcite{pakelBoundsAverageEffects2026} instead obtain tractable outer bounds.
Our finite input support and retained moment restrictions define an auxiliary problem.
The target remains the unrestricted $T(\theta)$, and the certificates in Section~\ref{sec:computation} translate the finite problem into valid bounds on that target.
The resulting procedure is related to cutting-plane, row-exchange, and column-generation methods for semi-infinite and large-scale linear programs \parencite{Kelley1960,HettichKortanek1993,Luebbecke2011,Muter2013}.
In econometrics, \textcite{SmeuldersCramaSpieksma2021} generate rational types for the finite program of \textcite{KitamuraStoye2018}.

\paragraph{Inference.}
We use the per-point coverage criterion of \textcite{im2004} and \textcite{stoya2009}, uniformly over the sampling distribution and the parameter value under test, rather than simultaneous coverage of the identified set \parencite{Chernozhukovetal2007,romano2010inference}.
The nonregularity is analogous to that in moment inequality models, where the binding restrictions can change along sequences of data-generating processes and motivate moment selection \parencite{RomanoShaikh2008,andrews2010}.
There, moment selection operates on observable moment inequalities.
Here the corresponding object is the contact set of test functions generated by the adversarial criterion, and uniform validity does not require uniform estimation of that set.
\textcite{GalichonHenry2009} and \textcite{loh2024} use related minimax statistics, while the pointwise analysis uses results for directionally differentiable functionals \parencite{fang2018,HongLi2018}.

\subsection{Notation}

We use $\one\{\cdot\}$ for the indicator function and $\overset{d}{=}$ for equality in distribution.
For each measurable space $\mathcal S$, let $\mathcal B(\mathcal S)$ denote its $\sigma$-algebra.
Product spaces carry the product $\sigma$-algebra.
We write $\Pa(\mathcal S)$ for the set of probability measures on $\mathcal B(\mathcal S)$ and $\delta_s$ for the Dirac measure at $s\in\mathcal S$, defined by $\delta_s(B)=\one\{s\in B\}$.
A probability kernel from $\mathcal S$ to $\mathcal S'$ is a map $K(\cdot\mid\cdot)$ such that $K(\cdot\mid s)$ is a probability measure on $\mathcal S'$ for each $s\in\mathcal S$ and $K(B\mid\cdot)$ is measurable for each $B\in\mathcal B(\mathcal S')$.
For a measurable map $f\colon \mathcal S\to\mathcal S'$ and $\mu\in\Pa(\mathcal S)$, the pushforward measure $f_*\mu\in\Pa(\mathcal S')$ is defined by $(f_*\mu)(B)=\mu(f^{-1}(B))$ for $B\in\mathcal B(\mathcal S')$.
For an index set $\mathcal I$, $\R^{(\mathcal I)}$ denotes the set of vectors in $\R^{\mathcal I}$ with finite support, and $\R^{(\mathcal I)}_+$ its nonnegative cone.
For a finite signed measure $\nu$, write $\inner{\phi,\nu}\equiv\int\phi\,\d\nu$ and, when $\nu=\mu$ is a probability measure, $\EE{\mu}{\phi}\equiv\inner{\phi,\mu}$.
For probability measures, we use the convention
$\norm{\mu-\mu'}_{\mathrm{TV}}\equiv\sup_{B\in\mathcal B(\mathcal S)}|\mu(B)-\mu'(B)|$
for total variation distance.
For a vector $X$, $X'$ denotes its transpose.

\section{Model}
\label{sec:model}

This section defines what it means for an econometric model to have a \emph{linear-in-measures representation} (LIMR).
Section~\ref{sec:examples} gives six examples of models with a LIMR.

Let $X \in \mathcal{X}$ and $U \in \mathcal{U}$ denote observable and unobservable inputs, respectively, and collect them into the input $W = (X, U) \in \mathcal{W} = \mathcal{X} \times \mathcal{U}$.
Let $Y \in \mathcal{Y}$ denote the outcome, and define the output as $Z = (Y, X) \in \mathcal{Z} = \mathcal{Y} \times \mathcal{X}$, so that the observable $X$ serves as both an input and an output.

In our leading binary choice panel models, $X = (X_1,\dots,X_T)$ collects time-varying covariates, $Y = (Y_1,\dots,Y_T)$ collects binary outcomes, and $U$ contains a time-invariant fixed effect and, in some specifications, idiosyncratic error terms.
In other applications, $U$ may contain random coefficients, true values of partially observed covariates, equilibrium selectors, potential outcomes, or other structural primitives.

\begin{asm}
\label{asm:measurable}
The supports $\mathcal Y$, $\mathcal X$, and $\mathcal U$ are measurable spaces.
\end{asm}
A probability measure $\gamma\in\Pa(\mathcal W)$ on the inputs is called an \emph{input measure}.
A parameter $\theta\in\Theta$ collects structural parameters, such as regression coefficients, and counterfactual parameters defined by restrictions linear in $\gamma$, such as ASFs.
\begin{asm}
\label{asm:moment_conditions}
For each $\theta\in\Theta$, the set of admissible input measures $\Gamma_\theta \subseteq \Pa(\mathcal W)$ is characterized by a system of moment restrictions:\footnote{These restrictions implicitly require integrability of the moment functions with respect to $\gamma$.}
$$
\Gamma_\theta = \left\{ \gamma \in \Pa(\mathcal W) : \begin{aligned}
\EE{\gamma}{g_{1,j}(W;\theta)} &= 0,   &&\quad\text{for all } j \in \mathcal J \\
\EE{\gamma}{g_{2,k}(W;\theta)} &\le 0, &&\quad\text{for all } k \in \mathcal K
\end{aligned} \right\},
$$
where $\mathcal J$ and $\mathcal K$ are arbitrary index sets, and $\{g_{1,j}(\cdot;\theta)\}_{j \in \mathcal J}$ and $\{g_{2,k}(\cdot;\theta)\}_{k \in \mathcal K}$ are families of real-valued measurable functions on $\mathcal W$.
\end{asm}
In the examples, the equalities represent sequential exogeneity, parametric restrictions, and random assignment.
The inequalities may represent shape restrictions.
The equalities can also define the counterfactual parameters, as in Example~\ref{ex:binary_parametric}.
The index sets $\mathcal J$ and $\mathcal K$ may be infinite, which matters in nonlinear panel models because restrictions conditional on fixed effects typically generate a continuum of moment equalities.

Not every model assumption is a linear restriction in $\gamma$ in its natural parameterization.
For example, independence with unrestricted marginals imposes a nonlinear factorization of the joint measure.
Such restrictions may admit an equivalent linear representation after reparameterization or augmentation, as in Example~\ref{ex:binary_sequential}.

\begin{asm}
\label{asm:linear_operator}
For each $\theta\in\Theta$, the model specifies an \emph{output kernel}, a probability kernel $K_\theta(\cdot\mid\cdot)$ from $\mathcal W$ to $\mathcal Z$ that does not depend on $\gamma$ and that preserves the observable input:
\begin{equation}
\label{eq:probability_kernel}
(\pi_{\mathcal X})_*K_\theta(\cdot\mid x,u)=\delta_x
\quad\text{for all } (x,u)\in\mathcal W,
\end{equation}
where $\pi_{\mathcal X}\colon\mathcal Z\to\mathcal X$ denotes the coordinate projection.
\end{asm}
Integrating the output kernel with respect to $\gamma$ defines the \emph{output measure} $\mu_{\theta,\gamma} \equiv \Lgen\gamma$ by
\begin{equation}
\label{eq:kernel_operator}
(\Lgen\gamma)(B) \equiv \int_{\mathcal W} K_\theta(B\mid w)\,\d\gamma(w)
\quad\text{for all } B\in\mathcal B(\mathcal Z).
\end{equation}
We call $\Lgen$ the \emph{output operator}.
Because $K_\theta$ is fixed as $\gamma$ varies, the output operator is linear in the input measure.%
\footnote{Formally, because $\Pa(\mathcal W)$ is not a vector space, the map $\Lgen\colon\Pa(\mathcal W)\to\Pa(\mathcal Z)$ is affine. Throughout, linearity refers to its extension to finite signed measures.}

Assumption~\ref{asm:linear_operator} does not require the outcome equation to be linear, separable, or parametric.
For example, consider a model with outcome equation $Y=h_\theta(W)$, where $h_\theta\colon\mathcal W\to\mathcal Y$ is measurable for each $\theta$, and define $\psi_\theta\colon\mathcal W\to\mathcal Z$ by $\psi_\theta(w)\equiv\bigl(h_\theta(w),x\bigr)$ for $w=(x,u)$.
The output kernel is
$K_\theta(B\mid w)=\one\{\psi_\theta(w)\in B\}$,
and the output measure $\Lgen\gamma=(\psi_\theta)_*\gamma$ is the pushforward of $\gamma$ through the map $w\mapsto\psi_\theta(w)$.

The representation is flexible in how the model assumptions are allocated among $\mathcal W$, $\Gamma_\theta$, and $K_\theta$.
Section~\ref{sec:examples} illustrates this flexibility: the parametric panel examples integrate out the error terms through $K_\theta$; the semiparametric panel example retains them in $W$ and restricts their distribution through $\Gamma_\theta$; the interval-censoring example builds the censoring restriction directly into $\mathcal W$; and the entry game augments $W$ with an equilibrium selector.

We call the family $\{(\Gamma_\theta,\Lgen):\theta\in\Theta\}$ a LIMR if, for each $\theta$, the set $\Gamma_\theta\subseteq\Pa(\mathcal W)$ is characterized by Assumption~\ref{asm:moment_conditions}, and the output operator $\Lgen\colon\Pa(\mathcal W)\to\Pa(\mathcal Z)$ is characterized by Assumption~\ref{asm:linear_operator}.
The name reflects that the restrictions characterizing $\Gamma_\theta$ and the output operator $\Lgen$ are linear in the input measure.
We say that an econometric model has a LIMR if there exists such a family for which, at every $\theta\in\Theta$, $\{\Lgen\gamma:\gamma\in\Gamma_\theta\}$ coincides with the set of probability measures of $Z$ that the model can generate at $\theta$.

\begin{asm}
\label{asm:domination}
For each $\theta\in\Theta$, there exists a $\sigma$-finite measure $\lambda_\theta$ on $\mathcal Z$ such that, for all $\gamma \in \Gamma_\theta$, $\Lgen\gamma$ is absolutely continuous with respect to $\lambda_\theta$.
\end{asm}

Assumption~\ref{asm:domination} ensures that every output measure has a Radon--Nikodym density in the common space $L^1(\lambda_\theta)$.
It supplies the dominating measure used in the separation argument of Section~\ref{sec:identification}.
This assumption is satisfied when $\mathcal Z$ is finite by taking $\lambda_\theta$ to be counting measure.
With continuous $X$ or $Y$, Assumption~\ref{asm:domination} holds, for example, when $\Gamma_\theta$ fixes the marginal of $X$ at some probability measure and, for every $(x,u)\in\mathcal W$, $K_\theta(\cdot\times\mathcal X\mid x,u)$ is dominated by a probability kernel from $\mathcal X$ to $\mathcal Y$ that does not depend on $u$.

\section{Examples}
\label{sec:examples}

This section illustrates the scope of the LIMR by constructing one for each of six econometric models.
Each example specifies the input $W$, the admissible set $\Gamma_\theta$, and the output kernel $K_\theta$.
Throughout, the outcomes and the observable input take finitely many values, so Assumption~\ref{asm:domination} holds with $\lambda_\theta$ equal to counting measure on $\mathcal Z$, and we can write the output kernel as a probability mass function.
The panel examples feature linear indices, additive fixed effects, and binary outcomes, although the framework does not require them.

\begin{example}[Parametric binary choice with fixed effects]
\label{ex:binary_parametric}
For $t=1,\ldots,T$, let
\begin{equation}
\label{eq:panel_binary_outcome}
Y_t=\one\{X_t'\beta+A-V_t\ge0\},
\end{equation}
where $A$ is a fixed effect whose distribution given the covariate path $X=(X_1,\ldots,X_T)$ is unrestricted, and $V=(V_1,\ldots,V_T)$ collects the error terms.
Let $H$ be a known continuous CDF.
We consider a baseline in which the errors are i.i.d.\ given $(A,X)$, and two relaxations:
\begin{equation}
\label{eq:parametric_three_specifications}
\begin{aligned}
V\mid(A,X)&\sim H^{\otimes T} && \text{(baseline)},\\
V_t\mid(A,X)&\sim H,\quad t=1,\ldots,T && \text{(serial dependence)},\\
V\mid X&\sim H^{\otimes T} && \text{(fixed effect--error dependence)}.
\end{aligned}
\end{equation}
The first relaxation allows serial dependence while retaining independence of each error from $(A,X)$.
The second allows dependence between the fixed effect and the errors while retaining serial independence given $X$.
Section~\ref{subsec:numerical_example1} compares their identifying information.
For the baseline specification, take $W=(X,A)$, $Z=(Y,X)$, and $\Gamma_\beta=\Pa(\mathcal X\times\R)$, leaving the input measure unrestricted.
Integrating out $V$ gives the output kernel
$$
K_\beta(\{(y,x)\}\mid x,a)=\prod_{t=1}^T H(x_t'\beta+a)^{y_t}\bigl(1-H(x_t'\beta+a)\bigr)^{1-y_t}.
$$
We incorporate the ASF $\tau_{\mathrm{ASF}}\equiv\Pr(\bar x'\beta+A-V_1\ge0)$ at a counterfactual covariate value $\bar x$ by setting $\theta=(\beta,\tau_{\mathrm{ASF}})$ and imposing the linear restriction $\EE{\gamma}{H(\bar x'\beta+A)-\tau_{\mathrm{ASF}}}=0$, which refines $\Gamma_\beta$ to $\Gamma_\theta$.
Similarly, an average treatment effect (ATE) between two counterfactual covariate values $\bar x^0$ and $\bar x^1$ can be incorporated with $\theta=(\beta,\tau_{\mathrm{ATE}})$ and $\EE{\gamma}{H\bigl((\bar x^1)'\beta+A\bigr)-H\bigl((\bar x^0)'\beta+A\bigr)-\tau_{\mathrm{ATE}}}=0$.
For the two relaxations, take $W=(X,A,V)$, use the kernel induced by~\eqref{eq:panel_binary_outcome}, and impose the corresponding distributional restriction through the admissible set.
Supplemental Appendix~\ref{app:computation_example1} develops the LIMR for each specification.
\end{example}

\begin{example}[Interval-censored covariates]
\label{ex:binary_interval}
For $t=1,\ldots,T$, let
$$
Y_t=\one\{(X_t^\star)'\beta+A-V_t\ge0\},
\qquad
X_{L,t}\le X_t^\star\le X_{U,t}.
$$
The covariate path $X^\star$ and the fixed effect $A$ are latent, and $X^\star$ is observed only through the endpoints $X=(X_L,X_U)$, the observable input, so the output is $Z=(Y,X_L,X_U)$.
Impose $V\mid(X_L,X_U,X^\star,A)\sim H^{\otimes T}$ for known $H$.
With $\theta=\beta$, take $W=(X_L,X_U,X^\star,A)$ on the input space $\mathcal W\equiv\{(x_L,x_U,x^\star,a):x_L\le x^\star\le x_U\}$ and set $\Gamma_\theta=\Pa(\mathcal W)$.
The support restriction carries the interval censoring, and the input measure is otherwise unrestricted.
Integrating out $V$ gives the output kernel: $Y$ has the probability mass function of Example~\ref{ex:binary_parametric} with $x^\star$ in place of $x$, and the endpoints in the output equal those in the input.
Additional Appendix~\ref{app:interval_results} reports identified sets for $\beta$ in a numerical example.
\end{example}

\begin{example}[Panel binary choice with sequential exogeneity]
\label{ex:binary_sequential}
Consider outcome equation~\eqref{eq:panel_binary_outcome}.
Denote by $X^t\equiv(X_1,\ldots,X_t)$ the covariate history through period $t$ and impose sequential exogeneity:
\begin{equation}
\label{eq:seq_exog}
V_t\mid (A,X^t)\overset{d}{=}V_1\mid (A,X_1),
\qquad
t=2,\ldots,T.
\end{equation}
This is the predetermined version of time homogeneity in \textcite[Assumption~3]{ChernValHahnNewey2013}: given the fixed effect, the error has the same distribution in every period, and that distribution may depend on the covariate history only through $X_1$.
Future covariates do not enter the conditioning set, so covariates may respond to past outcomes \parencite{ChernValHahnNewey2013,chamberlainFeedbackPanelData2022,bonhommeDanoGraham2023,Bonhomme2025BackToFeedback}.
The restriction is nonlinear in the distribution of $(X,V,A)$, but becomes linear once we replace $A$ by $\nu\equiv(\Pr(X=x\mid A))_{x\in\mathcal X}$, the distribution of the covariate path given the fixed effect, and $V_t$ by the composite error $\tilde V_t\equiv A-V_t$.
Write $\nu_t(x^t)\equiv\sum_{\tilde x\in\mathcal X:\,\tilde x^t=x^t}\nu(\tilde x)$ for the probability of the history $x^t$ under $\nu$.
Lemma~\ref{lem:nu_lifting} in Supplemental Appendix~\ref{app:semiparametric_binary_choice} shows that, for finite $\mathcal X$, the model has a LIMR with input $W=(X,\tilde V,\nu)$, output $Z=(Y,X)$, output operator the pushforward of $\gamma$ through the map $w\mapsto\bigl((\one\{x_t'\beta+\tilde v_t\ge0\})_{t=1}^T,x\bigr)$, and admissible set $\Gamma_\beta$ defined by
\begin{align}
&\EE{\gamma}{r(\nu)\bigl(\one\{X=x\}-\nu(x)\bigr)}=0,
\quad\text{for all }x,\ \text{and all bounded measurable }r,
\label{eq:nu_consistency}\\
&\EE{\gamma}{r(\nu)\bigl[f(\tilde V_t)\one\{X^t=x^t\}\nu_1(x_1)-f(\tilde V_1)\one\{X_1=x_1\}\nu_t(x^t)\bigr]}=0,
\label{eq:seq_lifted}\\
&\quad\text{for all }t=2,\ldots,T,\ \text{all }x^t,\ \text{and all bounded measurable }r,f.
\notag
\end{align}
Restriction~\eqref{eq:seq_lifted} is sequential exogeneity with $\nu$ in place of $A$, multiplied through by the history probabilities so that it is linear in $\gamma$.
We incorporate the ASF at a counterfactual covariate value $\bar x$ by setting $\theta=(\beta,\tau_{\mathrm{ASF}})$ and imposing $\EE{\gamma}{\one\{\bar x'\beta+\tilde V_1\ge0\}-\tau_{\mathrm{ASF}}}=0$ in addition to the restrictions defining $\Gamma_\beta$.
Section~\ref{subsec:numerical_semiparametric} compares the identified sets with those under conditional stationarity.
\end{example}

\begin{example}[Dynamic discrete choice and dynamic panels]
\label{ex:dynamic_state_dependence}
\textcite{honoreBoundsParametersPanel2006} study the dynamic panel model with outcome equation
$$
Y_t=\one\{X_t'\beta+\rho Y_{t-1}+A+V_t\ge0\}, \qquad t=1,\ldots,T,
$$
where the errors $V_t$ are i.i.d.\ standard normal variables independent of $(X,A,Y_0)$, and the binary initial condition $Y_0$ is latent.
With $\theta=(\beta,\rho)$, take $W=(X,A,Y_0)$, $Z=((Y_1,\ldots,Y_T),X)$, and $\Gamma_\theta=\Pa(\mathcal X\times\R\times\{0,1\})$, leaving the distribution of $(A,Y_0)$ given $X$ unrestricted.
Integrating out $V$ gives the output kernel
$$
K_\theta(\{(y,x)\}\mid x,a,y_0)=\prod_{t=1}^T \Phi(x_t'\beta+\rho y_{t-1}+a)^{y_t}\bigl(1-\Phi(x_t'\beta+\rho y_{t-1}+a)\bigr)^{1-y_t},
$$
where $\Phi$ denotes the standard normal CDF.

The same construction extends immediately to the habit persistence specification of \textcite{heckmanStatisticalModelsDiscrete1981}, in which the lagged latent index replaces $Y_{t-1}$ in the outcome equation: $Y_t=\one\{Y_t^*\ge0\}$ with $Y_t^*=X_t'\beta+\rho Y_{t-1}^*+A+V_t$.
The latent input component becomes $(A,Y_0^*)$, and the output kernel is given by multivariate normal orthant probabilities.
\end{example}

\begin{example}[Simultaneous entry game with multiple equilibria]
\label{ex:entry_game_main}
Firms $j=1,2$ simultaneously choose entry actions $y_j\in\{0,1\}$.
Firm $j$ receives payoff $y_j(\delta_jy_{-j}+V_j)$, where $\theta=(\delta_1,\delta_2)\in(-\infty,0)^2$ collects the spillover parameters.
Both firms observe the error terms $V=(V_1,V_2)$ before play, and $V$ follows a known joint CDF $H$ with continuous marginals.
The game has three Nash equilibria when $0<V_j<-\delta_j$ for both firms, so the model is incomplete \parencite{Tamer2003}.
We complete the model with an equilibrium selector $S\in\{1,2,3\}$ that indexes those equilibria.
Take $W=(V,S)$ and $Z=Y=(Y_1,Y_2)$, the realized action profile, since $\mathcal X$ is a singleton.
The output kernel $K_\theta(\cdot\mid v,s)$ is the outcome distribution of the equilibrium that $s$ selects at $v$.
The admissible set $\Gamma_\theta$ imposes the moment equalities $\EE{\gamma}{\one\{V\le c\}-H(c)}=0$ for all $c\in\R^2$, ensuring that $V$ has CDF $H$.
It leaves the conditional distribution of $S$ given $V$ unrestricted, thereby spanning every equilibrium selection mechanism \parencite{BerryTamer2007}.
Proposition~\ref{prop:entry_equivalence} in Supplemental Appendix~\ref{app:entry_game} shows, for bivariate standard normal errors, that the identified set coincides with the sharp identification region of \textcite{beresteanuSharpIdentificationRegions2011}.
\end{example}

\begin{example}[Randomized trial with imperfect compliance]
\label{ex:imperfect_compliance}
Let $R\in\{0,1\}$ be a randomized assignment, $D\in\{0,1\}$ the treatment received, and $Y\in\{0,1\}$ the outcome, allowing $D\ne R$ \parencite{balkePearlBoundsTreatmentEffects1997}.
Write $D(r)$ for treatment under assignment $r$ and $Y(d)$ for the outcome under treatment $d$, and collect the latent response type as $U=(Y(0),Y(1),D(0),D(1))\in\{0,1\}^4$.
Take the observable input to be $X=R$ and the model outcome to be $(Y,D)$, so $W=(R,U)$ and $Z=(Y,D,R)$.
The output kernel maps $(r,u)$, with $u=(y(0),y(1),d(0),d(1))$, to $(y(d(r)),d(r),r)$.
This mapping incorporates the exclusion restriction.
Let $\theta=\tau$ denote the ATE.
Random assignment and the target enter $\Gamma_\theta$ through the linear equalities
$$
\begin{aligned}
\EE{\gamma}{\bigl(\one\{R=1\}-p\bigr)h(U)}&=0
&&\text{for all bounded measurable }h,\\
\EE{\gamma}{Y(1)-Y(0)-\tau}&=0,
\end{aligned}
$$
where $p\equiv\Pr(R=1)$ is known by design.
Because $W$ takes finitely many values, the identified set for $\tau$ is the interval between the sharp bounds of \textcite{balkePearlBoundsTreatmentEffects1997}, which solve the same linear program over the distribution of $U$.
Restrictions beyond those of \textcite{balkePearlBoundsTreatmentEffects1997} enter directly: monotonicity $D(1)\ge D(0)$ is the moment equality $\EE{\gamma}{\one\{D(0)=1,\ D(1)=0\}}=0$, which rules out defiers.
\end{example}

\section{Identification}
\label{sec:identification}

We first define the identified set, then construct the ADF $T(\theta)$ via a separation argument.
Our main result shows that $\theta$ belongs to the identified set if and only if $T(\theta)=0$.

The econometrician observes the true probability measure $\mu^*$ of the output $Z$.
For each $\theta\in\Theta$, let $\mathcal M_\theta\subseteq\Pa(\mathcal Z)$ denote the set of probability measures of $Z$ that the model can generate at $\theta$.
If the model has a LIMR, this set can be written as
\begin{equation}
    \label{eq:model_probabilities}
    \mathcal M_\theta = \Lgen\Gamma_\theta = \{\mu_{\theta,\gamma}:\gamma\in\Gamma_\theta\}.
\end{equation}
Denote the TV closure of $\mathcal M_\theta$ by
\begin{equation}
    \label{eq:closure_definition}
    \Mtheta
    \equiv
    \Bigl\{
     m \in \Pa(\mathcal{Z}) :
     \forall \epsilon > 0,\ \exists\, \gamma \in \Gamma_\theta,
     \ \norm{\mu_{\theta,\gamma} - m}_{\mathrm{TV}} < \epsilon \Bigr\}.
\end{equation}
We define the identified set as
\begin{equation}
    \label{eq:identified_set}
    \Thetaw \equiv
    \{\theta \in \Theta :
    \mu^*\in\Mtheta\}.
\end{equation}

The closure adds probability measures that can be approximated arbitrarily closely in TV by measures in $\mathcal M_\theta$.%
\footnote{Identification using closures also arises in \textcite{schennachEntropicLatentVariable2014} and \textcite{li2026}.
The former takes the closure of attainable expected moment values, while the latter characterizes the moment closure of the identified set.
Both closures are taken in finite-dimensional moment space rather than in the space of output measures used here.}
When $\mathcal M_\theta$ is already TV closed, $\mathcal M_\theta=\Mtheta$ and the closure is redundant.
This is the case in Example~\ref{ex:binary_sequential}, as shown in Proposition~\ref{prop:se_finite_support} in Supplemental Appendix~\ref{app:semiparametric_binary_choice}.
By contrast, $\mathcal M_\theta\subsetneq\Mtheta$ in the logit specification of Example~\ref{ex:binary_parametric}.
For a fixed covariate path $x$, the point mass at $(Y,X)=((1,\ldots,1),x)$ is the TV limit of output measures as the fixed effect diverges to $+\infty$, but no admissible input measure generates it (Remark~\ref{rem:closure_role}).

Under i.i.d.\ sampling, the boundary probability measures added by the TV closure $\Mtheta$ cannot be statistically distinguished from the measures in $\mathcal M_\theta$: any test whose size is controlled uniformly over $\mathcal M_\theta$ has power no greater than size against a measure in $\Mtheta\setminus\mathcal M_\theta$; see Remark~\ref{rem:tv_indistinguishability}.

We construct the ADF from the LIMR and Assumptions~\ref{asm:measurable} and~\ref{asm:domination}.
The LIMR makes $\mathcal M_\theta$ convex, so its TV closure $\Mtheta$ is closed and convex, while Assumption~\ref{asm:domination} represents every element of $\Mtheta$ by a density in $L^1(\lambda_\theta)$.
The associated separating hyperplane argument yields $\mu^*\notin\Mtheta$ if and only if there is a bounded measurable test function $\phi$ such that $\EE{\mu^*}{\phi}-\sup_{\mu\in\Mtheta}\EE{\mu}{\phi}>0$.
By $L^1$--$L^\infty$ duality, every continuous linear functional on $L^1(\lambda_\theta)$ is integration against a bounded measurable function, so no larger class of test functions is needed.
Because positive rescaling and translation preserve the sign of the separating gap, we normalize this full class as
\begin{equation}
    \Phi(\mathcal{Z})
    \equiv
    \{\phi\colon\mathcal{Z}\to[0,1]\text{ measurable}\}.
    \label{eq:phi_class}
\end{equation}

The ADF is then the value of this separation problem.
The outer supremum maximizes the separating gap over $\Phi(\mathcal Z)$, while the inner infimum searches over the admissible input measures $\gamma\in\Gamma_\theta$:%
\footnote{We use the convention $\inf\varnothing=+\infty$, so parameter values $\theta$ with $\Gamma_\theta=\varnothing$ are automatically excluded.}
\begin{align}
    T(\theta)
    &\equiv
    \sup_{\phi \in \Phi(\mathcal{Z})} \inf_{\gamma \in \Gamma_\theta}
        \Bigl( \EE{\mu^*}{\phi} - \EE{\Lgen\gamma}{\phi} \Bigr)
    \label{eq:discrepancy_function_input}
    \\
    &=
    \sup_{\phi \in \Phi(\mathcal{Z})} \inf_{\mu \in \Mtheta}
        \Bigl( \EE{\mu^*}{\phi} - \EE{\mu}{\phi} \Bigr).
    \label{eq:discrepancy_function_model}
\end{align}
For each $\phi\in\Phi(\mathcal Z)$, the gap $\mu\mapsto\EE{\mu^*}{\phi}-\EE{\mu}{\phi}$ is continuous in TV.
Its infimum over $\mathcal M_\theta=\Lgen\Gamma_\theta$ therefore equals its infimum over the closure $\Mtheta$, which gives the second equality.

Let $\Thetam$ denote the zero set of $T(\theta)$:
$$\Thetam \equiv \{\theta \in \Theta : T(\theta)=0\}.$$
\begin{thm}
    \label{thm:main_sharpness}
    Under Assumptions~\ref{asm:measurable}--\ref{asm:domination}, the identified set is
    \begin{equation}
        \Thetaw = \Thetam.
    \end{equation}
\end{thm}
\begin{proof}
See Appendix~\ref{app:proofs}.
\end{proof}

Identification can be viewed as a separation problem in the space of probability measures on $\mathcal Z$.
By \eqref{eq:discrepancy_function_model}, $T(\theta)=0$ exactly when no bounded test function separates $\mu^*$ from $\Mtheta$; Theorem~\ref{thm:main_sharpness} shows that this happens exactly when $\theta\in\Thetaw$.\footnote{The value $T(\theta)$ is nonnegative because $\phi\equiv0$ is feasible.}
The separation argument uses only two properties of the set of output measures: convexity and common domination.
The LIMR supplies the first and Assumption~\ref{asm:domination} the second, and the LIMR also makes the model's response to each test function a linear problem in the input measure.
Evaluating $T(\theta)$ therefore does not require characterizing $\Mtheta$.
For computation and inference, we instead use \eqref{eq:discrepancy_function_input}, which works with the input measure $\gamma$.

\subsection{Discussion}
\label{sec:discussion_discrepancy}

\begin{remark}[Role of the closure]
\label{rem:closure_role}
Consider the baseline specification of
Example~\ref{ex:binary_parametric}, where $\theta=\beta$, $\Gamma_\beta=\Pa(\mathcal X\times\R)$, and $H$ has full support on $\R$, as in the logit case.
Fix a covariate path $x$ and let $\mu^*$ assign probability one to $(Y,X)=((1,\ldots,1),x)$.
For the admissible input measures $\gamma_a\equiv\delta_{(x,a)}$, the output measures satisfy $\mu_{\beta,\gamma_a}(\{((1,\ldots,1),x)\})=\prod_{t=1}^T H(x_t'\beta+a)\to1$ as $a\to+\infty$, so $\mu_{\beta,\gamma_a}\to\mu^*$ in TV and $\mu^*\in\overline{\mathcal M}_\beta$ for every $\beta$.
No admissible input measure generates $\mu^*$, because $\EE{\gamma}{\prod_{t=1}^T H(X_t'\beta+A)}<1$ for every $\gamma\in\Gamma_\beta$.
The limiting output measure here is the one that would be generated if $A$ were allowed to take the boundary value $+\infty$.
More generally, allowing the boundary values $\pm\infty$ would let $\Pr(Y=(1,\ldots,1)\mid X=x)$ range over $[0,1]$ rather than $(0,1)$.
This matters when $\mu^*$ is itself such a boundary point, as in this example: the closure yields $\Thetaw=\Theta$ rather than the empty set.
In this example, that inclusiveness is desirable, because a population with no outcome variation cannot distinguish values of $\beta$.

For this baseline specification, the closed set of output measures can be obtained by adjoining $\pm\infty$ to the support of $A$ and extending the output kernel continuously to those values.
We instead take the closure in output space, which avoids imposing model-specific conditions guaranteeing a closed set of output measures.
\end{remark}

\begin{remark}[Statistical indistinguishability of TV-closure points]
\label{rem:tv_indistinguishability}
Fix $\theta\in\Theta$ and $n\in\mathbb N$, and suppose that
$Z_1,\ldots,Z_n$ are i.i.d.
Let $\psi_n\colon\mathcal Z^n\to[0,1]$ be any possibly randomized test satisfying
\begin{equation}
    \sup_{\nu\in\mathcal M_\theta}
    \EE{\nu^{\otimes n}}{\psi_n}
    \leq \alpha .
    \label{eq:uniform_size_exact_model}
\end{equation}
Then, for every $\mu\in\Mtheta$,
\begin{equation}
    \EE{\mu^{\otimes n}}{\psi_n}
    \leq \alpha .
    \label{eq:closure_indistinguishability}
\end{equation}
Hence, if $\mu\in\Mtheta\setminus\mathcal M_\theta$, no test whose size is at most $\alpha$ uniformly over $\mathcal M_\theta$ rejects $\mu$ with probability above $\alpha$.\footnote{To see~\eqref{eq:closure_indistinguishability}, fix $\mu\in\Mtheta$ and choose
$\nu_j\in\mathcal M_\theta$ with
$\norm{\nu_j-\mu}_{\mathrm{TV}}\to0$. Then, for fixed $n$,
$$
    \norm{\nu_j^{\otimes n}-\mu^{\otimes n}}_{\mathrm{TV}}
    \leq
    n\norm{\nu_j-\mu}_{\mathrm{TV}}
    \longrightarrow0.
$$
Since $0\leq\psi_n\leq1$,
$\EE{\nu_j^{\otimes n}}{\psi_n}
\to\EE{\mu^{\otimes n}}{\psi_n}$, and
\eqref{eq:closure_indistinguishability} follows from
\eqref{eq:uniform_size_exact_model}.
The result applies only to $\mu\in\Mtheta$ and gives no impossibility
statement for $\mu\notin\Mtheta$.}
This TV-closure indistinguishability argument is standard in the literature on impossible inference; see, e.g., \textcite{BERTANHA2020247}.
Relatedly, \textcite{BaiPonomarevSantosShaikhTabordMeehanTorgovitsky2026} characterize the TV closure of the null hypothesis in their analysis of partially identified linear systems.
When $\mathcal Z$ is finite, $\Mtheta$ is simply the Euclidean closure of the feasible set of cell-probability vectors.
\end{remark}

\begin{remark}[Point-to-set integral probability metric]
\label{rem:ipm}
The ADF can be interpreted as an extension of an integral probability metric (IPM) from two fixed probability measures to a point-to-set discrepancy.
More generally, for any class $\mathcal F$ for which the expectations are well defined, define the corresponding point-to-set discrepancy:
$$
T_{\mathcal F}(\theta)
\equiv
\sup_{\phi\in\mathcal F}\inf_{\gamma\in\Gamma_\theta}
\Bigl(
\EE{\mu^*}{\phi}-\EE{\Lgen\gamma}{\phi}
\Bigr).
$$
For two fixed probability measures, different choices of $\mathcal F$ give familiar discrepancies, e.g., the unit ball of a reproducing kernel Hilbert space gives maximum mean discrepancy \parencite{Gretton2012}, while the $1$-Lipschitz functions give the $1$-Wasserstein distance.
The corresponding max--min criterion equals the infimum of these pairwise discrepancies over the model set when the relevant minimax conditions hold.
We establish this equality for TV under our maintained assumptions.

The separation argument above, which uses Assumptions~\ref{asm:measurable}--\ref{asm:domination}, \emph{selects} $\mathcal F = \Phi(\mathcal Z)$.
More regular classes require additional topological, metric, or kernel structure on $\mathcal Z$ that the construction of $T(\theta)$ does not use.
If $\mathcal F\subseteq\Phi(\mathcal Z)$ and $0\in\mathcal F$, then
$0\leq T_{\mathcal F}(\theta)\leq T(\theta)$, so restricting the class of test functions can only enlarge the zero set.

For $\Phi(\mathcal Z)$ in \eqref{eq:phi_class}, convexity of $\Mtheta$ and the minimax argument in the proof of Theorem~\ref{thm:main_sharpness} give:
$$
T(\theta)
=
\inf_{\mu\in\Mtheta}
\sup_{\phi\in\Phi(\mathcal Z)}
\Bigl(
\EE{\mu^*}{\phi}-\EE{\mu}{\phi}
\Bigr)
=
\inf_{\mu\in\Mtheta}
\norm{\mu^*-\mu}_{\mathrm{TV}}.
$$
Thus, $T(\theta)$ is the TV distance from $\mu^*$ to $\Mtheta$.
We do not use this representation for computation or inference.

The IPM representation shows that $T(\theta)$ is a supremum of linear expectation differences over the class of test functions.
Combined with the LIMR, this makes the objective bilinear in $(\phi, \gamma)$, leading to the finite LPs of Section~\ref{sec:computation}.
Since TV is an $f$-divergence, $T(\theta)$ is the distance from $\mu^*$ to $\Mtheta$ in an $f$-divergence.
Inference and model specification using $f$-divergences within the point-to-set geometry appear in, e.g., \textcite{KitamuraStutzer1997,kaidoMolinari2025}.
While $T(\theta)$ retains the connection to the divergence-based literature, it exploits the linear structure of both the IPM representation and the LIMR.
\end{remark}

\section{Computation}
\label{sec:computation}

Section~\ref{sec:identification} characterizes $\Thetaw$ through the zeros of $T(\theta)$, but evaluating $T(\theta)$ requires optimization over all test functions and all admissible input measures.
We construct an auxiliary finite linear program by restricting the input measure to finite support and retaining finitely many moment restrictions.
Because the program omits input values and moment restrictions, its value need not equal $T(\theta)$.
The program is auxiliary because $T(\theta)$ remains the computational target.

Section~\ref{subsec:computation_main} gives two conditions under which the auxiliary program is exact.
Column certification requires an optimal input measure of the program to satisfy the full family of moment restrictions, and row certification requires an optimal solution to satisfy its constraint at every input value, not only at the finitely many imposed.
Together they give $T_{\mathrm{LP}}(\theta)=T(\theta)$ (Theorem~\ref{thm:exact_computation}).
Section~\ref{subsec:computation_exchange} shows that a verdict on $\theta$ needs less: column certification alone bounds $T(\theta)$ from above, a bound on the row residual bounds it from below (Theorem~\ref{thm:enclosure}), and a column-and-row generation algorithm iteratively enlarges the finite sets, tightening the two bounds.

We maintain Assumptions~\ref{asm:measurable}--\ref{asm:linear_operator} and replace Assumption~\ref{asm:domination} by:
{\renewcommand{\theasm}{4\ensuremath{'}}%
\def\theHasm{4prime}%
\begin{asm}
\label{asm:finite_Z}
The observable space $\mathcal Z$ is finite and every subset of $\mathcal Z$ is measurable.
\end{asm}}\setcounter{asm}{4}%
The assumption restricts only the observable space and imposes no restriction or topology on $\mathcal W$.
It makes $\Phi(\mathcal Z)=[0,1]^{\mathcal Z}$ a convex polytope, so the outer supremum in $T(\theta)$ is finite dimensional.
The remaining sources of infinite dimensionality are the input measure $\gamma$ and the possibly infinite family of moment restrictions defining $\Gamma_\theta$.
When both the outcomes and the covariates are discrete, point identification typically fails in nonlinear panels with fixed effects \parencite{Chamberlain2010}, so this is the setting in which a characterization of the identified set matters most.

\subsection{Main result}
\label{subsec:computation_main}

Fix $\theta\in\Theta$.
For a nonempty finite set $\mathcal W'\subseteq\mathcal W$ and finite sets
$\mathcal J'\subseteq\mathcal J$ and $\mathcal K'\subseteq\mathcal K$, define the \emph{restricted discrepancy}
\begin{align}
\TLD{\theta}{\mathcal W'}{\mathcal J',\mathcal K'}
&\equiv
\sup_{\phi\in[0,1]^{\mathcal Z}}
\inf_{\gamma_p\in\Gamma_\theta(\mathcal W',\mathcal J',\mathcal K')}
\Bigl(
\EE{\mu^*}{\phi}-\EE{\Lgen\gamma_p}{\phi}
\Bigr),
\label{eq:T_restricted}
\end{align}
where $\Gamma_\theta(\mathcal W',\mathcal J',\mathcal K')$ consists of finite-support input measures
$\gamma_p\equiv\sum_{w\in\mathcal W'}p_w\delta_w$, with probability weights
$p=(p_w)_{w\in\mathcal W'}$, satisfying the retained moment restrictions:
\begin{equation}
\label{eq:restricted_gamma}
\Gamma_\theta(\mathcal W',\mathcal J',\mathcal K')
\equiv
\left\{
\gamma_p:
\begin{aligned}
p&\in\R^{(\mathcal W')}_+, &&\textstyle\sum_{w\in\mathcal W'}p_w=1,\\
\EE{\gamma_p}{g_{1,j}(W;\theta)}&=0, &&\quad\text{for all } j\in\mathcal J',\\
\EE{\gamma_p}{g_{2,k}(W;\theta)}&\le0, &&\quad\text{for all } k\in\mathcal K'
\end{aligned}
\right\}.
\end{equation}
The definitions in~\eqref{eq:T_restricted} and~\eqref{eq:restricted_gamma} also apply to the full index sets $\mathcal J$ and $\mathcal K$.
Because $\mathcal W'$, $\mathcal J'$, and $\mathcal K'$ are finite, the inner infimum in~\eqref{eq:T_restricted} is a finite-dimensional linear program.
Dualizing it gives the auxiliary finite linear program
\begin{align}
\TLD{\theta}{\mathcal W'}{\mathcal J',\mathcal K'}
&=
\sup_{\substack{\phi\in[0,1]^{\mathcal Z},\ \zeta\in\R\\
u\in\R^{(\mathcal J')},\ v\in\R^{(\mathcal K')}_{+}}}
\left\{\EE{\mu^*}{\phi}-\zeta\right\}
\label{eq:TLD_restricted}
\\
\quad\text{subject to}\quad
q_{\phi,\theta}(w)
&\le
\zeta
+
\sum_{j\in\mathcal J'}u_j g_{1,j}(w;\theta) +
\sum_{k\in\mathcal K'}v_k g_{2,k}(w;\theta)
\quad\text{for all }w\in\mathcal W',
\nonumber
\end{align}
where the \emph{input-space payoff}
$
q_{\phi,\theta}(w)
\equiv
\EE{\Lgen\delta_w}{\phi}
=
\sum_{z\in\mathcal Z}\phi(z)\,K_\theta(\{z\}\mid w)
$
is the mean of $\phi$ implied by the model at input $w$.
For fixed $\mathcal W'$, $\mathcal J'$, and $\mathcal K'$, write
$$
T_{\mathrm{LP}}(\theta)\equiv\TLD{\theta}{\mathcal W'}{\mathcal J',\mathcal K'}.
$$

Restricting the input measure to $\mathcal W'$ weakly raises $T_{\mathrm{LP}}(\theta)$, whereas omitting moment restrictions weakly lowers it, so $T_{\mathrm{LP}}(\theta)$ need not bound $T(\theta)$ in either direction.
In~\eqref{eq:TLD_restricted}, each input value $w\in\mathcal W'$ contributes a constraint, or row, and each retained moment restriction contributes a multiplier, or column.

Column certification controls the error from omitting moment restrictions.
\begin{defi}[Column certification]
\label{def:column_certified}
The pair $(\mathcal J',\mathcal K')$ is \emph{column-certified for $\mathcal W'$} if $\Gamma_\theta(\mathcal W',\mathcal J,\mathcal K)$ is nonempty and
\begin{equation}
\label{eq:column_certificate}
\TLD{\theta}{\mathcal W'}{\mathcal J',\mathcal K'}
=
\TLD{\theta}{\mathcal W'}{\mathcal J,\mathcal K}.
\end{equation}
\end{defi}
Under column certification, imposing the omitted moment restrictions in $\mathcal J\setminus\mathcal J'$ and $\mathcal K\setminus\mathcal K'$ leaves the value at $\mathcal W'$ unchanged.
Equivalently, $(\mathcal J',\mathcal K')$ is column-certified for $\mathcal W'$ when some optimal $p^*$ of the inner program in~\eqref{eq:T_restricted} gives an input measure $\gamma_{p^*}$ that satisfies the full family of moment restrictions.
Because the admissible input measures supported on $\mathcal W'$ are a subset of $\Gamma_\theta$, column certification implies $T(\theta)\le T_{\mathrm{LP}}(\theta)$, leaving only the error from the omitted input values.
In Examples~\ref{ex:binary_parametric} and~\ref{ex:binary_sequential}, column certification can be established directly and requires no additional optimization (Section~\ref{subsec:computation_exchange}).

Write the slack of~\eqref{eq:TLD_restricted} at $w\in\mathcal W$ as
\begin{equation}
\label{eq:slack}
s(w;\phi,\zeta,u,v)
\equiv
\zeta
+
\sum_{j\in\mathcal J'}u_j g_{1,j}(w;\theta)
+
\sum_{k\in\mathcal K'}v_k g_{2,k}(w;\theta)
-
q_{\phi,\theta}(w),
\end{equation}
so that the constraint in~\eqref{eq:TLD_restricted} is
$s(w;\phi,\zeta,u,v)\ge0$ for all $w\in\mathcal W'$.
For an optimizer $(\phi^*,\zeta^*,u^*,v^*)$, define the \emph{row residual}
\begin{equation}
\label{eq:row_oracle_residual}
r^*
\equiv
\inf_{w\in\mathcal W}
s(w;\phi^*,\zeta^*,u^*,v^*) \leq 0.
\end{equation}
If $r^*=0$, the constraint holds at every input value in $\mathcal W$, not only on $\mathcal W'$.

\begin{defi}[Row certification]
\label{def:row_certified}
The support $\mathcal W'$ is \emph{row-certified for $(\mathcal J',\mathcal K')$} if the row residual~\eqref{eq:row_oracle_residual} equals zero at some optimizer of~\eqref{eq:TLD_restricted}.
\end{defi}

\begin{thm}[Exact computation]
\label{thm:exact_computation}
Under Assumptions~\ref{asm:measurable}--\ref{asm:linear_operator} and~\ref{asm:finite_Z}, if $(\mathcal J',\mathcal K')$ is column-certified for $\mathcal W'$ and $\mathcal W'$ is row-certified for $(\mathcal J',\mathcal K')$, then
\begin{equation}
\label{eq:exact_computation}
T_{\mathrm{LP}}(\theta)=T(\theta).
\end{equation}
\end{thm}
\begin{proof}
See Appendix~\ref{app:proofs}, where Theorem~\ref{thm:exact_computation} is obtained from Theorem~\ref{thm:enclosure} below by setting the row bound to zero.
\end{proof}

\subsection{Column-and-row generation}
\label{subsec:computation_exchange}

Under the conditions of Theorem~\ref{thm:exact_computation}, the auxiliary finite linear program gives a verdict on $\theta$: $\theta\in\Thetaw$ if $T_{\mathrm{LP}}(\theta)=0$, and $\theta\notin\Thetaw$ if $T_{\mathrm{LP}}(\theta)>0$.
Of the two certifications the theorem requires, row certification is the demanding one, and it is more than a verdict needs.
Inclusion needs only column certification, and for exclusion row certification can be replaced by a bound on the row residual.

The column side is straightforward: a finite, column-certified $(\mathcal J',\mathcal K')$ can be constructed directly, with no optimization, in all of our examples for any finite $\mathcal W'$ that supports an admissible input measure.
Supplemental Appendices~\ref{app:computation_example1} and~\ref{app:se_exact_pricing} do so for Examples~\ref{ex:binary_parametric} and~\ref{ex:binary_sequential}, and Lemma~\ref{lem:finite_column_reduction} in Supplemental Appendix~\ref{app:algorithm} gives a general construction whenever the moment inequalities are finite in number, allowing arbitrary moment equalities.
Column certification implies $T(\theta)\le T_{\mathrm{LP}}(\theta)$, so $T_{\mathrm{LP}}(\theta)=0$ settles inclusion.
A positive value does not settle exclusion, because the bound $0\le T(\theta)\le T_{\mathrm{LP}}(\theta)$ allows both zero and positive values.
Without a bound on the error from omitted input values, a column-certified discretization gives only this upper bound.

The row side is harder, because computing $r^*$ requires minimizing the slack over all of $\mathcal W$.
A verdict does not need $r^*$ itself, only a number $r_{\mathrm{low}}$ that is guaranteed to lie below it, which we call a \emph{row bound}.
In Example~\ref{ex:binary_parametric}, a row bound requires, for each covariate path, an upper bound on the model-implied mean of $\phi^*$ over the scalar fixed effect, which we obtain by bounding this mean on intervals of the real line and splitting the intervals until the bound is tight enough.
Under fixed effect--error dependence, no such calculation is needed (Additional Appendix~\ref{app:parametric_implementation}).%
\footnote{For Example~\ref{ex:binary_sequential}, Supplemental Appendix~\ref{app:se_exact_pricing} obtains the row bound by a bilinear optimization over admissible conditional distributions.}
A row bound gives the lower bound $T_{\mathrm{LP}}(\theta)+r_{\mathrm{low}}$ on $T(\theta)$, so a positive value settles exclusion.

Column certification and a row bound together enclose $T(\theta)$:
\begin{thm}[Enclosure]
\label{thm:enclosure}
Under Assumptions~\ref{asm:measurable}--\ref{asm:linear_operator} and~\ref{asm:finite_Z}, let $(\mathcal J',\mathcal K')$ be column-certified for $\mathcal W'$.
Then~\eqref{eq:TLD_restricted} has an optimizer, and for the row residual $r^*$ of any optimizer, every $r_{\mathrm{low}}\le r^*$ satisfies
\begin{equation}
\label{eq:certified_enclosure}
\max\{0,T_{\mathrm{LP}}(\theta)+r_{\mathrm{low}}\}
\;\le\;
T(\theta)
\;\le\;
T_{\mathrm{LP}}(\theta).
\end{equation}
\end{thm}
\begin{proof}
See Appendix~\ref{app:proofs}.
\end{proof}

\begin{remark}[Row residual over input measures]
\label{rem:row_measures}
The proof of Theorem~\ref{thm:enclosure} uses the row inequality only through its integral against measures in $\Gamma_\theta$.
The theorem therefore holds with the row residual~\eqref{eq:row_oracle_residual} replaced by $\inf_{\gamma\in\Gamma_\theta}\EE{\gamma}{s(W;\phi^*,\zeta^*,u^*,v^*)}$, which is weakly larger.
This refinement is used under sequential exogeneity, in Example~\ref{ex:binary_sequential} and Supplemental Appendix~\ref{app:se_exact_pricing}.
\end{remark}

The enclosure at a given $(\mathcal W',\mathcal J',\mathcal K')$ may be too wide to decide, and we then enlarge the sets by column-and-row generation, which combines two standard devices in linear programming \parencite{Luebbecke2011,Muter2013}.
An input value with negative slack becomes a new row, and we then construct a column-certified $(\mathcal J',\mathcal K')$ for the enlarged support.
We re-solve the auxiliary program and repeat, as stated precisely in Algorithm~\ref{alg:master} of Supplemental Appendix~\ref{app:algorithm}.
A fixed discretization can serve as the initial support, which the algorithm then certifies and, if needed, refines.
Because the bounds of Theorem~\ref{thm:enclosure} are valid at every step, the algorithm may stop as soon as they determine whether $\theta\in\Thetaw$, or once they are narrow enough.

We impose no topology on $\mathcal W$, and hence no compactness or continuity, so the usual arguments for the convergence of a refined grid do not apply, and we make no general convergence claim.
Because the bounds are valid at every step, a verdict does not require convergence.
When the algorithm stops before the bounds decide, the parameter value is reported as undecided.
The implementation treats a value as zero only to linear programming accuracy, with no separate membership threshold (Additional Appendix~\ref{app:parametric_implementation}).
Model-specific compactness and continuity, or a finite reduction, may yield a general result through standard cutting-plane and semi-infinite programming results \parencite{Kelley1960,HettichKortanek1993}.

Figure~\ref{fig:enclosure_iterations} illustrates the algorithm in the baseline design of Section~\ref{subsec:numerical_example1}: the fixed effect probit model of Example~\ref{ex:binary_parametric} with $T=2$, a binary covariate correlated with the fixed effect, error terms independent of both, and $\beta_0=1$.
For two coefficient values only $1.6\times10^{-4}$ apart, on opposite sides of the boundary of the identified set, the algorithm reaches a verdict in eight iterations.
The figure shows the two enclosures, starting from a single fixed effect value at each covariate path, $\mathcal W'=\{(x,0):x\in\mathcal X\}$, with $\mathcal J'=\mathcal J=\mathcal K'=\mathcal K=\varnothing$.
Each iteration requires one solution of the auxiliary program and four one-dimensional interval branch and bound searches, one per covariate path, to bound the row residual (Additional Appendix~\ref{app:parametric_implementation}), and at the deciding iteration the program has $22$ rows, spanning nine distinct fixed effect values.
Across one hundred equally spaced coefficient values on $[0.01,2]$, the median number of iterations to a verdict is three, against eight for the two values above.
Deciding one value takes about ten milliseconds on a single core of an Intel Core i7-11370H.

\begin{figure}
    \centering
    \includegraphics[width=0.6\linewidth]{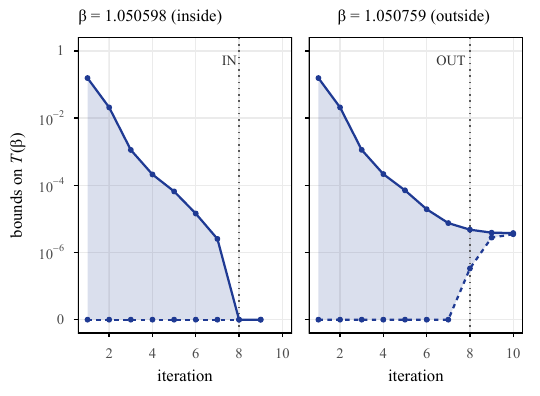}
    \caption{Bounds on $T(\beta)$ by iteration, at $\beta=1.050598$ (inside the identified set) and $\beta=1.050759$ (outside).
    Solid: $T_{\mathrm{LP}}(\beta)$. Dashed: $\max\{0,T_{\mathrm{LP}}(\beta)+r_{\mathrm{low}}\}$. Dotted: first decision.
    Logarithmic scale, with zero on the floor of each panel.}
    \label{fig:enclosure_iterations}
\end{figure}

\section{Inference}
\label{sec:inference}

We develop inference procedures for the finite output setting of Assumption~\ref{asm:finite_Z}.
We construct tests of candidate parameter values and, by test inversion, confidence sets that cover each point of the identified set $\Thetaw$ with asymptotic probability at least $1-\alpha$, uniformly over the data-generating process and the point under test.
The procedure uses the sample analog $T_n(\theta)$ of the ADF and a bootstrap approximation to the distribution of a penalized statistic that dominates $\sqrt nT_n(\theta)$ under the null.
Both may be obtained with the machinery of Section~\ref{sec:computation}: once the retained problem is column-certified and a certified lower bound on $T_n(\theta)$ has been computed, each bootstrap replicate requires only one finite linear program to obtain a conservative upper bound on the bootstrap statistic (Remark~\ref{rem:crit_val_LP}).
Supplemental Appendix~\ref{app:inference_computation} provides implementation details. 

Under Assumption~\ref{asm:finite_Z}, the space of test functions $\Phi(\mathcal{Z}) = [0,1]^{\mathcal Z}$ is finite-dimensional, as in moment inequality problems, but our approach differs from standard procedures in two ways.
First, the same model-specific row bound calculation used for population identification evaluates the empirical distribution of output variables against the full set of input-space moment restrictions. This is possible because the constraint functions are always fully known to the researcher.
Second, our bootstrap critical values use a geometric contact-set penalty that focuses power on the separating hyperplanes most informative about whether $\theta$ belongs to the identified set.  

\subsection{Setup and notation}

We denote the sample expectation by $\EE{n}{\cdot} \equiv n^{-1}\sumn$.
We observe $n$ i.i.d.\ realizations $\{Z_i\}_{i=1}^n$ from $\mu^*$.
Define the sample analog of the ADF as
\begin{align}
	\label{E:Tn_definition}
	T_n(\theta) \equiv \sup_{\phi \in \Phi(\mathcal{Z})}  \inf_{\mu \in \Mtheta} \left(\EE{n}{\phi} - \EE{\mu}{\phi}\right),
\end{align}
which replaces the population expectation in the ADF by its empirical counterpart.
Throughout this section, we subscript the population ADF of Section~\ref{sec:identification} by $\mu^*$ to make its dependence explicit:
\begin{align*}
	T_{\mu^*}(\theta) \equiv \sup_{\phi \in \Phi(\mathcal{Z})} \inf_{\mu \in \Mtheta} (\EE{\mu^*}{\phi} - \EE{\mu}{\phi}).
\end{align*}
	
We define the penalty function
\begin{align}
	\label{E:penalty}
	\eta_{\theta, \mu^*}(\phi) &\equiv \inf_{\mu \in \Mtheta} \left(\EE{\mu^*}{\phi} - \EE{\mu}{\phi}\right),
\end{align}
with empirical analog $\eta_{\theta, n}(\phi) \equiv \inf_{\mu \in \Mtheta} (\EE{n}{\phi} - \EE{\mu}{\phi})$.

This penalty function measures how well a test function $\phi$ discriminates between $\mu^*$ and the model set $\Mtheta$.
We likewise write $\Thetaw(\mu^*)$ for the identified set when the data are generated from $\mu^*$.
Note that $\eta_{\theta, \mu^*}(\phi) \le 0$ for all $\phi$ when $\theta \in \Thetaw(\mu^*)$.
By definition, $T_{\mu^*}(\theta) = \sup_{\phi \in \Phi(\mathcal{Z})} \eta_{\theta, \mu^*}(\phi)$ and $T_{n}(\theta) = \sup_{\phi \in \Phi(\mathcal{Z})} \eta_{\theta, n}(\phi)$.

Let $\mathbf{P} \subset \Pa(\mathcal{Z})$ denote a class of distributions over which we seek uniform results.
For $\mu^* \in \mathbf{P}$ and $\theta \in \Theta$, define the contact set
\begin{align}
\label{E:contact_set}
K_{\mu^*}(\theta) \equiv \left\{\phi \in \Phi(\mathcal{Z}): \eta_{\theta, \mu^*}(\phi) = T_{\mu^*}(\theta)\right\}.
\end{align}
The contact set $K_{\mu^*}(\theta)$ consists of the test functions attaining the population adversarial supremum.
When $\mu^*\in\Mtheta$, its elements define supporting directions at $\mu^*$; if $\mu^*$ lies in the relative interior of a set $\Mtheta$ that is full-dimensional relative to the probability simplex, these reduce to the constant functions.
When $\mu^*\notin\Mtheta$, the contact set consists of test functions attaining the maximal separating gap.

We say that a set of random variables $\{A_{n,\mu^*}: \mu^* \in \mathbf{P}\}$ is $o_p(1)$ uniformly in $\mu^*$ if, for all $c > 0$,
$
	\limsup_{n \to \infty} \sup_{\mu^* \in \mathbf{P}} \mu^*(\{|A_{n,\mu^*}| > c\}) = 0.
$
If instead
$
	\lim_{C \to \infty} \limsup_{n \to \infty} \sup_{\mu^* \in \mathbf{P}} \mu^*(\{|A_{n,\mu^*}| > C\}) = 0,
$
then the set is $O_p(1)$ uniformly in $\mu^*$.
Incorporating an additional supremum over $\theta$ extends these definitions to uniformity over $\theta$.

\subsection{Asymptotic distribution}

\begin{asm}[Random sampling]
	\label{A:finite_inference}
	The observations $\{Z_i\}_{i=1}^n$ are i.i.d.\ from $\mu^* \in \mathbf{P}$, where $\mathbf{P}$ is a collection of probability measures on the finite space $\mathcal{Z}$.
\end{asm}

In particular, $\mathbf{P}$ may be taken as all of $\Pa(\mathcal{Z})$: because $\mathcal{Z}$ is finite and test functions in $\Phi(\mathcal{Z})$ are uniformly bounded, no further regularity on $\mathbf{P}$ is required for the uniform results that follow.%
\footnote{The argument does not divide by cell probabilities or require them to be bounded away from zero. Cells with zero probability produce degenerate coordinates of the empirical process, which is allowed.}
Define the centered empirical process $\mathbb{G}_{n,\mu^*}(\phi) = \sqrt{n}(\EE{n}{\phi} - \EE{\mu^*}{\phi})$.
Because $\mathcal{Z}$ is finite, the multivariate central limit theorem implies that $\mathbb{G}_{n,\mu^*}$ converges in distribution to a Gaussian process $\mathbb{G}_{\mu^*}$ indexed by $\phi \in \Phi(\mathcal{Z})$.

\begin{remark}[Finite-dimensional stochastic representation]\label{rem:remark_finitedimrep}
If $\mathcal Z=\{z_1,\ldots,z_{|\mathcal Z|}\}$, write $p$ and $\widehat p_n$ for the population and empirical probability vectors and set
$$
V_n\equiv\sqrt n(\widehat p_n-p)\in\mathbb R^{|\mathcal Z|}.
$$
Identifying $\phi$ with $(\phi(z_1),\ldots,\phi(z_{|\mathcal Z|}))'$, we have $\mathbb G_{n,\mu^*}(\phi)=\phi'V_n$.
Thus all stochastic variation in the empirical process is $|\mathcal Z|$-dimensional.
The representation also makes explicit that the parameter space for the observable distribution is the entire probability simplex.
We retain the empirical-process formulation below because the finite output space gives the required distribution-free entropy bound uniformly over this simplex.
\end{remark}

Let $\{Z_i^*\}_{i=1}^n$ denote a bootstrap sample drawn with replacement from the observed data, and define the bootstrap empirical process $\mathbb{G}_{n,\mu^*}^*(\phi) = \sqrt{n}(\mathbb{E}_n^*[\phi] - \EE{n}{\phi})$, where $\mathbb{E}_n^*[\cdot] = n^{-1}\sumn (\cdot)(Z_i^*)$ denotes the empirical bootstrap average.
We reserve $\mathbb E^*$ and $\mathbb P^*$ for conditional bootstrap expectation and probability:
$$
\mathbb P^*(A)\equiv
\mathbb P(A\mid Z_1,\ldots,Z_n),
\qquad
\mathbb E^*[X]\equiv
\mathbb E[X\mid Z_1,\ldots,Z_n].
$$

Let $d_{\mathrm{BL}}$ denote the bounded Lipschitz metric between probability distributions \parencite[\S 1.12]{VW1996}.
We write $\mathbb{G}_n \Rightarrow \mathbb{G}$ for random variables if $d_{\mathrm{BL}}(L_n,L)$ converges to $0$, where $L_n$ is the distribution of $\mathbb{G}_n$ and $L$ is the distribution of $\mathbb{G}$.
For a bootstrap quantity $\mathbb{G}_n^*$ whose distribution $L_n^*$ is itself random for each $n$, we write $\mathbb{G}_n^*  \overset{\mu^*}{\Rightarrow} \mathbb{G}$ if $d_{\mathrm{BL}}(L_n^*,L)$ converges in probability to $0$.

\begin{prop}
	\label{P:inf_finite}
	Let Assumptions~\ref{asm:measurable}, \ref{asm:moment_conditions}, \ref{asm:linear_operator}, \ref{asm:finite_Z}, and~\ref{A:finite_inference} hold with $\Gamma_\theta$ nonempty.
	For any sequence of positive penalty parameters $(\lambda_n)$,
	\begin{align}
		\sqrt{n} T_n(\theta) &\le \sup_{\phi \in \Phi(\mathcal{Z})} \left(\mathbb{G}_{n,\mu^*}(\phi) + \lambda_n \eta_{\theta,n}(\phi)\right) - \lambda_n T_n(\theta)\label{E:inf_bound}
	\end{align}
	for all $n$, $\mu^* \in \mathbf{P}$, and $\theta \in \Thetaw(\mu^*)$.
	If, in addition, $\lambda_n \to \infty$ and
	$\lambda_n=o(\sqrt n)$, then, for each fixed
	$\mu^* \in \mathbf{P}$ and $\theta \in \Theta$ such that $\Gamma_\theta\neq\varnothing$, the contact set $K_{\mu^*}(\theta)$ is nonempty and, as $n\to\infty$,
	\begin{align}
		\sqrt{n}(T_n(\theta) - T_{\mu^*}(\theta)) &\Rightarrow \sup_{\phi \in K_{\mu^*}(\theta)} \mathbb{G}_{\mu^*}(\phi), \label{E:inf_convergence}
	\end{align}
	and the bootstrap analog satisfies
	\begin{align}
		\sup_{\phi \in \Phi(\mathcal{Z})} \left(\mathbb{G}_{n,\mu^*}^*(\phi) + \lambda_n \eta_{\theta,n}(\phi)\right) - \lambda_n T_n(\theta) \overset{\mu^*}{\Rightarrow} \sup_{\phi \in K_{\mu^*}(\theta)} \mathbb{G}_{\mu^*}(\phi). \label{E:inf_bootstrap}
	\end{align}
\end{prop}

\begin{proof}
    See Appendix~\ref{app:proofs}.
\end{proof}

Proposition~\ref{P:inf_finite} bounds $\sqrt{n} T_n(\theta)$ by a penalized supremum that is consistently approximated by the bootstrap.
The empirical penalty term $\lambda_n\eta_{\theta,n}(\phi)$ is the sample counterpart of the population contact-set penalty and down-weights test functions $\phi$ outside $K_{\mu^*}(\theta)$, which do not enter the limit in~\eqref{E:inf_convergence}.
The $\sqrt n$ population drift localizes the sample criterion to the contact set in~\eqref{E:inf_convergence}.
Divergence of $\lambda_n$ produces the corresponding localization of the bootstrap criterion in~\eqref{E:inf_bootstrap}, while $\lambda_n=o(\sqrt n)$ makes estimation error in the empirical penalty asymptotically negligible.
The domination inequality~\eqref{E:inf_bound} is finite-sample and does not require either $\lambda_n\to\infty$ or $\lambda_n=o(\sqrt n)$.
These rate conditions are used only for the pointwise and bootstrap approximations in~\eqref{E:inf_convergence} and \eqref{E:inf_bootstrap}.

The requirement that $\Gamma_\theta$ is nonempty ensures that the penalty function $\eta_{\theta, n}$ and $T_n$ are well defined. 
Nonemptiness of $\Gamma_\theta$ is easy to verify directly in many of our examples, and the feasibility check FEAS in Supplemental Appendix~\ref{app:algorithm} certifies it.  
Parameter values for which FEAS certifies $\Gamma_\theta=\varnothing$ are not in the identified set and can be excluded before inference. 

For every fixed null pair $(\mu^*,\theta)$ with $\theta\in\Thetaw(\mu^*)$, \eqref{E:inf_bound} holds for every $n$ and every positive $\lambda_n$.
Under $\lambda_n\to\infty$ and $\lambda_n=o(\sqrt n)$, the bootstrap statistic in \eqref{E:inf_bootstrap} consistently estimates the pointwise limiting distribution of this upper-bounding statistic. 
This combination of an exact finite-sample null domination and a pointwise-consistent bootstrap approximation parallels the approach in \textcite{HongLi2018} \parencite[see also][]{fang2018}.

\begin{remark}[Bootstrap computation]
\label{rem:crit_val_LP}
Supplemental Appendix~\ref{app:inference_computation} shows that the bootstrap value $\sup_{\phi \in \Phi(\mathcal{Z})} \left(\mathbb{G}_{n,\mu^*}^*(\phi) + \lambda_n \eta_{\theta,n}(\phi)\right)$ used for estimating the asymptotic distribution of $T_n$ has an LP representation which exactly mirrors that of $T_n$, with a signed measure replacing the empirical measure.
This fact makes valid, if conservative, inference especially computationally tractable.
For example, under column certification of $(\mathcal{J}', \mathcal{K}')$ for $\mathcal{W}'$ obtained through equality of the feasible sets, as in Lemma~\ref{lem:finite_column_reduction}, Supplemental Appendix~\ref{app:inference_computation} bounds the left-hand side of~\eqref{E:inf_bootstrap} above by the direct bootstrap analog of $T_{\mathrm{LP}}(\theta)$ minus $\lambda_n$ times a certified lower bound on $T_n(\theta)$.
This upper bound is explicitly the value of a finite linear program and requires no row or column generation.
\end{remark}

\subsection{Confidence sets}

The corollary below gives tests whose inversion yields confidence sets with uniform asymptotic coverage of each point in the identified set.
It also considers power over the alternatives $\Theta_n^\Delta(\mu^*) \equiv \{\theta \in \Theta: T_{\mu^*}(\theta) \ge \Delta / \sqrt{n}\}$ for $\Delta > 0$.
\begin{cor}
	\label{C:inf_finite}
	Let the assumptions of Proposition~\ref{P:inf_finite} hold.
	Let $\varepsilon>0$ and let $\hat{c}_{1-\alpha}(\theta)$ denote the $1-\alpha$ quantile of the bootstrap distribution of $\sup_{\phi \in \Phi(\mathcal{Z})} (\mathbb{G}_{n,\mu^*}^*(\phi) + \lambda_n \eta_{\theta,n}(\phi)) - \lambda_n T_n(\theta)$.
	Then
	\begin{align}
		\liminf_{n \to \infty} \inf_{\substack{\mu^* \in \mathbf{P} \\ \theta \in \Thetaw(\mu^*)}} \PP{\mu^*}{\sqrt{n} T_n(\theta) \le \hat{c}_{1-\alpha}(\theta) + \varepsilon} \ge 1 - \alpha. \label{E:coverage}
	\end{align}
	On the other hand,
	\begin{align}
		\limsup_{\Delta \to \infty} \limsup_{n \to \infty} \sup_{\substack{\mu^* \in \mathbf{P} \\ \theta \in \Theta_{n}^\Delta(\mu^*)}} \PP{\mu^*}{\sqrt{n} T_n(\theta) \le \hat{c}_{1-\alpha}(\theta) + \varepsilon} = 0. \label{E:power}
	\end{align}
\end{cor}
\begin{proof}
See Supplemental Appendix~\ref{sec:app-additional-proofs}.
\end{proof}

The first result of Corollary~\ref{C:inf_finite} is that the critical values $\hat{c}_{1-\alpha}(\theta)+\varepsilon$ control asymptotic size uniformly over points in the identified set.
Define the confidence set
\begin{align}
    \mathrm{CS}_{1-\alpha} \equiv \left\{\theta \in \Theta: \sqrt{n} T_n(\theta) \le \hat{c}_{1-\alpha}(\theta) + \varepsilon\right\}.
\end{align}
Corollary~\ref{C:inf_finite} implies that $\mathrm{CS}_{1-\alpha}$ follows the paradigm of \textcite{im2004}: for every $\theta$ in the identified set, $\mathrm{CS}_{1-\alpha}$ covers $\theta$ with asymptotic probability at least $1-\alpha$ uniformly over both the parameter $\theta \in \Thetaw(\mu^*)$ and the underlying distribution $\mu^* \in \mathbf{P}$. 
This is per-point rather than simultaneous coverage: the result does not require the confidence set to contain the entire identified set with probability at least $1-\alpha$.
Uniform coverage does not require uniform estimation of $K_{\mu^*}(\theta)$.
It follows from the exact null domination in \eqref{E:inf_bound}, the uniform empirical-process approximation implied by the distribution-free entropy bound on the finite output space, and the fact that the penalty-estimation error does not depend on $\theta$.

For alternatives with $T_{\mu^*}(\theta)\geq \Delta/\sqrt n$, the second result states that the worst-case asymptotic acceptance probability converges to zero as $\Delta\to\infty$.
The fixed positive tolerance $\varepsilon$ accommodates possible atoms in the limiting distribution $\sup_{\phi \in K_{\mu^*}(\theta)} \mathbb{G}_{\mu^*}(\phi)$.
For example, when $K_{\mu^*}(\theta)$ degenerates to the constant functions, the limit collapses to a point mass at zero.
Because uniform approximation of distributions does not in general imply uniform approximation of their quantiles at discontinuity points, the positive slack allows the bootstrap approximation to yield uniform coverage without imposing an anti-concentration condition.
Thus, for any fixed $\varepsilon>0$, Corollary~\ref{C:inf_finite} gives asymptotic coverage of at least $1-\alpha$.
Related positive-slack adjustments to bootstrap quantiles are used by~\textcite{andrews2013} and~\textcite{marcoux2024}.\footnote{The tolerance is fixed on the $\sqrt n T_n(\theta)$ scale and is therefore part of the rejection rule, rather than a numerical approximation tolerance.}

\begin{remark}[Specification test]
\label{rem:specification_test}
Corollary~\ref{C:inf_finite} yields a specification test as a
by-product. 
If the model is correctly specified, then $\Thetaw(\mu^{*})$ is nonempty, and any fixed $\theta^{*}\in\Thetaw(\mu^{*})$ is covered with asymptotic probability at least $1-\alpha$.
Because $\{\mathrm{CS}_{1-\alpha}=\varnothing\}\subseteq\{\theta^{*}\notin\mathrm{CS}_{1-\alpha}\}$, it follows that $\limsup_{n\to\infty}\PP{\mu^*}{\mathrm{CS}_{1-\alpha}=\varnothing}\le\alpha$.
Rejecting correct specification when the confidence set is empty is therefore an asymptotically level-$\alpha$ test.
Tests of this form are studied for moment inequality models by
\textcite{bugni2015}, who show that such by-product tests are valid but generally conservative relative to dedicated specification tests. 
\end{remark}
\section{Binary choice models with fixed effects}
\label{sec:numerical}

We report identified sets and inference in the models of Examples~\ref{ex:binary_parametric} and~\ref{ex:binary_sequential}.
Results for Example~\ref{ex:binary_interval} are in Additional Appendix~\ref{app:interval_results}.

\subsection{Parametric binary choice}
\label{subsec:numerical_example1}

We begin from the design in Figure~2 of \textcite{ChernValHahnNewey2013}:
\begin{equation}
\label{eq:cfhn_dgp}
Y_{t} = \one\{\beta_0\, X_{t} + A - V_{t} \ge 0\},
\qquad
V \mid (A,X) \sim H^{\otimes T},
\qquad
X_{t} = \one\{A - \eta_{t} \ge 0\},
\end{equation}
with $\eta_{t}$ and $A$ standard normal, the $\eta_t$ independent across $t$ and independent of $A$, so the binary covariate is correlated with the fixed effect.
We set $\beta_0=1$ and consider six known links $H$, each standardized to have mean zero and variance one: the normal (probit) link, and the logistic (logit), Laplace, Gumbel, uniform, and asymmetric truncated normal links.
The restriction $V\mid(A,X)\sim H^{\otimes T}$ makes the error terms independent across periods and independent of $(A,X)$.
The serial dependence specification of Example~\ref{ex:binary_parametric} keeps each error term's distribution given $(A,X)$ but leaves the dependence across periods unrestricted.
The fixed effect--error dependence specification instead keeps $V\mid X\sim H^{\otimes T}$ but lets the error terms depend on the fixed effect given $X$.
Each relaxation is strictly weaker than the baseline, and the two are not nested in each other.
For each link we generate data from the baseline~\eqref{eq:cfhn_dgp}, so the data satisfy all three specifications, and we hold the resulting distribution $\mu^*$ of $Z=(Y,X)$ fixed.
We then compute the identified set for $\beta$ at $T=2$ under each specification, with the baseline as the reference point and the comparison of interest between the two relaxations.%
\footnote{Supplemental Appendix~\ref{app:computation_example1} derives the LPs and, where needed, the row bounds, and Additional Appendix~\ref{app:parametric_details} reports implementation details.}

\begin{figure}
    \centering
    \includegraphics[width=0.49\linewidth]{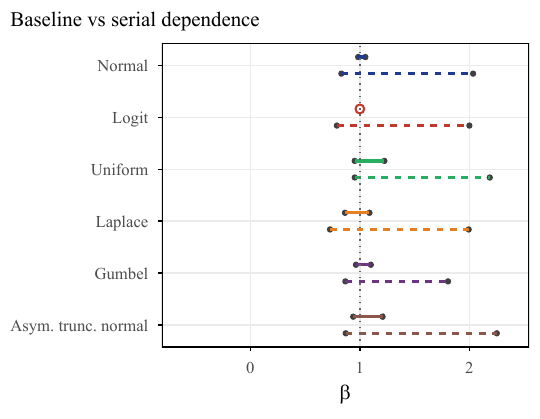}\hfill
    \includegraphics[width=0.49\linewidth]{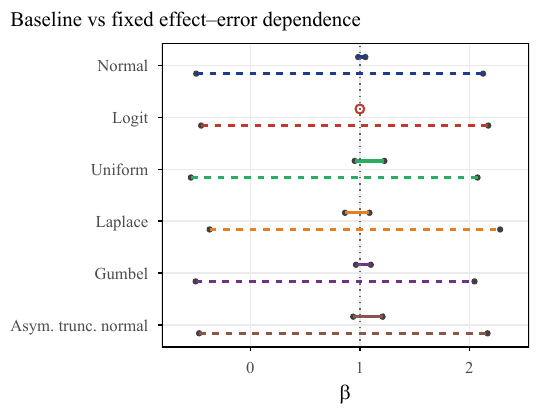}
    \caption{Identified sets for $\beta$ at $T=2$.
    Solid: baseline.
    Dashed: serial dependence (left) and fixed effect--error dependence (right).
    Dotted: $\beta_0=1$.
    Open circle: point identification.}
    \label{fig:numerical_example1_beta}
\end{figure}
Under the baseline, Figure~\ref{fig:numerical_example1_beta} shows that $\beta$ is point identified at the logit link, the classical conditional logit result \parencite{cham1980,Chamberlain2010}, while the other five links leave sets of width $0.066$ (probit) to $0.27$ (uniform).%
\footnote{Along a path from standardized probit to standardized logit, the coefficient sets contract toward $\beta_0$ (Additional Appendix~\ref{app:parametric_mixture}).}
Across all six links, serial dependence widens the identified set but leaves it bounded away from zero, while fixed effect--error dependence produces a set that contains zero.
In these designs, independence between the fixed effect and the error terms therefore carries more identifying content than serial independence.

In Figure~\ref{fig:numerical_example1_beta}, a coefficient value belongs to the identified set when its enclosure is the single point zero, and lies outside when the enclosure lies entirely above zero, so neither verdict rests on a membership threshold for the ADF (Theorems~\ref{thm:enclosure} and~\ref{thm:main_sharpness}).
The values left undecided occupy a band narrower than $2\times10^{-4}$ around each endpoint, well inside the width of the plotted line (Additional Appendix~\ref{app:parametric_implementation}).

\begin{figure}
    \centering
    \includegraphics[width=\linewidth]{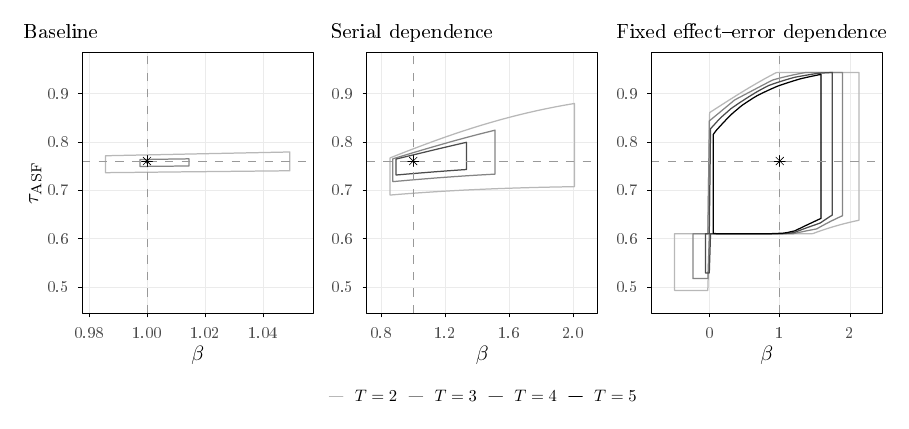}
    \caption{Joint identified sets for $(\beta,\tau_{\mathrm{ASF}})$ in the probit design, at $T=2,3,4$, and $T=5$ in the right panel.
    Asterisk: true value.
    Each panel has its own scale for $\beta$.}
    \label{fig:numerical_example1_joint}
\end{figure}
For the probit link, we also compute joint identified sets for $(\beta,\tau_{\mathrm{ASF}})$ under all three specifications, where $\tau_{\mathrm{ASF}}\equiv\Pr(\beta+A-V_1\ge0)$ is the ASF at the counterfactual covariate value $\bar x=1$.
Under the baseline, the identified set for $\beta$ is narrower than $2\times10^{-3}$ by $T=4$.
Under serial dependence, the set at $T=4$ is still wider, for both $\beta$ and the ASF, than the baseline set at $T=2$ (Figure~\ref{fig:numerical_example1_joint}).
Under fixed effect--error dependence, contraction is even slower, and a fifth period is needed before the identified set for $\beta$ excludes zero, so $T=5$ periods identify the sign of the coefficient without any restriction on the dependence between the fixed effect and the error terms.
In this design, the rapid contraction of identified sets with the number of time periods that one expects in a panel model is specific to the baseline specification.

\begin{figure}
    \centering
    \includegraphics[width=\linewidth]{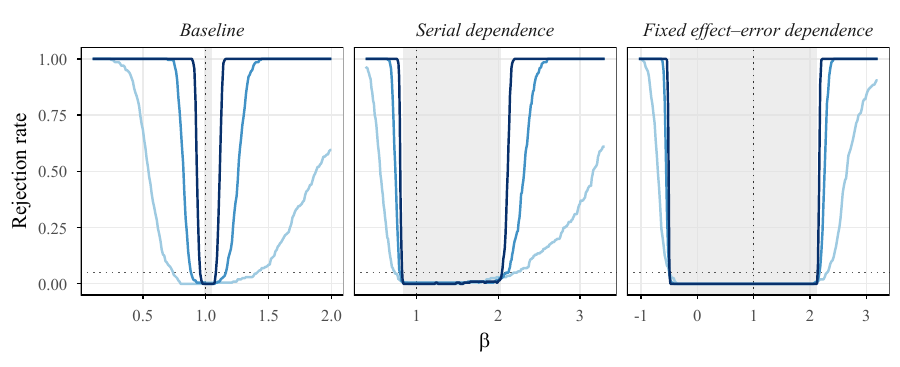}
    \caption{Rejection rates for $\beta$ in the probit design, at nominal level $0.05$, for $n=10^3,10^4,10^5$ (light to dark).
    Shaded: identified set.
    Dotted: $\beta_0=1$ and the nominal level.
    Each panel has its own scale for $\beta$.
    Based on $200$ Monte Carlo samples with $499$ bootstrap draws each.}
    \label{fig:numerical_example1_power}
\end{figure}
We next illustrate inference on the coefficient using the penalized bootstrap test of Section~\ref{sec:inference}.
Because it rejects only when the bounds settle the comparison of the statistic with the critical value (Supplemental Appendix~\ref{app:inference_computation}), the test rejects no more often than the exact test on the same bootstrap draws.
We fix the penalty at $\lambda_n=2\sqrt{\log n}$, which satisfies the conditions of Proposition~\ref{P:inf_finite}.
In the probit design at $T=2$, only $14\%$ of units switch both covariates and outcomes, which suggests that the nominal sample sizes overstate the information available.
In Figure~\ref{fig:numerical_example1_power}, the Monte Carlo rejection rate never exceeds the nominal level at any evaluated coefficient inside the identified set, in any specification, at any sample size.
Power outside the set rises sharply with the sample size.
At $n=10^5$, the rejection rate reaches $1$ away from the identified set in all three specifications, and remains low only near its boundary.
Additional Appendix~\ref{app:parametric_inference_implementation} records the implementation and reports the rejection curves for $c\in\{1/2,1,2,5\}$ in $\lambda_n=c\sqrt{\log n}$.
Constants at or below $2$ hold the nominal level in all three specifications at every sample size, while $c=5$ over-rejects inside the identified set, most at $n=10^3$.
Power increases with $c$, sharply at $n=10^3$ and little at $n=10^5$.

\subsection{Sequential exogeneity}
\label{subsec:numerical_semiparametric}

We now drop the known error distribution of Section~\ref{subsec:numerical_example1} and impose only the sequential exogeneity (SE) of Example~\ref{ex:binary_sequential}, which permits feedback from past outcomes into future covariates.
Existing identification results for the coefficient under SE require a parametric error distribution or a special regressor.%
\footnote{Supplemental Appendix~\ref{app:semiparametric_binary_choice} derives the program for SE.
Additional Appendix~\ref{app:se_details} carries the results behind the computed sets.}
As a benchmark we carry along conditional stationarity (CS), which equates the error distributions across periods given the whole covariate path,
\begin{equation}
\label{eq:cond_stat}
V_t\mid (X,A) \overset{d}{=} V_1\mid (X,A), \qquad t=2,\ldots,T,
\end{equation}
whose identified set for the coefficient is characterized by the rank inequalities of \textcite{Manski1987}, which are sharp \parencite{pakesPorterMomentInequalities,gaoIdentificationNonlinearDynamic2024}.
Because CS conditions on future covariates, it rules out feedback.
A CS structure becomes an SE structure once the fixed effect is redefined as the pair $(A,X)$, so the CS set is contained in the SE set.

Both are evaluated on one family of designs,
\begin{equation}
\label{eq:cs_dgp}
Y_{t} = \one\{(t-1)+\beta_0\, X_{t} + A - V_{t} \ge 0\},
\qquad
V_t \mid (A,X_1,\ldots,X_t) \sim \mathcal N(0,1) \text{ i.i.d.},
\end{equation}
with $T\in\{2,3\}$, $\beta_0=0.4$, covariates on $\{0,1,2\}$ with $X_1$ uniform, and $A=\gamma_0+\gamma_1 X_1+\sigma_a\,\xi$ where $\xi\sim\mathcal N(0,1)$ is independent of $X_1$ and the error terms, and, for $t\ge1$,
\begin{equation}
\label{eq:cs_covariate}
X_{t+1}=
\begin{cases}
\min\{\max\{X_{t}+2Y_{t}-1,\,0\},\,2\} & \text{with probability }\rho,\\
X_{t} & \text{with probability }(1-\rho)\phi,\\
\text{a uniform draw on }\{0,1,2\} & \text{with probability }(1-\rho)(1-\phi).
\end{cases}
\end{equation}
The first line is feedback, the second persistence, and the third a fresh uniform draw, so $\rho=\phi=0$ makes the covariate independent of the past.
In both specifications the coefficient on the time trend is known and equal to one.
The baseline sets $(\gamma_0,\gamma_1,\sigma_a)=(-1,0.5,1)$ and $\rho=\phi=0$.
Because each error is drawn independently of the past, SE holds at every design in the family, while CS holds exactly when $\rho=0$.
Figure~\ref{fig:numerical_beta_comparison} varies feedback $\rho$, persistence $\phi$, and the dependence $\gamma_1$ of the fixed effect on the initial covariate, one at a time.

\begin{figure}
    \centering
    \includegraphics[width=\linewidth]{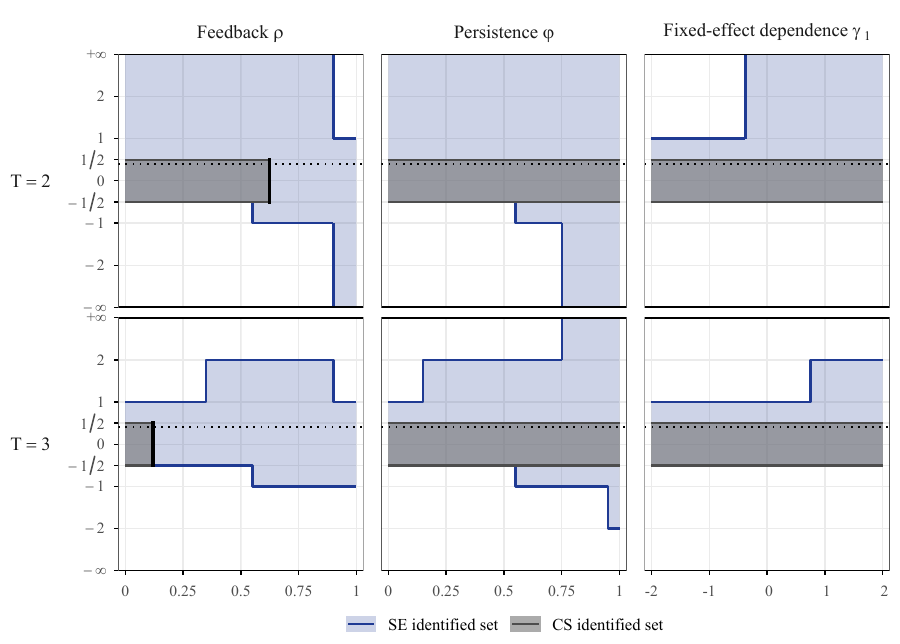}
    \caption{Identified sets for $\beta$ under SE (blue) and CS (grey), at $T=2$ (top) and $T=3$ (bottom), varying one design parameter per column.
    A band reaching the frame is unbounded on that side.
    The CS band ends where the CS set becomes empty.
    Dotted horizontal: $\beta_0=0.4$.}
    \label{fig:numerical_beta_comparison}
\end{figure}

Feedback breaks CS but not SE.
Under both CS and SE, the coefficient enters only through how it orders the period indices $t-1+\beta x_t$ along each covariate path, so the endpoints of both sets are values at which two indices coincide (Additional Appendix~\ref{app:se_details}).
At every design we compute, the CS set is either empty or the same interval $(-1/2,1/2)$, at both horizons.
The set is empty once $\rho$ exceeds about $0.62$ at two periods and about $0.12$ at three.
A third period leaves the nonempty CS set unchanged but lowers the feedback strength at which it becomes empty.

Because SE holds at every design, its identified set is never empty.
At the baseline the set is $(-1/2,+\infty)$ at two periods and $(-1/2,1)$ at three, so the sign of the coefficient is not identified and the third period supplies an upper bound.

Unlike the CS set, the SE set varies with each design parameter.
At two periods, feedback widens it to $(-1,+\infty)$ and, at the strongest feedback, replaces the lower bound with an upper bound at $1$, so the set is not monotone in $\rho$.%
\footnote{The feedback move also shifts the marginal distribution of the later covariates away from uniform, so the $\rho$ axis varies the composition of covariate paths together with feedback.}
Greater persistence only widens the set: from $\phi=0.75$ the two-period set is all of $\R$.
Only $\gamma_1$ shrinks the set, adding an upper bound at $1$ once $\gamma_1\le-1/2$.
Adding a third period can only shrink the identified set.
Under strong feedback it recovers a finite lower bound that the two-period set has lost, and under persistence it recovers a lower bound that falls from $-1/2$ to $-2$ as $\phi$ rises.

\appendix
\section{Proofs}
\label{app:proofs}

\begin{proof}[Proof of Theorem~\ref{thm:main_sharpness}]
Fix $\theta\in\Theta$.
If $\Gamma_\theta=\varnothing$, then $\Mtheta=\varnothing$ and $T(\theta)=+\infty$, so $\theta$ lies in neither $\Thetaw$ nor $\Thetam$.
Hence we assume that $\Gamma_\theta\ne\varnothing$ and show that $\mu^*\in\Mtheta$ if and only if $T(\theta)=0$.

We use three facts about $\Mtheta$: it is convex; it is dominated by $\lambda_\theta$; and taking the TV closure does not change the inner infima defining $T(\theta)$.
First, the set $\Gamma_\theta$ is convex, because the restrictions of Assumption~\ref{asm:moment_conditions} are linear in $\gamma$ and integrability of the moment functions is preserved under convex combinations.
The operator $\Lgen$ of Assumption~\ref{asm:linear_operator} preserves convex combinations, so $\mathcal M_\theta=\Lgen\Gamma_\theta$ is also convex.
Convexity passes to the TV closure, because a convex combination of TV limits is the TV limit of the convex combinations.
Second, domination survives TV limits: if $\mu_n\in\mathcal M_\theta$ satisfy $\norm{\mu_n-m}_{\mathrm{TV}}\to0$, then $\mu_n\ll\lambda_\theta$ by Assumption~\ref{asm:domination}, and $\lambda_\theta(B)=0$ implies $m(B)=\lim_n\mu_n(B)=0$; thus every element of $\Mtheta$ is dominated by $\lambda_\theta$.
Third, for each $\phi\in\Phi(\mathcal Z)$, TV-continuity of $\mu\mapsto\EE{\mu}{\phi}$ implies that the infimum over $\mathcal M_\theta$ equals the infimum over its closure $\Mtheta$, which is the second equality in~\eqref{eq:discrepancy_function_model}.

Suppose $\mu^*\in\Mtheta$.
For every $\phi\in\Phi(\mathcal Z)$, the measure $\mu^*$ is feasible in the inner problem in~\eqref{eq:discrepancy_function_model}, so the inner infimum is at most $\EE{\mu^*}{\phi}-\EE{\mu^*}{\phi}=0$.
Taking the supremum over $\phi$, with $\phi\equiv0$ yielding zero, gives $T(\theta)=0$.

Conversely, suppose $\mu^*\notin\Mtheta$.
We distinguish the two cases (i) $\mu^*\not\ll\lambda_\theta$ and (ii) $\mu^*\ll\lambda_\theta$.
\emph{Case (i).}
Pick a measurable set $B$ with $\lambda_\theta(B)=0<\mu^*(B)$; every $m\in\Mtheta$ assigns $m(B)=0$ by the second fact, so $\phi(z)=\one\{z\in B\}$ gives $T(\theta)\ge\mu^*(B)>0$.

\emph{Case (ii).}
Here the $\sigma$-finite measure $\lambda_\theta$ supplied by Assumption~\ref{asm:domination} dominates $\mu^*$, and by the second fact above it also dominates every $m\in\Mtheta$.
For each probability measure $\mu$ dominated by $\lambda_\theta$, write $f_\mu\equiv\frac{\d\mu}{\d\lambda_\theta}\in L^1(\lambda_\theta)$.
The density map $\mu\mapsto f_\mu$ is affine, and Scheff\'e's theorem \parencite[Lemma~2.1]{tsybakov2008introduction} gives $\norm{\mu-m}_{\mathrm{TV}}=\tfrac12\norm{f_\mu-f_m}_{L^1(\lambda_\theta)}$ for any probability measures $\mu$ and $m$ dominated by $\lambda_\theta$.

Although Theorem~\ref{thm:main_sharpness} may be proved directly with the separating hyperplane theorem, we give a minimax argument which makes the connection between $T(\theta)$ and the TV norm transparent.
Let $C_\theta\equiv\{f_m:m\in\Mtheta\}$, which is the convex set of probability densities corresponding to $\Mtheta$, and let $f^*\equiv f_{\mu^*}$.
For each $\phi\in\Phi(\mathcal{Z})$, let $\tilde{\phi}$ denote its equivalence class in $L^\infty(\lambda_\theta)$.
The set $\tilde{\Phi}\equiv\{\tilde{\phi}:\phi\in\Phi(\mathcal{Z})\}$ is the set of elements of $L^\infty(\lambda_\theta)$ with values in $[0,1]$ $\lambda_\theta$-almost everywhere.

Let $\overline{B}_{L^\infty}$ denote the closed unit ball of $L^\infty(\lambda_\theta)$, that is, the set of elements bounded between $-1$ and $1$ $\lambda_\theta$-almost everywhere.
Then $\tilde{\Phi}=\{(1+\psi)/2:\psi\in\overline{B}_{L^\infty}\}$ is the image of the weak-* compact set $\overline{B}_{L^\infty}$ \parencite[Theorem~III.3.1]{C1994} under a weak-* continuous affine map, and is therefore weak-* compact.
Sion's minimax theorem \parencite{Sion1958} and $L^1$--$L^\infty$ duality \parencite[Theorem~III.5.6]{C1994} imply
\begin{align*}
    T(\theta)
    &=\sup_{\tilde{\phi}\in\tilde{\Phi}}\inf_{f\in C_\theta}\int_{\mathcal{Z}}\tilde{\phi}(f^*-f)\,\d\lambda_\theta
    =\inf_{f\in C_\theta}\sup_{\tilde{\phi}\in\tilde{\Phi}}\int_{\mathcal{Z}}\tilde{\phi}(f^*-f)\,\d\lambda_\theta\\
    &=\inf_{f\in C_\theta}\tfrac12\norm{f^*-f}_{L^1(\lambda_\theta)}
    =\inf_{\mu\in\Mtheta}\norm{\mu^*-\mu}_{\mathrm{TV}}.
\end{align*}
Therefore, $T(\theta)=0$ if and only if $\mu^*$ is the TV limit of a sequence in $\Mtheta$, which holds if and only if $\mu^*\in\Mtheta$ because $\Mtheta$ is TV closed.
\end{proof}

\begin{proof}[Proof of Theorems~\ref{thm:exact_computation} and~\ref{thm:enclosure}]
We prove the enclosure of Theorem~\ref{thm:enclosure} first; Theorem~\ref{thm:exact_computation} is the case $r_{\mathrm{low}}=0$, which row certification supplies.
The kernel property and finiteness of $\mathcal Z$ make $q_{\phi,\theta}$ bounded and measurable.
Its definition and~\eqref{eq:kernel_operator} give, for every $\phi\in\Phi(\mathcal Z)$ and for every $\gamma\in\Pa(\mathcal W)$,
\begin{equation}
\label{eq:qphi-integration-proof}
\EE{\gamma}{q_{\phi,\theta}(W)}
=
\sum_{z\in\mathcal Z}\phi(z)\int_{\mathcal W}K_\theta(\{z\}\mid w)\,\d\gamma(w)
=
\EE{\Lgen\gamma}{\phi}.
\end{equation}

The dual of~\eqref{eq:TLD_restricted} associates a nonnegative weight $p_w$ with each retained row.
The free variables $\zeta$ and $u$ and the restriction $v\ge0$ give, respectively,
$$
\sum_{w\in\mathcal W'}p_w=1,\quad
\sum_{w\in\mathcal W'}p_wg_{1,j}(w;\theta)=0
\quad(j\in\mathcal J'),\quad
\sum_{w\in\mathcal W'}p_wg_{2,k}(w;\theta)\le0
\quad(k\in\mathcal K').
$$
Solving for the dual variables associated with the bounds on $\phi$ coordinatewise, the dual minimizes, over $p$ satisfying these restrictions,
$$
\sup_{\phi\in\Phi(\mathcal Z)}
\left\{
\EE{\mu^*}{\phi}
-
\sum_{w\in\mathcal W'}p_wq_{\phi,\theta}(w)
\right\}.
$$
This is the problem defining $\TLD{\theta}{\mathcal W'}{\mathcal J',\mathcal K'}$, and the same problem with $\mathcal J$ and $\mathcal K$ in place of $\mathcal J'$ and $\mathcal K'$ defines $\TLD{\theta}{\mathcal W'}{\mathcal J,\mathcal K}$.
Under column certification the latter has a nonempty feasible set, a closed subset of the compact probability simplex, so it attains its value at some $p^\dagger$ with $\gamma_{p^\dagger}\in\Gamma_\theta$.
By~\eqref{eq:column_certificate} that value is $T_{\mathrm{LP}}(\theta)$, so $p^\dagger$ is also optimal for the dual displayed above.
The finite LP is feasible at $(\phi,\zeta,u,v)=(0,0,0,0)$, while its certified dual solution is feasible with finite value.
Finite-dimensional LP strong duality therefore implies that the finite LP has an optimizer and gives
$$
T_{\mathrm{LP}}(\theta)
=
\sup_{\phi\in\Phi(\mathcal Z)}
\left\{
\EE{\mu^*}{\phi}
-
\sum_{w\in\mathcal W'}p^\dagger_wq_{\phi,\theta}(w)
\right\}.
$$
Because $\gamma_{p^\dagger}\in\Gamma_\theta$, the inner infimum defining $T(\theta)$ is at most its value at $\gamma_{p^\dagger}$; using~\eqref{eq:qphi-integration-proof} at $\gamma_{p^\dagger}$ and $\EE{\gamma_{p^\dagger}}{q_{\phi,\theta}(W)}=\sum_{w\in\mathcal W'}p^\dagger_wq_{\phi,\theta}(w)$ yields
$$
T(\theta)
\le
\sup_{\phi\in\Phi(\mathcal Z)}
\left\{
\EE{\mu^*}{\phi}
-
\EE{\gamma_{p^\dagger}}{q_{\phi,\theta}(W)}
\right\}
=T_{\mathrm{LP}}(\theta).
$$
Moreover, $\Gamma_\theta$ is nonempty because it contains $\gamma_{p^\dagger}$, and the test function $\phi\equiv0$ gives $T(\theta)\ge0$.

Now fix any optimizer $(\phi^*,\zeta^*,u^*,v^*)$ of the finite LP, with row residual $r^*$, and any finite $r_{\mathrm{low}}\le r^*$.
The definition in~\eqref{eq:row_oracle_residual} implies
$$
q_{\phi^*,\theta}(w)
\le
\zeta^*-r_{\mathrm{low}}
+
\sum_{j\in\mathcal J'}u_j^*g_{1,j}(w;\theta)
+
\sum_{k\in\mathcal K'}v_k^*g_{2,k}(w;\theta)
\quad\text{for all }w\in\mathcal W.
$$
For any $\gamma\in\Gamma_\theta$, each moment function is $\gamma$-integrable by Assumption~\ref{asm:moment_conditions}, so integrating this inequality and using the moment restrictions together with $v^*\ge0$ gives
$$
\EE{\gamma}{q_{\phi^*,\theta}(W)}
\le
\zeta^*-r_{\mathrm{low}}
+
\sum_{j\in\mathcal J'}u_j^*\EE{\gamma}{g_{1,j}(W;\theta)}
+
\sum_{k\in\mathcal K'}v_k^*\EE{\gamma}{g_{2,k}(W;\theta)}
\le
\zeta^*-r_{\mathrm{low}}.
$$
Using~\eqref{eq:qphi-integration-proof} and $T_{\mathrm{LP}}(\theta)=\EE{\mu^*}{\phi^*}-\zeta^*$, this reads $\inner{\phi^*,\mu^*-\Lgen\gamma}\ge T_{\mathrm{LP}}(\theta)+r_{\mathrm{low}}$ for every $\gamma\in\Gamma_\theta$.
Taking the infimum over $\gamma$ and bounding the outer supremum below by its value at $\phi^*$ gives
$$
T_{\mathrm{LP}}(\theta)+r_{\mathrm{low}}
\le
\inf_{\gamma\in\Gamma_\theta}
\inner{\phi^*,\mu^*-\Lgen\gamma}
\le
T(\theta).
$$
Combining this inequality with $T(\theta)\ge0$ and $T(\theta)\le T_{\mathrm{LP}}(\theta)$ proves~\eqref{eq:certified_enclosure} for finite $r_{\mathrm{low}}$; for $r_{\mathrm{low}}=-\infty$, the display is the pair of bounds $0\le T(\theta)\le T_{\mathrm{LP}}(\theta)$ just cited.
This proves Theorem~\ref{thm:enclosure}.

For Theorem~\ref{thm:exact_computation}, add row certification.
The argument above produced an optimizer of~\eqref{eq:TLD_restricted}, and row certification makes the row residual~\eqref{eq:row_oracle_residual} zero at one of them, so~\eqref{eq:certified_enclosure} applies at that optimizer with $r_{\mathrm{low}}=0$ and yields
$$
\max\{0,T_{\mathrm{LP}}(\theta)\}\le T(\theta)\le T_{\mathrm{LP}}(\theta).
$$
Taking $\phi\equiv0$, $\zeta=0$, $u=0$ and $v=0$ in~\eqref{eq:TLD_restricted} shows $T_{\mathrm{LP}}(\theta)\ge0$, so the outer bounds coincide and~\eqref{eq:exact_computation} follows.
\end{proof}

\begin{proof}[Proof of Proposition~\ref{P:inf_finite}]
We establish the result in three parts, corresponding to the bound \eqref{E:inf_bound}, the limiting distribution \eqref{E:inf_convergence}, and the bootstrap consistency \eqref{E:inf_bootstrap}.

\paragraph{Preliminaries.}
For notational convenience we denote $\Phi(\mathcal{Z})$ as $\Phi$, and let $\ell^\infty(\Phi)$ denote the space of uniformly bounded functions from $\Phi$ to $\R$ equipped with the uniform norm $\norm{\cdot}_\infty$.
By Assumption~\ref{asm:finite_Z}, $\mathcal{Z}$ is finite, so the set $\Phi$ can be identified with the set $[0,1]^{\mathcal Z}$ by the map $\phi \mapsto (\phi(z))_{z \in \mathcal Z}$.
The set $\Phi$ has envelope $F \equiv 1$. 
For $0<\ve\le1$, the set $\Phi$ can be covered by at most $C\ve^{-|\mathcal Z|}$ uniform-norm balls of radius $\ve$, for a constant $C$ depending only on $|\mathcal Z|$.
As $\norm{\phi - \phi'}_{Q,2} \le \norm{\phi - \phi'}_\infty$ for every probability measure $Q$, $\Phi$ meets the uniform entropy requirement of Theorem~2.8.3 of \textcite{VW1996}, and is Donsker and pre-Gaussian uniformly in $\mu^* \in \mathbf{P}$.
Because $\Phi=[0,1]^{\mathcal Z}$ is separable under $\|\cdot\|_\infty$, the empirical-process suprema below are measurable and no outer-probability qualifications are needed.

In particular, the processes $\mathbb{G}_{n,\mu^*}$ converge to tight limits $\mathbb{G}_{\mu^*}$ uniformly in the bounded Lipschitz metric, and are uniformly asymptotically tight and equicontinuous with respect to $\norm{\cdot}_\infty$. For every $\theta$ and $\phi$,
$
\eta_{\theta,n}(\phi)-\eta_{\theta,\mu^*}(\phi)
=
\EE{n}{\phi}-\EE{\mu^*}{\phi},
$
with no dependence on $\theta$ in the difference. Hence,
\begin{align}
    \sup_{\phi \in \Phi} \sup_{\substack{\theta \in \Theta}}
    | \eta_{\theta, n}(\phi) - \eta_{\theta, \mu^*}(\phi)|
    &=
    \sup_{\phi \in \Phi}
    n^{-1/2}|\mathbb{G}_{n,\mu^*}(\phi)|
    =
    O_p(n^{-1/2})
    \text{ uniformly in }\mu^* \in \mathbf{P}.
    \label{E:etaconv}
\end{align}
Since $\lambda_n=o(\sqrt n)$, it follows that $\lambda_n \eta_{\theta, n} = \lambda_n \eta_{\theta, \mu^*} + o_p(1)$ uniformly in $\mu^*$ and $\theta \in \Theta$.
Similarly, $\lambda_n |T_n(\theta) - T_{\mu^*}(\theta) | = o_p(1)$ uniformly in $\mu^* \in \mathbf{P}, \theta \in \Theta$.

\paragraph{Step one: proof of \eqref{E:inf_bound}.}
Let $\mu^* \in \mathbf{P}$ and $\theta \in \Theta_{\mathrm{I}}(\mu^*)$.
The map $\eta_{\theta, n}: \Phi \ra \R$ is continuous with respect to the uniform norm and achieves its maximum $T_n(\theta)$ at a point which we may denote $\hat{\phi}_n$.
As $\mu^* \in \overline{\mathcal{M}}_\theta$,
$$
    \sqrt{n}T_n(\theta)
    =
    \sqrt{n} \inf_{\mu \in \overline{\mathcal{M}}_\theta}
    (\EE{n}{\hat{\phi}_n} - \EE{\mu}{\hat{\phi}_n})
    \le
    \sqrt{n}
    (\EE{n}{\hat{\phi}_n} - \EE{\mu^*}{\hat{\phi}_n})
    =
    \mathbb{G}_{n,\mu^*}(\hat{\phi}_n).
$$
Evaluating the supremum in \eqref{E:inf_bound} at $\hat{\phi}_n$ yields precisely $\mathbb{G}_{n,\mu^*}(\hat{\phi}_n)$ and proves \eqref{E:inf_bound}.

\paragraph{Step two: proof of \eqref{E:inf_convergence}.}
Let $\mu^* \in \mathbf{P}$ and $\theta \in \Theta$ be arbitrary and fixed.
The map $\phi \mapsto \eta_{\theta, \mu^*}(\phi)$ is continuous with respect to uniform norm on $\Phi$, and $\Phi$ is uniform-norm compact.
Therefore, $\eta_{\theta, \mu^*}$ attains its maximum value of $T_{\mu^*}(\theta)$ over $\Phi$, $K_{\mu^*}(\theta)$ is nonempty, and
\begin{align}
    \sqrt{n}(T_n(\theta) - T_{\mu^*} (\theta ))
    &=
    \sup_{\phi \in\Phi}
    \left(
    \mathbb{G}_{n,\mu^*}(\phi)
    +
    \sqrt{n} \eta_{\theta, \mu^*}(\phi)
    -
    \sqrt{n} T_{\mu^*}(\theta)
    \right)
    \ge
    \sup_{\phi \in K_{\mu^*}(\theta)}
    \mathbb{G}_{n,\mu^*}(\phi).
    \label{E:lowerb1}
\end{align}

Let $\delta>0$ be arbitrary and let $K_{\mu^*}^\delta(\theta)$ denote the $\delta$-open neighborhood of $K_{\mu^*}(\theta)$ in $\Phi$.
If $K_{\mu^*}^\delta(\theta)=\Phi$, then \eqref{E:deltaineq} below is immediate.
Otherwise, $\Phi\setminus K_{\mu^*}^\delta(\theta)$ is nonempty and compact, so continuity of $\eta_{\theta,\mu^*}$ and the definition of $K_{\mu^*}(\theta)$ imply
$$
    c
    \equiv
    \sup_{\phi\notin K_{\mu^*}^\delta(\theta)}
    \eta_{\theta,\mu^*}(\phi)
    <
    T_{\mu^*}(\theta).
$$
Because $\mathbb{G}_{n,\mu^*}$ is uniformly $O_p(1)$,
\begin{align*}
&\limsup_{n \ra \infty}
\PP{\mu^*}{
\sup_{\phi \not\in K_{\mu^*}^\delta (\theta)}
\left(
\mathbb{G}_{n,\mu^*}(\phi)
+
\sqrt{n} \eta_{\theta, \mu^*} (\phi)
-
\sqrt{n} T_{\mu^*}(\theta)
\right)
\ge
\sup_{\phi \in K_{\mu^*}(\theta)}
\mathbb{G}_{n,\mu^*}(\phi)
}
\\
&\quad\le
\limsup_{n \ra \infty}
\PP{\mu^*}{
\sup_{\phi \in \Phi} \mathbb{G}_{n,\mu^*}(\phi)
-
\inf_{\phi \in \Phi} \mathbb{G}_{n,\mu^*}(\phi)
\ge
\sqrt{n}
\left(
T_{\mu^*}(\theta)-c
\right)
}
=
0.
\end{align*}

In light of the equality and the lower bound in~\eqref{E:lowerb1}, the preceding display implies that
\begin{align}
&\limsup_{n \ra \infty}
\PP{\mu^*}{
\sqrt{n}(T_n(\theta)-T_{\mu^*}(\theta))
>
\sup_{\phi \in K_{\mu^*}^\delta(\theta)}
\mathbb{G}_{n,\mu^*}(\phi)
}
\nonumber\\
&\le
\limsup_{n \ra \infty}
\PP{\mu^*}{
\sqrt{n}(T_n(\theta)-T_{\mu^*}(\theta))
>
\sup_{\phi \in K_{\mu^*}^\delta(\theta)}
(
\mathbb{G}_{n,\mu^*}(\phi)
+
\sqrt{n}\eta_{\theta,\mu^*}(\phi)
-
\sqrt{n}T_{\mu^*}(\theta)
)
}
\nonumber\\
&=
0.
\label{E:deltaineq}
\end{align}

By asymptotic equicontinuity of the empirical process $\mathbb{G}_{n,\mu^*}$, for every $\ve > 0$ there exists some $\delta> 0$ such that
\begin{align}
\label{E:unifaseq}
\limsup_{n \ra \infty}
\sup_{\mu^* \in \mathbf{P}}
\PP{\mu^*}{
\sup_{\norm{\phi - \phi'}_\infty < \delta}
|
\mathbb{G}_{n,\mu^*}(\phi)
-
\mathbb{G}_{n,\mu^*}(\phi')
|
>
\ve
}
<
\ve.
\end{align}

Let $\ve >0$ be arbitrary and choose $\delta$ to satisfy \eqref{E:unifaseq}.
Then,
\begin{align*}
\limsup_{n \ra \infty}
\sup_{\mu^* \in \mathbf{P}}
\PP{\mu^*}{
\left|
\sup_{\phi \in K_{\mu^*}^\delta(\theta)}
\mathbb{G}_{n,\mu^*}(\phi)
-
\sup_{\phi \in K_{\mu^*}(\theta)}
\mathbb{G}_{n,\mu^*}(\phi)
\right|
>
\ve
}
<
\ve.
\end{align*}

This bound, the fact that $\ve$ is arbitrary, \eqref{E:deltaineq}, and the lower bound in \eqref{E:lowerb1} imply
\begin{align*}
\sqrt n\{T_n(\theta)-T_{\mu^*}(\theta)\}
-
\sup_{\phi\in K_{\mu^*}(\theta)}
\mathbb G_{n,\mu^*}(\phi)
=
o_p(1).
\end{align*}
Since
$
G\mapsto
\sup_{\phi\in K_{\mu^*}(\theta)}G(\phi)
$
is $1$-Lipschitz on $\ell^\infty(\Phi)$ and
$\mathbb G_{n,\mu^*}\Rightarrow\mathbb G_{\mu^*}$,
the continuous mapping theorem and Slutsky's theorem give \eqref{E:inf_convergence}.

\paragraph{Step three: proof of \eqref{E:inf_bootstrap}.}
Let $\xi: \ell^\infty(\Phi) \ra \R$ be the functional
$
\xi:
G
\mapsto
\sup_{\phi \in \Phi}
\inf_{\mu \in \overline{\mathcal{M}}_\theta}
\left(
G(\phi)-\EE{\mu}{\phi}
\right).
$
Note that $\xi$ is Lipschitz continuous.
Let $\mathcal{G}: \mathcal{P}(\mathcal{Z}) \ra \ell^\infty(\Phi)$ be the functional which maps
$
\mathcal{G}(\mu)=\EE{\mu}{\cdot},
$
and note that, with $\mu_n\equiv n^{-1}\sumn\delta_{Z_i}$ the empirical measure,
$
T_n(\theta)=\xi(\mathcal{G}(\mu_n)).
$

For any $H\in\ell^\infty(\Phi)$, define
$
f_H(t)
\equiv
\xi\{\mathcal G(\mu^*)+tH\}
$
for $t\in\mathbb R$.
Because $\xi$ is finite, convex, and Lipschitz, $f_H$ is a finite convex function on $\mathbb R$.
Theorem~23.1 of \textcite{rockafellar1970convex}, applied to $f_H$, implies that
$
\lim_{t\downarrow0}
t^{-1}
[
\xi\{\mathcal G(\mu^*)+tH\}
-
\xi\{\mathcal G(\mu^*)\}
]
$
exists and is finite for every $H\in\ell^\infty(\Phi)$.
Thus $\xi$ is directionally differentiable at $\mathcal G(\mu^*)$.
Since $\xi$ is Lipschitz, Proposition~2.49 of \textcite{bonnans2000perturbation} implies that $\xi$ is Hadamard directionally differentiable there.

Lemma~A.2 of \textcite{linton2010} implies that
$
\mathbb{G}_{n, \mu^*}^*
\overset{\mu^*}{\Rightarrow}
\mathbb{G}_{\mu^*}
$
uniformly in $\mu^*$, and we have already shown that
$\sqrt{n}
\left(
\mathcal{G}(\mu_n)
-
\mathcal{G}(\mu^*)
\right)
=
\mathbb{G}_{n,\mu^*}
\Rightarrow
\mathbb{G}_{\mu^*},
$
where $\mathbb{G}_{\mu^*}$ is tight.
Therefore, by Theorem~3.1 of \textcite{HongLi2018},
\begin{align}
\label{E:hongliapp}
\lambda_n
\left(
\xi(
\mathcal{G}(\mu_n)
+
\lambda_n^{-1}
\mathbb{G}_{n,\mu^*}^*
)
-
\xi(\mathcal{G}(\mu_n))
\right)
\overset{\mu^*}{\Rightarrow}
\xi'(\mathbb{G}_{\mu^*}),
\end{align}
where $\xi'$ denotes the directional derivative.
The same result shows that the limit in \eqref{E:inf_convergence} is $\xi'(\mathbb{G}_{\mu^*})$.
The left-hand side of \eqref{E:hongliapp} is precisely the left-hand side of \eqref{E:inf_bootstrap}, which establishes \eqref{E:inf_bootstrap}.

\end{proof}

\printbibliography[heading=bibintoc]

\clearpage
\begin{refsection}
\setcounter{page}{1}
\renewcommand{\thepage}{S\arabic{page}}
\setcounter{section}{0}
\renewcommand{\thesection}{S\Alph{section}}
\renewcommand{\theHsection}{S\Alph{section}}
\begin{center}
  {\Large\bfseries Supplemental Appendix to\\[0.5ex]
  ``An Adversarial Approach to Identification, Computation, and Inference in Models with a Linear-in-Measures Representation''}\\[1ex]
  Irene Botosaru, Isaac Loh, and Chris Muris
\end{center}

\section{Column-and-row generation algorithm}
\label{app:algorithm}

We maintain the assumptions of Section~\ref{sec:computation} and fix $\theta\in\Theta$.
This appendix describes the finite linear program, how column certification is obtained, the row bound, and the feasibility check, and then states the column-and-row generation algorithm.

\paragraph{Auxiliary finite linear program (FINLP).}
At finite sets $(\mathcal W',\mathcal J',\mathcal K')$ with $\mathcal W'\ne\varnothing$, FINLP solves the auxiliary finite linear program~\eqref{eq:TLD_restricted} and its LP dual.
We solve these programs with standard linear programming solvers (Gurobi and HiGHS).
When the retained moment restrictions admit probability weights on $\mathcal W'$, the two programs have the same value and yield an optimizer $(\phi^*,\zeta^*,u^*,v^*)$ of~\eqref{eq:TLD_restricted}.

\paragraph{Column certification.}
The algorithm obtains column certification by construction, through the following lemma.

\begin{lemma}[Finite column reduction]
\label{lem:finite_column_reduction}
Fix a nonempty finite support $\mathcal W'\subseteq\mathcal W$, and suppose $\mathcal K$ is finite.
There exists $\mathcal J'\subseteq\mathcal J$, with $|\mathcal J'|\le|\mathcal W'|$, such that
$$
\Gamma_\theta(\mathcal W',\mathcal J',\mathcal K)
=
\Gamma_\theta(\mathcal W',\mathcal J,\mathcal K).
$$
If this set is nonempty, $(\mathcal J',\mathcal K)$ is column-certified for $\mathcal W'$.
\end{lemma}

\begin{proof}
For each $j\in\mathcal J$, let $g_j\equiv(g_{1,j}(w;\theta))_{w\in\mathcal W'}\in\R^{(\mathcal W')}$.
Choose $\mathcal J'\subseteq\mathcal J$ such that $\{g_j:j\in\mathcal J'\}$ is a basis of the span of $\{g_j:j\in\mathcal J\}$, so that $|\mathcal J'|\le|\mathcal W'|$.
Every $g_j$ with $j\in\mathcal J$ is a linear combination of $\{g_j:j\in\mathcal J'\}$, so a probability vector on $\mathcal W'$ satisfies the equalities indexed by $\mathcal J'$ if and only if it satisfies those indexed by $\mathcal J$.
Hence $\Gamma_\theta(\mathcal W',\mathcal J',\mathcal K)=\Gamma_\theta(\mathcal W',\mathcal J,\mathcal K)$.
Because $\mathcal K$ is finite, $(\mathcal J',\mathcal K)$ is a pair of finite index sets.
If the common feasible set is nonempty, the restricted discrepancies at $(\mathcal J',\mathcal K)$ and $(\mathcal J,\mathcal K)$ coincide, and Definition~\ref{def:column_certified} gives column certification.
\end{proof}
Infinitely many moment inequalities need not admit a finite subfamily with the same feasible set.
The algorithm retains every moment restriction of a finite family and, for an infinite family of equalities, a basis as in Lemma~\ref{lem:finite_column_reduction} for the current support, recomputed whenever the support grows.
By Lemma~\ref{lem:finite_column_reduction}, column certification then holds whenever the feasible set on $\mathcal W'$ is nonempty, which FEAS establishes.
With infinitely many moment inequalities, column certification has to be established separately, for example through the characterization that follows Definition~\ref{def:column_certified}.
In Example~\ref{ex:binary_sequential} the cell equalities at each retained atom already impose the full family (Supplemental Appendix~\ref{app:se_exact_pricing}), and Example~\ref{ex:binary_parametric} has no moment restrictions.

\paragraph{Row bound (RB).}
RB returns a row bound: a number $r_{\mathrm{low}}$ at most the row residual $r^*$ of~\eqref{eq:row_oracle_residual} at an optimizer $(\phi^*,\zeta^*,u^*,v^*)$ of~\eqref{eq:TLD_restricted}, or at most the weakly larger residual of Remark~\ref{rem:row_measures}.
The value $r_{\mathrm{low}}=-\infty$ is always valid.
Since $\zeta^*$ is the only term of~\eqref{eq:slack} that does not depend on $w$, any certified lower bound over $\mathcal W$ on the remaining terms, added to $\zeta^*$, gives such an $r_{\mathrm{low}}$.
Additional Appendix~\ref{app:parametric_implementation} computes a bound this way, and Supplemental Appendix~\ref{app:se_exact_pricing} uses the admissible-measure version of the row residual in Remark~\ref{rem:row_measures}.

\paragraph{Feasibility check (FEAS).}
FEAS checks whether the retained moment restrictions can be satisfied by probability weights on the current finite support.
If the corresponding feasibility LP is infeasible, there are multipliers $u\in\R^{(\mathcal J')}$ and $v\in\R^{(\mathcal K')}_+$ such that
$$
h(w)\equiv\sum_{j\in\mathcal J'}u_j g_{1,j}(w;\theta)+\sum_{k\in\mathcal K'}v_k g_{2,k}(w;\theta)
$$
is strictly positive on $\mathcal W'$.
FEAS then uses the model-specific global optimization underlying RB either to find a new input value $w$ with $h(w)\le0$ or to certify that $\inf_{w\in\mathcal W}h(w)\ge\delta>0$.
In the first case, $w$ is added to the support.
In the second, no admissible input measure exists: every $\gamma\in\Gamma_\theta$ satisfies $\EE{\gamma}{h(W)}\le0$ by the retained moment restrictions, whereas $h\ge\delta$ on $\mathcal W$ gives $\EE{\gamma}{h(W)}\ge\delta$.

Algorithm~\ref{alg:master} retains the largest established lower bound and the smallest established upper bound, which remain valid when the retained sets change.
The target width controls precision, not membership: an enclosure that meets it without deciding membership returns \textsc{undecided}.

\begin{algorithm}[!ht]
\footnotesize
\caption{Column-and-row generation}
\label{alg:master}
\begin{algorithmic}[1]
\setlength{\itemsep}{0pt}
\State \textbf{Input:} parameter $\theta$; finite initial sets $(\mathcal W',\mathcal J',\mathcal K')$ with $\mathcal W'\ne\varnothing$; target width $\varepsilon_{\mathrm{alg}}\ge0$.
\State \textbf{Output:} \textsc{in}, \textsc{out}, or \textsc{undecided}, with enclosure $[\underline T,\overline T]$.
\State $\ell\gets 1$; $[\underline T,\overline T]\gets[0,+\infty]$.
\While{the computational budget is not exhausted}
    \State Retain each finite moment family in full and choose a basis for any infinite equality family.
    \State Run FEAS on $(\mathcal W',\mathcal J',\mathcal K')$.
    \If{FEAS returns $w\notin\mathcal W'$ with $h(w)\le0$}
        \State $\mathcal W'\gets\mathcal W'\cup\{w\}$; \textbf{continue}.
    \ElsIf{FEAS returns a global infeasibility certificate}
        \State \Return \textsc{out} with $T(\theta)=+\infty$.
    \ElsIf{FEAS is unresolved}
        \State \textbf{break}.
    \EndIf
    \State Run FINLP, giving $T_{\mathrm{LP},\ell}(\theta)$ and an optimizer.
    \State $\overline T\gets\min\{\overline T,T_{\mathrm{LP},\ell}(\theta)\}$. \Comment{Theorem~\ref{thm:enclosure}, upper endpoint}
    \If{$\overline T=0$}
        \State \Return \textsc{in} with $T(\theta)=0$.
    \EndIf
    \State RB computes a row bound $r_{\mathrm{low},\ell}$ on $r_{\ell}^*$, or on the residual of Remark~\ref{rem:row_measures}, recording violated rows as $\mathcal W^*$.
    \State $\underline T\gets\max\{\underline T,\ T_{\mathrm{LP},\ell}(\theta)+r_{\mathrm{low},\ell}\}$. \Comment{Theorem~\ref{thm:enclosure}, lower endpoint}
    \If{$\underline T>0$}
        \State \Return \textsc{out} with $[\underline T,\overline T]$.
    \ElsIf{$\overline T-\underline T\le\varepsilon_{\mathrm{alg}}$ or RB found no violated row}
        \State \textbf{break}.
    \Else
        \State $\mathcal W'\gets\mathcal W'\cup\mathcal W^*$.
    \EndIf
    \State $\ell\gets\ell+1$.
\EndWhile
\State \Return \textsc{undecided} with $[\underline T,\overline T]$.
\end{algorithmic}
\end{algorithm}

\FloatBarrier

\section{Linear programs for inference}
\label{app:inference_computation}

This appendix shows how to compute the sample discrepancy $T_n(\theta)$, the bootstrap statistic, and the critical value of Section~\ref{sec:inference} with the linear program~\eqref{eq:TLD_restricted} of Section~\ref{sec:computation}.
We maintain the assumptions of Section~\ref{sec:inference}, fix $\theta$ with $\Gamma_\theta\ne\varnothing$, and take $(\mathcal J',\mathcal K')$ to retain every moment restriction, or more generally to define the same feasible set on $\mathcal W'$ as the full family (as in Examples~\ref{ex:binary_parametric} and~\ref{ex:binary_sequential}, or Lemma~\ref{lem:finite_column_reduction}).
When that set is nonempty, column certification holds whatever the objective of~\eqref{eq:TLD_restricted}, because the feasible set does not depend on it.

\paragraph{Sample discrepancy.}
Replacing $\EE{\mu^*}{\phi}$ by $\EE{n}{\phi}$ in~\eqref{eq:TLD_restricted} gives the linear program
\begin{equation}
\label{eq:Tn_LP}
T_{\mathrm{LP},n}(\theta)
\equiv
\sup_{\phi,\zeta,u,v}
\left\{\EE{n}{\phi}-\zeta\right\}
\quad\text{subject to the constraints of~\eqref{eq:TLD_restricted}}.
\end{equation}
For a row bound $r_{\mathrm{low},n}$ at an optimizer of~\eqref{eq:Tn_LP}, Theorem~\ref{thm:enclosure} gives
$$
\underline T_n
\equiv
\max\{0,T_{\mathrm{LP},n}(\theta)+r_{\mathrm{low},n}\}
\le
T_n(\theta)
\le
T_{\mathrm{LP},n}(\theta).
$$

\paragraph{Bootstrap statistic.}
The bootstrap statistic of Section~\ref{sec:inference} is
$$\sup_{\phi\in\Phi(\mathcal Z)}\left(\mathbb{G}_{n,\mu^*}^*(\phi)+\lambda_n\eta_{\theta,n}(\phi)\right)-\lambda_nT_n(\theta).$$
The second term is the same for every bootstrap draw.
Replacing $\EE{\mu^*}{\phi}$ by $\lambda_n^{-1}\mathbb{G}_{n,\mu^*}^*(\phi)+\EE{n}{\phi}$ in~\eqref{eq:TLD_restricted} and multiplying the objective by $\lambda_n$ gives the linear program
\begin{equation}
\label{eq:bootstrap_LP}
Q_{\mathrm{LP},n}^*(\theta)
\equiv
\sup_{\phi,\zeta,u,v}
\left\{\mathbb{G}_{n,\mu^*}^*(\phi)+\lambda_n\left(\EE{n}{\phi}-\zeta\right)\right\}
\quad\text{subject to the constraints of~\eqref{eq:TLD_restricted}}.
\end{equation}
Its value is at least the first term, by the upper bound of Theorem~\ref{thm:enclosure}, whose proof uses $\EE{\mu^*}{\phi}$ only as a linear function of $\phi$.
Together with $T_n(\theta)\ge\underline T_n$, this gives
\begin{equation}
\label{eq:bootstrapub}
\sup_{\phi\in\Phi(\mathcal Z)}\left(\mathbb{G}_{n,\mu^*}^*(\phi)+\lambda_n\eta_{\theta,n}(\phi)\right)-\lambda_nT_n(\theta)
\le
Q_{\mathrm{LP},n}^*(\theta)-\lambda_n\underline T_n.
\end{equation}

\paragraph{Critical value.}
For each of $B$ bootstrap draws, solve~\eqref{eq:bootstrap_LP} on the same $(\mathcal W',\mathcal J',\mathcal K')$.
By~\eqref{eq:bootstrapub}, the empirical $(1-\alpha)$ quantile of the right-hand side over the $B$ draws is at least the empirical $(1-\alpha)$ quantile of the bootstrap statistic.
Rejecting when $\sqrt n\,\underline T_n$ exceeds the former by more than $\varepsilon$ therefore implies rejection by the test of Corollary~\ref{C:inf_finite} computed on the same draws, so the resulting confidence set is conservative.
The computation requires row bounds for the sample program only, and one linear program per bootstrap draw.

\section{Parametric binary choice}
\label{app:computation_example1}

We derive the linear program for $T(\beta)$ under each of the three specifications of Example~\ref{ex:binary_parametric}, together with the row bound where one is needed.
The framework of Section~\ref{sec:computation} is unchanged throughout: the finite program, column certification, and the enclosure of Theorem~\ref{thm:enclosure} use only the LIMR primitives, and what varies by model is the precise LIMR construction and certified row bound.
We also show how the counterfactual target $\tau_{\mathrm{ASF}}\equiv\Pr(\bar x'\beta+A-V_1\ge0)$ enters when the parameter is enlarged to $\theta=(\beta,\tau_{\mathrm{ASF}})$.

\subsection{The baseline specification}
\label{app:joint_parametric}

The baseline fixes $V\mid(X,A)\sim H^{\otimes T}$ and leaves the input measure $\gamma\in\Pa(\mathcal X\times\R)$ free.
Integrating over $V$ gives
$$
q_{\phi,\beta}(x,a)=\EE{K_\beta(\cdot\mid x,a)}{\phi}=\sum_{y\in\{0,1\}^T}\phi(y,x)\prod_{t=1}^T H(x_t'\beta+a)^{y_t}\bigl(1-H(x_t'\beta+a)\bigr)^{1-y_t}
$$
on $\mathcal W=\mathcal X\times\R$, where $K_\beta$ is the output kernel of Example~\ref{ex:binary_parametric}.
For $\theta=\beta$ there are no moment restrictions, so $\mathcal J=\mathcal K=\varnothing$, every $(\mathcal J',\mathcal K')$ is column-certified for any nonempty $\mathcal W'$, and only the row side needs work.
The row residual for $T(\beta)$ reduces to $|\mathcal X|$ one-dimensional maximizations,
$$
a^\star(x)\in\argmax_{a\in\overline\R}q_{\phi,\beta}(x,a),\qquad x\in\mathcal X,
$$
over the extended line $\overline\R=[-\infty,+\infty]$, evaluated at the test function $\phi$ returned by the auxiliary program.
The payoff is continuous on $\overline\R$, so each maximum is attained.
As $a\to\pm\infty$, the output kernel $K_\beta(\cdot\mid x,a)$ converges in total variation on the finite output space.
The extended inputs therefore generate only output measures in the total variation closure, so the extension leaves $T(\beta)$ unchanged.

The payoff is multilinear in the probabilities $\bigl(H(x_t'\beta+a)\bigr)_{t=1}^{T}$, each of which is monotone in the common scalar $a$.
Their common dependence on $a$ leaves a one-dimensional global search.
On a constant covariate path it reduces to a polynomial of degree at most $T$ in a single probability, which is maximized over $[0,1]$ at its stationary points or at the endpoints, the latter corresponding to $a=\pm\infty$.
A certified upper bound $q_{\mathrm{up}}$ on the maximum of the payoff over $\mathcal W$ gives the certified row bound $r_{\mathrm{low}}=\zeta^*-q_{\mathrm{up}}\le r^*$ of Theorem~\ref{thm:enclosure}.
Additional Appendix~\ref{app:parametric_implementation} gives the bound we use and the resulting certificates.

The ASF equals $\EE{\gamma}{H(\bar x'\beta+A)}$.
It therefore enters through the moment restriction $g_{\mathrm{ASF}}(x,a;\theta)\equiv H(\bar x'\beta+a)-\tau_{\mathrm{ASF}}$.
This adds only one unrestricted multiplier, and the row residual remains a one-dimensional search in $a$.

\subsection{Serial dependence}
\label{app:marginal_parametric}

We now fix only the one-period marginals $V_t\mid(X,A)\sim H$ and leave the copula of $(V_1,\dots,V_T)$ free.
The copula reaches the observable only through the conditional distribution $\pi$ of the outcome path $Y\in\{0,1\}^T$.
The marginals restrict $\pi$ to the Fr\'echet polytope
$$
\mathcal F(x,a)\equiv\Bigl\{\pi\in\R_+^{2^T}:\textstyle\sum_y \pi_y=1,\ \sum_{y:\,y_t=1}\pi_y=H(x_t'\beta+a)\ \text{for all }t\Bigr\}.
$$
Every $\pi\in\mathcal F(x,a)$ can be reached: draw $Y$ from $\pi$, then draw each $V_t$ from $H$ truncated below or above its threshold according to $Y_t$.
Each $V_t$ is then marginally distributed as $H$, and the resulting outcome path has distribution $\pi$.
The payoff depends on the copula only through $\pi$ and every $\pi\in\mathcal F(x,a)$ arises from some copula, so replacing the copula by $\pi$ leaves the set of output measures unchanged.

An input value is then $w=(x,a,\pi)$ with $\pi\in\mathcal F(x,a)$.
The specification restricts the support of the input measure rather than imposing moment restrictions, so only the row side has to be analyzed.
At $w$ the model draws the outcome path from $\pi$ and returns the covariate path unchanged, so the payoff is the average of $\phi(\cdot,x)$ under $\pi$, $q_{\phi,\beta}(x,a,\pi)=\sum_y\phi(y,x)\,\pi_y$.

The row residual requires a maximization over the triple, which can be done in two stages.
At each $(x,a)$, the inner maximization over $\pi$ is a small linear program.
At $T=2$, it has a closed form: the optimal $\Pr(Y_1=1,Y_2=1)$ is the Fr\'echet--Hoeffding bound according to the sign of $\phi((1,1),x)-\phi((1,0),x)-\phi((0,1),x)+\phi((0,0),x)$.
Substituting its value leaves a one-dimensional global search in $a$ at each covariate path.
Additional Appendix~\ref{app:parametric_implementation} gives the details.

For each $(x,a)$, dualizing the inner linear program replaces the continuum of rows indexed by $\pi\in\mathcal F(x,a)$ by the finite block
\begin{align*}
\sup_{\pi\in\mathcal F(x,a)}\sum_y\phi(y,x)\pi_y
=\min_{\xi_0,\xi}\Bigl\{\ &\xi_0+\sum_{t=1}^T\xi_tH(x_t'\beta+a):\\
&\xi_0+\sum_{t=1}^T\xi_ty_t\ge\phi(y,x)\quad\text{for all }y\Bigr\}.
\end{align*}
The auxiliary program carries one vector $(\xi_0,\xi)$ for each of the finitely many $(x,a)$ in the retained support $\mathcal W'$, subject to the constraints of the display and to $\xi_0+\sum_{t=1}^T\xi_tH(x_t'\beta+a)\le\zeta$, and remains linear.
The counterfactual enters exactly as under the baseline (Subsection~\ref{app:joint_parametric}) because $V_1\mid(A,X)\sim H$ still holds, without changing the inner maximization over $\pi$.

\subsection{Fixed effect--error dependence}
\label{app:Xconditional_parametric}

Fix $V\mid X\sim H^{\otimes T}$ and leave the kernel of $A$ given $(X,V)$ free, so $A$ and $V$ may be dependent given $X$.
Fix $(x,v)$ and let $a$ increase from $-\infty$ to $+\infty$: each outcome $Y_t=\one\{a\ge v_t-x_t'\beta\}$ switches from $0$ to $1$ once, when $a$ crosses its threshold.
Write $\sigma$ for the permutation that sorts the thresholds.
When the thresholds are distinct, increasing $a$ recovers exactly $T+1$ outcome paths $y^{(0)}_\sigma,\dots,y^{(T)}_\sigma$, where $y^{(k)}_\sigma$ sets coordinates $\sigma(1),\dots,\sigma(k)$ to one and the rest to zero.
The supremum over the free $A$ is $\phi^\star_\sigma(x)\equiv\max_{0\le k\le T}\phi(y^{(k)}_\sigma,x)$, which depends on $v$ only through the order $\sigma$.

Because $H$ is continuous, ties occur with probability zero, so $\sigma$ ranges over the $T!$ permutations $\mathcal S_T$ of $\{1,\dots,T\}$, with probabilities $p_\sigma(x;\beta)$ fixed by the specification.
Integrating $\phi^\star_\sigma(x)$ and maximizing over the covariate marginal gives
\begin{align*}
\sup_{\gamma\in\Gamma_\beta}\EE{\Lbeta\gamma}{\phi}
&=\max_{x\in\mathcal X}\sum_{\sigma\in \mathcal S_T}p_\sigma(x;\beta)\,\phi^\star_\sigma(x),\\
p_\sigma(x;\beta)
&\equiv\Pr_{V\sim H^{\otimes T}}\bigl(V_{\sigma(1)}-x_{\sigma(1)}'\beta<\cdots<V_{\sigma(T)}-x_{\sigma(T)}'\beta\bigr),
\end{align*}
with the probabilities $p_\sigma$ precomputed once per $\beta$.
The payoff depends on the input only through the covariate path, the order $\sigma$, and which of the $T+1$ regions of the sweep $a$ falls in.
The input space reduces to the finite set $\mathcal X\times \mathcal S_T\times\{0,1,\ldots,T\}$, on which the specification pins $\Pr(\sigma\mid X=x)=p_\sigma(x;\beta)$ and leaves the covariate marginal and the region unrestricted.

Introduce a scalar $\zeta$ for the maximum over $x$ and slacks $\{s_\sigma(x)\}$ for the maxima over $k$.
Constraining each variable to lie above every term in the corresponding maximum gives the finite LP
\begin{align*}
T(\beta)=\max_{\phi,\,s,\,\zeta}\quad & \sum_{(y,x)}\mu^*(y,x)\,\phi(y,x)-\zeta \\
\text{subject to}\quad & \phi(y,x)\in[0,1] && \text{for all }(y,x),\\
& s_\sigma(x)\ge\phi(y^{(k)}_\sigma,x) && \text{for all }\sigma,x,k,\\
& \zeta\ge\sum_{\sigma\in \mathcal S_T}p_\sigma(x;\beta)\,s_\sigma(x) && \text{for all }x.
\end{align*}
This finite program computes $T(\beta)$, one solve per $\beta$, with neither a row bound nor a column check.
The counterfactual enters through the moment restriction $g_{\mathrm{ASF}}(x,a,v;\theta)\equiv\one\{\bar x'\beta+a-v_1\ge0\}-\tau_{\mathrm{ASF}}$.
Its indicator switches at the threshold $v_1-\bar x'\beta$, which joins the $T$ outcome thresholds in the sweep over $a$, so the reduced input space remains finite and the joint computation again needs neither a row bound nor a column check.
Additional Appendix~\ref{app:parametric_reductions} describes the computation of the reported identified sets.

\section{Sequential exogeneity}
\label{app:semiparametric_binary_choice}

This appendix constructs the LIMR behind the sequential exogeneity results of Section~\ref{subsec:numerical_semiparametric} (Example~\ref{ex:binary_sequential}), reduces the error terms to finitely many cells, and derives the auxiliary program with its column certification and its row bound.
The framework of Section~\ref{sec:computation} applies unchanged, and Additional Appendix~\ref{app:se_details} carries the remaining results and the computational details.

\subsection{Constructing the LIMR}
\label{app:se_lifting}

As in Example~\ref{ex:binary_sequential}, the covariate path $X=(X_1,\ldots,X_T)$ takes values in a finite set $\mathcal X$, $X^t\equiv(X_1,\ldots,X_t)$ is its history through period $t$, $\tilde V_t\equiv A-V_t$ is the composite error, $\nu\equiv(\Pr(X=x\mid A))_{x\in\mathcal X}$ takes values in the simplex $\Delta\equiv\{\nu\in\R^{\mathcal X}_+:\sum_{x\in\mathcal X}\nu(x)=1\}$, and $\nu_t(x^t)$ is the probability of the history $x^t$ under $\nu$.
Because $A$ is in the conditioning set on both sides, sequential exogeneity~\eqref{eq:seq_exog} holds if and only if it holds with the composite error in place of $V_t$, that is, $\tilde V_t\mid(A,X^t)\overset{d}{=}\tilde V_1\mid(A,X_1)$ for $t=2,\ldots,T$.
Example~\ref{ex:binary_sequential} has input $W=(X,\tilde V,\nu)$ on $\mathcal W\equiv\mathcal X\times\R^T\times\Delta$ and admissible set
$$
\Gamma_\beta
\equiv
\Bigl\{\gamma\in\Pa(\mathcal W):\ \gamma\ \text{satisfies~\eqref{eq:nu_consistency} and~\eqref{eq:seq_lifted}}\Bigr\},
$$
where, for all $x\in\mathcal X$, all $t\in\{2,\ldots,T\}$, all histories $x^t$, and all bounded measurable $r$ and $f$,
\begin{align}
&\EE{\gamma}{r(\nu)\bigl(\one\{X=x\}-\nu(x)\bigr)}=0,
\tag{\ref*{eq:nu_consistency}}
\\
&\EE{\gamma}{r(\nu)\bigl[f(\tilde V_t)\one\{X^t=x^t\}\nu_1(x_1)-f(\tilde V_1)\one\{X_1=x_1\}\nu_t(x^t)\bigr]}=0.
\tag{\ref*{eq:seq_lifted}}
\end{align}

Restrictions~\eqref{eq:nu_consistency} and~\eqref{eq:seq_lifted} hold for all bounded measurable $r$ and $f$ if and only if, almost surely,
\begin{equation}
\label{eq:se_conditional_form}
\begin{aligned}
\Pr(X=x\mid\nu)&=\nu(x)\quad\text{for all } x\in\mathcal X, \\
\tilde V_t\mid(\nu,X^t=x^t) &\overset{d}{=} \tilde V_1\mid(\nu,X_1=x_1)\ \ \text{for all histories } x^t \text{ with } \nu_t(x^t)>0,\ t=2,\ldots,T.
\end{aligned}
\end{equation}
Because $r$ ranges over all bounded measurable functions of $\nu$, each restriction is a conditional moment equality given $\nu$.
Restriction~\eqref{eq:nu_consistency} gives $\Pr(X^t=x^t\mid\nu)=\nu_t(x^t)$.
When $\nu_t(x^t)>0$, the period-1 history also has positive probability, and dividing the two terms in~\eqref{eq:seq_lifted} by $\nu_1(x_1)\nu_t(x^t)$ gives $\EE{}{f(\tilde V_t)\mid\nu,X^t=x^t}=\EE{}{f(\tilde V_1)\mid\nu,X_1=x_1}$ for every bounded $f$, that is, the two conditional distributions are equal (indicators of half-lines with rational endpoints already determine a distribution on $\R$, so one null set serves every $f$).
When $\nu_t(x^t)=0$, restriction~\eqref{eq:nu_consistency} already forces $\Pr(X^t=x^t\mid\nu)=0$, so both terms in~\eqref{eq:seq_lifted} vanish.
The converse reverses these steps, multiplying the two conditional expectations by the history probabilities.

\begin{lemma}
\label{lem:nu_lifting}
Let $\mathcal X$ be finite.
\begin{enumerate}
\item[(a)] If the distribution of $(X,V,A)$ satisfies~\eqref{eq:seq_exog}, then the distribution of $(X,\tilde V,\nu)$ with $\nu\equiv\Pr(X=\cdot\mid A)$ belongs to $\Gamma_\beta$.
\item[(b)] If the distribution of $(X,\tilde V,\nu)$ belongs to $\Gamma_\beta$, then there is a measurable function $A$ of $\nu$ such that, with $V_t\equiv A-\tilde V_t$, the distribution of $(X,V,A)$ satisfies~\eqref{eq:seq_exog} and $\Pr(X=\cdot\mid A)=\nu$ almost surely.
\end{enumerate}
\end{lemma}

Because the output $Z$ is a function of $(X,\tilde V)$ alone, parts~(a) and~(b) match the two representations distribution by distribution.
Hence $\Lbeta\Gamma_\beta$ is the set of probability measures of $Z$ that the model of Example~\ref{ex:binary_sequential} generates at $\beta$, which is the adequacy requirement in the definition of a LIMR in Section~\ref{sec:model}.

\begin{proof}
\emph{Part~(a).}
Let the distribution of $(X,V,A)$ satisfy~\eqref{eq:seq_exog} and set $\nu\equiv\Pr(X=\cdot\mid A)$, a function of $A$ with values in $\Delta$.
Because $r(\nu)$ is a function of $A$ and $\Pr(X=x\mid A)=\nu(x)$, the law of iterated expectations gives~\eqref{eq:nu_consistency} for every bounded measurable $r$.
For~\eqref{eq:seq_lifted}, fix $t\in\{2,\ldots,T\}$, a history $x^t$, and a bounded measurable $f$, and recall that $\Pr(X^t=x^t\mid A)=\nu_t(x^t)$.
By the law of iterated expectations,
\begin{align*}
\EE{}{f(\tilde V_t)\one\{X^t=x^t\}\mid A}
&= \EE{}{\EE{}{f(\tilde V_t)\mid A,X^t}\one\{X^t=x^t\}\mid A} \\
&= \EE{}{f(\tilde V_t)\mid A,X^t=x^t}\,\nu_t(x^t) \\
&= \EE{}{f(\tilde V_1)\mid A,X_1=x_1}\,\nu_t(x^t).
\end{align*}
The last equality is sequential exogeneity~\eqref{eq:seq_exog}, stated for the composite error.
When $\nu_t(x^t)=0$, both sides are zero, whatever value the conditional expectation given a history of probability zero takes.
At $t=1$, the same computation, without sequential exogeneity, gives
$$
\EE{}{f(\tilde V_1)\one\{X_1=x_1\}\mid A}
=\EE{}{f(\tilde V_1)\mid A,X_1=x_1}\,\nu_1(x_1).
$$
Multiply the first display by $\nu_1(x_1)$ and the second by $\nu_t(x^t)$.
The right-hand sides coincide, so the term in brackets in~\eqref{eq:seq_lifted} has conditional mean zero given $A$, and because $r(\nu)$ is a function of $A$, the law of iterated expectations gives~\eqref{eq:seq_lifted}.

\emph{Part~(b).}
Let $(X,\tilde V,\nu)$ have a distribution in $\Gamma_\beta$.
By~\eqref{eq:se_conditional_form}, $\nu$ plays the role of the fixed effect, but the model's fixed effect is a real number.
Because $\Delta$ is a Borel subset of $\R^{\mathcal X}$, it is Borel isomorphic to a Borel subset of $\R$, which provides a one-to-one measurable $\iota\colon\Delta\to\R$ with measurable inverse.
Set $A\equiv\iota(\nu)$ and $V_t\equiv A-\tilde V_t$.
Since $A$ and $\nu$ determine each other, conditioning on one is conditioning on the other, and~\eqref{eq:se_conditional_form} becomes $\Pr(X=\cdot\mid A)=\nu$ and $\tilde V_t\mid(A,X^t)\overset{d}{=}\tilde V_1\mid(A,X_1)$ on every history with positive probability given $A$.
Since $A$ is in the conditioning set, the same holds for $V_t=A-\tilde V_t$, which is~\eqref{eq:seq_exog}.
\end{proof}

\subsection{Finite reduction and closedness}
\label{app:se_cell_reduction}

The observable space $\mathcal Z\equiv\{0,1\}^T\times\mathcal X$ is finite, so Assumption~\ref{asm:finite_Z} holds, and Assumption~\ref{asm:domination} holds with counting measure.
The outcome equation is $Y_t=\one\{\tilde V_t\ge c(t,X_t;\beta)\}$ with thresholds $c(t,x_t;\beta)\equiv-x_t'\beta$, where $x_t$ may include period-specific components with known coefficients, such as the time trend of Section~\ref{subsec:numerical_semiparametric}, and the ASF at a counterfactual covariate value $\bar x$ is $\tau_{\mathrm{ASF}}=\Pr(\tilde V_1\ge\bar c)$ with $\bar c\equiv-\bar x'\beta$.
Collect the thresholds in $\mathcal C(\beta)\equiv\bigl\{c(t,x_t;\beta):t\in\{1,\ldots,T\},\ x\in\mathcal X\bigr\}\cup\{\bar c\}$, sort its distinct elements as $c_{(1)}<\cdots<c_{(M)}$, and partition the real line into the cells $C_1\equiv(-\infty,c_{(1)})$, $C_j\equiv[c_{(j-1)},c_{(j)})$ for $j=2,\ldots,M$, and $C_{M+1}\equiv[c_{(M)},+\infty)$.
Fix a point $m_j$ in each cell $C_j$.
Because the cells are cut at the thresholds, $\one\{\tilde V_t\ge c\}=\one\{m_j\ge c\}$ whenever $\tilde V_t\in C_j$ and $c\in\mathcal C(\beta)$.
A cell profile $k=(k_1,\ldots,k_T)\in\{1,\ldots,M+1\}^T$ lists the cell of each period's error term, and we write $k$ also for the random cell profile of $\tilde V$.
Hence, when $\tilde V_t\in C_{k_t}$ for every $t$, the outcome path, the output, and the counterfactual outcome are $Y=y(X,k;\beta)$, $Z=z(X,k;\beta)$, and $\one\{\tilde V_1\ge\bar c\}=\bar y(k;\beta)$.
Here $y(x,k;\beta)\equiv\bigl(\one\{m_{k_t}\ge c(t,x_t;\beta)\}\bigr)_{t=1}^T$, $z(x,k;\beta)\equiv(y(x,k;\beta),x)$, and $\bar y(k;\beta)\equiv\one\{m_{k_1}\ge\bar c\}$.

For $\nu\in\Delta$, call a distribution of $(X,\tilde V)$ on $\mathcal X\times\R^T$ admissible at $\nu$ if its covariate marginal is $\nu$ and
\begin{equation}
\label{eq:se_fiber_se}
\tilde V_t\mid X^t=x^t
\overset{d}{=}
\tilde V_1\mid X_1=x_1
\quad\text{for all } t\ge2 \text{ and all } x^t \text{ with } \Pr(X^t=x^t)>0.
\end{equation}
By~\eqref{eq:se_conditional_form}, an input measure belongs to $\Gamma_\beta$ if and only if, for almost every $\nu$, the conditional distribution of $(X,\tilde V)$ given $\nu$ is admissible at $\nu$.
The cell distribution of $(X,\tilde V)$ is $\pi(x,k)\equiv\Pr(X=x,\ \tilde V_t\in C_{k_t}\text{ for all }t)$, a probability mass function on $\mathcal X\times\{1,\ldots,M+1\}^T$.
Let $\Pi(\nu)$ be the set of probability mass functions $\pi$ on $\mathcal X\times\{1,\ldots,M+1\}^T$ that satisfy~\eqref{eq:se_sc_cells} for all $x\in\mathcal X$ and~\eqref{eq:se_se_cells} for all $t\in\{2,\ldots,T\}$, all histories $x^t=(x_1,\ldots,x_t)$, and all cell indices $j\in\{1,\ldots,M+1\}$,
\begin{align}
&\sum_{k}\pi(x,k)=\nu(x),
\label{eq:se_sc_cells}
\\
&\nu_1(x_1)\sum_{x':\,(x'_1,\ldots,x'_t)=x^t}\ \sum_{k:\,k_t=j}\pi(x',k)
=
\nu_t(x^t)\sum_{x':\,x'_1=x_1}\ \sum_{k:\,k_1=j}\pi(x',k).
\label{eq:se_se_cells}
\end{align}
Restriction~\eqref{eq:se_sc_cells} is the covariate marginal, and~\eqref{eq:se_se_cells} is~\eqref{eq:se_fiber_se} on cells, multiplied through by the history probabilities so that it is linear in $\pi$: on histories with positive probability, the cell of the period-$t$ error term given $x^t$ is distributed as the cell of the period-1 error term given $x_1$.

\begin{lemma}[Cell reduction]
\label{lem:se_cell_reduction}
Fix $\nu\in\Delta$.
Then $\pi\in\Pi(\nu)$ if and only if $\pi$ is the cell distribution of some distribution admissible at $\nu$.
\end{lemma}

\begin{proof}
\emph{$\Leftarrow$.}
Let $(X,\tilde V)$ be admissible at $\nu$ and let $\pi$ be its cell distribution.
Restriction~\eqref{eq:se_sc_cells} is the covariate marginal.
For~\eqref{eq:se_se_cells}, fix $t\ge2$, $x^t$, and $j$.
If $\nu_t(x^t)=0$, both sides are zero.
Otherwise, dividing by $\nu_1(x_1)\nu_t(x^t)$ turns~\eqref{eq:se_se_cells} into $\Pr(\tilde V_t\in C_j\mid X^t=x^t)=\Pr(\tilde V_1\in C_j\mid X_1=x_1)$, which is~\eqref{eq:se_fiber_se} evaluated on the cell $C_j$.

\emph{$\Rightarrow$.}
Let $\pi\in\Pi(\nu)$.
For $x_1$ with $\nu_1(x_1)>0$, let $\varrho_{x_1}(j)\equiv\Pr_\pi(k_1=j\mid X_1=x_1)$, the cell distribution of the period-1 error term given the initial covariate value.
Dividing~\eqref{eq:se_se_cells} by $\nu_1(x_1)\nu_t(x^t)$ and using~\eqref{eq:se_sc_cells} shows that $\Pr_\pi(k_t=j\mid X^t=x^t)=\varrho_{x_1}(j)$ on every history with $\nu_t(x^t)>0$.
Now draw $(X,k)$ from $\pi$ and set $\tilde V_t\equiv m_{k_t}$, the point chosen in cell $C_{k_t}$.
Then $\tilde V_t\in C_{k_t}$, so the cell distribution of $(X,\tilde V)$ is $\pi$, and the covariate marginal is $\nu$ by~\eqref{eq:se_sc_cells}.
Given any history $x^t$ with $\nu_t(x^t)>0$, the error term $\tilde V_t$ equals $m_j$ with probability $\varrho_{x_1}(j)$, so its conditional distribution is $\sum_j\varrho_{x_1}(j)\,\delta_{m_j}$ for every $t$ (at $t=1$ by the definition of $\varrho_{x_1}$), which is~\eqref{eq:se_fiber_se}, so $(X,\tilde V)$ is admissible at $\nu$.
\end{proof}

Each $\pi\in\Pi(\nu)$ determines an output measure $\mu_\pi(y,x)\equiv\sum_{k:\,y(x,k;\beta)=y}\pi(x,k)$, and by Lemma~\ref{lem:se_cell_reduction} the set $\{\mu_\pi:\pi\in\Pi(\nu)\}$ is exactly the set of output measures the model generates at $\beta$ conditional on $\nu$.

\begin{prop}[Finite support in $\nu$ and closedness]
\label{prop:se_finite_support}
Let $\mathcal M_\beta\equiv\Lbeta\Gamma_\beta$ be the set of output measures at $\beta$, and, at $\theta=(\beta,\tau_{\mathrm{ASF}})$, let $\Gamma_\theta$ refine $\Gamma_\beta$ by the restriction $\EE{\gamma}{\one\{\tilde V_1\ge\bar c\}-\tau_{\mathrm{ASF}}}=0$ of Example~\ref{ex:binary_sequential}.
\begin{enumerate}
\item[(a)] The set $\mathcal M_\beta$ is the convex hull of $\{\mu_\pi:\nu\in\Delta,\ \pi\in\Pi(\nu)\}$, and every member is generated by an input measure under which $\nu$ takes at most $|\mathcal Z|$ values.
\item[(b)] The set $\{\mu_\pi:\nu\in\Delta,\ \pi\in\Pi(\nu)\}$ is compact, and so is $\mathcal M_\beta$.
\item[(c)] The same conclusions hold, with $|\mathcal Z|+1$ in place of $|\mathcal Z|$, for the pairs $\bigl(\mu_\pi,\sum_{x,k}\pi(x,k)\,\bar y(k;\beta)\bigr)\in\R^{|\mathcal Z|+1}$, whose second coordinate is the ASF at fixed $\nu$ because $\bar c\in\mathcal C(\beta)$, so $\mathcal M_\theta\equiv\Lgen\Gamma_\theta$ is compact and convex, with every member generated by an input measure under which $\nu$ takes at most $|\mathcal Z|+1$ values.
\end{enumerate}
\end{prop}
\begin{proof}
\emph{Conditioning on $\nu$.}
Write $\gamma\in\Gamma_\beta$ as $\gamma_\nu(\d x,\d\tilde v)\,\lambda(\d\nu)$, with $\lambda$ the marginal distribution of $\nu$ and $\gamma_\nu$ the conditional distribution of $(X,\tilde V)$ given $\nu$. The disintegration exists because $\mathcal X\times\R^T\times\Delta$ is Polish.
By~\eqref{eq:se_conditional_form}, for $\lambda$-almost every $\nu$ the conditional distribution $\gamma_\nu$ is admissible at $\nu$, one null set serving all $t$, $x^t$, and $f$ as above, because there are finitely many histories.
By Lemma~\ref{lem:se_cell_reduction}, the cell distribution $\pi_\nu$ of $\gamma_\nu$ lies in $\Pi(\nu)$, the output measure of $\gamma_\nu$ is $\mu_{\pi_\nu}$, and $\Lbeta\gamma=\int_\Delta\mu_{\pi_\nu}\,\lambda(\d\nu)$.
Conversely, let $\pi^i\in\Pi(\nu^i)$ and let $\gamma^i$ be the distribution of $(X,\tilde V)$ that Lemma~\ref{lem:se_cell_reduction} provides, for $i=1,\ldots,m$.
Under $\gamma^i\otimes\delta_{\nu^i}$ the conditional form~\eqref{eq:se_conditional_form} holds, so $\gamma^i\otimes\delta_{\nu^i}\in\Gamma_\beta$, and because the restrictions defining $\Gamma_\beta$ are linear in $\gamma$, every mixture $\gamma\equiv\sum_iw_i\,(\gamma^i\otimes\delta_{\nu^i})$ belongs to $\Gamma_\beta$ and has $\Lbeta\gamma=\sum_iw_i\mu_{\pi^i}$.

\emph{Parts~(a) and~(b).}
The set $G\equiv\{(\nu,\pi):\nu\in\Delta,\ \pi\in\Pi(\nu)\}$ is closed in the compact set $\Delta\times\Pa(\mathcal X\times\{1,\ldots,M+1\}^T)$, because~\eqref{eq:se_sc_cells} and~\eqref{eq:se_se_cells} are continuous in $(\nu,\pi)$, so $G$ is compact, and so is $\{\mu_\pi:(\nu,\pi)\in G\}$, because $\mu_\pi$ is continuous in $\pi$.
The convex hull of a compact subset of $\R^{|\mathcal Z|}$ is compact, so the convex hull of this set is closed and contains every integral $\int_\Delta\mu_{\pi_\nu}\,\lambda(\d\nu)$ (the mean of a random vector with values in a closed convex set lies in that set).
Each of its points is a finite mixture $\sum_iw_i\mu_{\pi^i}$, which the first paragraph realizes as $\Lbeta\gamma$ with $\gamma\in\Gamma_\beta$, and Carath\'eodory's theorem in the probability simplex, which has dimension $|\mathcal Z|-1$, allows $m\le|\mathcal Z|$ points, so $\nu$ takes at most $|\mathcal Z|$ values under the mixture.

\emph{Part~(c).}
Because $\bar c\in\mathcal C(\beta)$, no cell straddles $\bar c$, so conditional on $\nu$, with cell distribution $\pi$, the ASF equals $\sum_{x,k}\pi(x,k)\,\bar y(k;\beta)$.
Applying the preceding argument to $(\nu,\pi)\mapsto\bigl(\mu_\pi,\sum_{x,k}\pi(x,k)\,\bar y(k;\beta)\bigr)$ gives the same conclusions, with Carath\'eodory's theorem in the affine hull of these pairs, which has dimension at most $|\mathcal Z|$.
The set of pairs $\bigl(\Lbeta\gamma,\EE{\gamma}{\one\{\tilde V_1\ge\bar c\}}\bigr)$ over $\gamma\in\Gamma_\beta$ is therefore compact and convex, with each pair realized under at most $|\mathcal Z|+1$ values of $\nu$, and $\mathcal M_\theta$ is the set of first coordinates of the pairs whose second coordinate is $\tau_{\mathrm{ASF}}$, hence compact and convex.
\end{proof}

\subsection{The linear program and its certification}
\label{app:se_exact_pricing}

After the cell reduction, Example~\ref{ex:binary_sequential} has input $w=(x,k,\nu)$ on $\mathcal X\times\{1,\ldots,M+1\}^T\times\Delta$, and we write $\Gamma_\beta$ for the admissible set in these coordinates.
By Lemma~\ref{lem:se_cell_reduction} and Proposition~\ref{prop:se_finite_support}(a), this LIMR has the same output measures $\mathcal M_\beta$, so $T(\beta)$ is unchanged.
The output is the function $z(x,k;\beta)$ of the input, so the input-space payoff is $q_{\phi,\beta}(x,k,\nu)=\phi\bigl(z(x,k;\beta)\bigr)$ and does not depend on $\nu$.
There are no inequality restrictions.
The equality restrictions are the conditional moment equalities given $\nu$ in their cell form~\eqref{eq:se_sc_cells} and~\eqref{eq:se_se_cells}: conditionally on $\nu$, the distribution of $(X,k)$ lies in $\Pi(\nu)$.

Fix a finite atom list $N=\{\nu^1,\ldots,\nu^m\}\subset\Delta$.
The covariate and cell coordinates are already finite, so restricting $\nu$ to $N$ makes the input space finite.
Take $\mathcal W'\equiv\mathcal X\times\{1,\ldots,M+1\}^T\times N$.
For each atom $\nu^i$ in the list, impose~\eqref{eq:nu_consistency} and~\eqref{eq:seq_lifted} at $r=\one\{\nu=\nu^i\}$, in cell form.
The weights at $\nu^i$ satisfy~\eqref{eq:se_se_cells} as written and~\eqref{eq:se_sc_cells} with $\nu^i(x)$ multiplied by the mass the input measure gives that atom.
These equalities are linear in $f$, and every function of the cell index is a combination of the cell indicators, so imposing them at the $M+1$ indicators imposes the whole family, and $\mathcal J'$ is finite.
The auxiliary finite linear program is~\eqref{eq:TLD_restricted} at these selections:
\begin{equation}
\label{eq:se_inner_lp}
T_{\mathrm{LP}}(\beta;N)
\equiv
\TLD{\beta}{\mathcal W'}{\mathcal J',\varnothing}.
\end{equation}

Column certification needs no computation here.
An input measure on $\mathcal W'$ gives $\nu$ only the values in the list, so a function $r$ of $\nu$ matters only through its values there.
The error terms may be taken independent of $X$ and identically distributed across periods, so $\Pi(\nu)$ is never empty.
Putting all the mass on a single atom then gives weights that satisfy every equality, and the program is feasible.
Both requirements of Definition~\ref{def:column_certified} therefore hold, and Theorem~\ref{thm:enclosure} gives $T(\beta)\le T_{\mathrm{LP}}(\beta;N)$.
The value of the program is the minimum of $\sum_{z\in\mathcal Z}\bigl(\mu^*(z)-\mu(z)\bigr)_+$ over the output measures it can reach (as in Remark~\ref{rem:ipm}).
The sum is zero only at $\mu^*$, so the value is zero exactly when some allowed weights generate $\mu^*$, and then $\beta$ is in the identified set.
Conversely, if $\beta$ is in the identified set, $\mathcal M_\beta$ is closed by Proposition~\ref{prop:se_finite_support}(b), so some admissible input measure generates $\mu^*$.
By part~(a), one such measure puts $\nu$ on at most $|\mathcal Z|$ values, and the list of those values gives value zero.

Let $(\phi^*,\zeta^*)$ be the test function and the envelope at an optimal solution of~\eqref{eq:se_inner_lp}, in the notation of Section~\ref{sec:computation}.
Then $T_{\mathrm{LP}}(\beta;N)=\inner{\phi^*,\mu^*}-\zeta^*$, and at every atom $\nu^i$ in the list, $\inner{\phi^*,\mu_\pi}\le\zeta^*$ for every $\pi\in\Pi(\nu^i)$.
We check that inequality at every $\nu\in\Delta$ by solving
\begin{equation}
\label{eq:se_pricing}
\delta^*
\equiv
\max_{\nu\in\Delta,\ \pi\in\Pi(\nu)}\ \Bigl\{\inner{\phi^*,\mu_\pi}-\zeta^*\Bigr\}.
\end{equation}
The maximum is attained (Proposition~\ref{prop:se_finite_support}(b)).
Let $p^*$ be optimal weights of the dual of the program, one per row of $\mathcal W'$, as in the proof of Theorem~\ref{thm:enclosure}.
Some atom $\nu^i$ carries positive mass under $p^*$, and rescaling its weights to sum to one leaves the cell equalities intact, so the result is a $\pi\in\Pi(\nu^i)$.
Every row with positive weight is tight at the optimum, so the maximand is zero at that $\pi$.
The maximum is therefore never negative.
Every $\gamma\in\Gamma_\beta$ satisfies those equalities, so the slack of~\eqref{eq:slack} integrates to $\EE{\gamma}{s(W)}=\zeta^*-\inner{\phi^*,\Lbeta\gamma}$.
By Proposition~\ref{prop:se_finite_support}(a), the $\Lbeta\gamma$ are exactly the mixtures of the $\mu_\pi$.
Hence $\inf_{\gamma\in\Gamma_\beta}\EE{\gamma}{s(W)}=-\delta^*$, which is the row residual of Remark~\ref{rem:row_measures}, an infimum over admissible input measures rather than over input values.
Any $\delta_{\mathrm{ub}}\ge\delta^*$ therefore serves as $r_{\mathrm{low}}=-\delta_{\mathrm{ub}}$ in Theorem~\ref{thm:enclosure}, read in the form of Remark~\ref{rem:row_measures}, and gives the enclosure
$$
\max\{0,\ T_{\mathrm{LP}}(\beta;N)-\delta_{\mathrm{ub}}\}\ \le\ T(\beta)\ \le\ T_{\mathrm{LP}}(\beta;N).
$$

At every iteration of column-and-row generation, if the lower end of the enclosure is positive, then $T(\beta)>0$ and $\beta$ is outside the identified set, with nothing further to compute.
If $\delta^*>0$, the maximizing $\nu$ is a value to add, and it is not already in the list, since the maximand is at most zero there.
We solve~\eqref{eq:se_pricing} in the smaller exact parameterization of Additional Appendix~\ref{app:se_history_reduction}, to a bound that is guaranteed to be at least $\delta^*$ rather than to a local maximum.

\section{Additional proofs}
\label{sec:app-additional-proofs}

This appendix proves Corollary~\ref{C:inf_finite} of Section~\ref{sec:inference}.
We use the notation of the proof of Proposition~\ref{P:inf_finite} in Appendix~\ref{app:proofs}, in particular $\Phi\equiv\Phi(\mathcal Z)$.

\begin{proof}[Proof of Corollary~\ref{C:inf_finite}]
We establish the coverage property~\eqref{E:coverage} and then the consistency property~\eqref{E:power}.

First, define the functional
$$
\Psi_{n,\mu^*,\theta}:
G
\mapsto
\sup_{\phi\in\Phi}
\left(
G(\phi)
+
\lambda_n
\eta_{\theta,\mu^*}(\phi)
\right)
-
\lambda_n
T_{\mu^*}(\theta)
$$
on $\ell^\infty(\Phi)$.
Define also
$$
B_{n,\mu^*}(\theta)
=
\sup_{\phi\in\Phi}
\left(
\mathbb{G}_{n,\mu^*}(\phi)
+
\lambda_n
\eta_{\theta,n}(\phi)
\right)
-
\lambda_n
T_n(\theta),
$$
and $B_{n,\mu^*}^*(\theta)$ as its bootstrap analog with $\mathbb{G}_{n,\mu^*}$ replaced by $\mathbb{G}_{n,\mu^*}^*$.
Because $\Psi_{n,\mu^*,\theta}$ is Lipschitz for all $n$, $\mu^*$, and $\theta$, Lemma~A.2 of \textcite{linton2010} implies that
$$
\Psi_{n,\mu^*,\theta}
(\mathbb{G}_{n,\mu^*}^*)
\overset{\mu^*}{\Rightarrow}
\Psi_{n,\mu^*,\theta}
(\mathbb{G}_{\mu^*})
$$
uniformly in $\theta$ and $\mu^*$, and likewise
$$
\Psi_{n,\mu^*,\theta}
(\mathbb{G}_{n,\mu^*})
\Rightarrow
\Psi_{n,\mu^*,\theta}
(\mathbb{G}_{\mu^*})
$$
uniformly.

Define $\Delta_{n,\mu^*}\equiv\sup_{\phi\in\Phi}\left|\EE{n}{\phi}-\EE{\mu^*}{\phi}\right|$.
By the exact cancellation in~\eqref{E:etaconv}, $\sup_{\phi\in\Phi}|\eta_{\theta,n}(\phi)-\eta_{\theta,\mu^*}(\phi)|=\Delta_{n,\mu^*}$ for every $\theta\in\Theta^{\mathrm{feas}}\equiv\{\theta\in\Theta:\Gamma_\theta\neq\varnothing\}$, and the right-hand side does not depend on $\theta$.
Moreover, $|T_n(\theta)-T_{\mu^*}(\theta)|\le\Delta_{n,\mu^*}$.
Hence replacing $\eta_{\theta,\mu^*}$ by $\eta_{\theta,n}$ and $T_{\mu^*}(\theta)$ by $T_n(\theta)$ changes the value of the functional by at most $2\lambda_n\Delta_{n,\mu^*}$.
Since $\Delta_{n,\mu^*}=O_p(n^{-1/2})$ uniformly over $\mu^*\in\mathbf P$ and $\lambda_n=o(\sqrt n)$, we have $2\lambda_n\Delta_{n,\mu^*}=o_p(1)$ uniformly over $\mu^*\in\mathbf P$ and $\theta\in\Theta^{\mathrm{feas}}$.
The preceding uniform convergences therefore remain valid with the empirical penalty and centering; in particular, $B_{n,\mu^*}^*(\theta)\overset{\mu^*}{\Rightarrow}\Psi_{n,\mu^*,\theta}(\mathbb G_{\mu^*})$ uniformly.

For any constant $r$, let $h_r:\R\ra[0,1]$ be identically $1$ for $t\le r$, $0$ for $t\ge r+\ve/2$, and linear between $r$ and $r+\ve/2$.
The functions $h_r$ are Lipschitz uniformly in $r$.
Let $d\in(0,1-\alpha)$.
By bounded Lipschitz convergence,
\begin{align}
\label{E:boundedlipcor}
\liminf_{n\ra\infty}
\inf_{\substack{\mu^*\in\mathbf P\\\theta\in\Theta}}
\PP{\mu^*}{
\sup_{r\in\R}
\left|
\mathbb E^*[h_r(B^*_{n,\mu^*}(\theta))]
-
\EE{\mu^*}{
h_r(
\Psi_{n,\mu^*,\theta}(\mathbb G_{\mu^*})
)
}
\right|
\le d
}
=
1.
\end{align}

Let $F_{n,\mu^*,\theta}$ denote the distribution function of $\Psi_{n,\mu^*,\theta}(\mathbb G_{\mu^*})$.
On the event above,
\begin{align}
F_{n,\mu^*,\theta}
\big(
\hat c_{1-\alpha}(\theta)+\ve/2
\big)
&\ge
\int
h_{\hat c_{1-\alpha}(\theta)}
\,\d F_{n,\mu^*,\theta}
\nonumber\\
&\ge
\mathbb E^*[
h_{\hat c_{1-\alpha}(\theta)}
(
B_{n,\mu^*}^*(\theta)
)
]
-d
\nonumber\\
&\ge
1-\alpha-d,
\label{E:conlb1}
\end{align}
and so
$$
\hat c_{1-\alpha}(\theta)+\ve/2
\ge
q_{n,\mu^*,1-\alpha-d}(\theta),
$$
where $q_{n,\mu^*,1-\alpha-d}(\theta)$ is the $(1-\alpha-d)^{\text{th}}$ quantile of $\Psi_{n,\mu^*,\theta}(\mathbb G_{\mu^*})$.
Similarly, applying $h_{q_{n,\mu^*,1-\alpha-d}(\theta)}$ to $B_{n,\mu^*}(\theta)$ recovers the bound
\begin{align}
\label{E:conlb2}
\inf_{\substack{\mu^*\in\mathbf P\\
\theta\in\Thetaw(\mu^*)}}
\PP{\mu^*}{
B_{n,\mu^*}(\theta)
\le
q_{n,\mu^*,1-\alpha-d}(\theta)
+\ve/2
}
&\ge
\inf_{\substack{\mu^*\in\mathbf P\\
\theta\in\Thetaw(\mu^*)}}
\EE{\mu^*}{
h_{q_{n,\mu^*,1-\alpha-d}(\theta)}
(
B_{n,\mu^*}(\theta)
)
},
\end{align}
which in the $\liminf$ is at least $1-\alpha-d$.
As~\eqref{E:conlb1} holds with probability tending uniformly to $1$, \eqref{E:conlb2}, \eqref{E:inf_bound}, and the union bound imply
\begin{align*}
\liminf_{n\ra\infty}
\inf_{\substack{\mu^*\in\mathbf P\\
\theta\in\Thetaw(\mu^*)}}
\PP{\mu^*}{
\sqrt n T_n(\theta)
\le
\hat c_{1-\alpha}(\theta)+\ve
}
\ge
1-\alpha-d.
\end{align*}
As $d$ was arbitrary, \eqref{E:coverage} is proved.

Let $d\in(0,\alpha)$.
Substituting $r=q_{n,\mu^*,1-\alpha+d}(\theta)$ in the bounded Lipschitz display~\eqref{E:boundedlipcor} likewise shows that
\begin{align}
\label{E:boundedliptwo}
\mathbb P^*(
B^*_{n,\mu^*}(\theta)
\le
q_{n,\mu^*,1-\alpha+d}(\theta)
+\ve/2
)
\ge
1-\alpha,
\end{align}
with probability tending uniformly to $1$ in $\mu^*\in\mathbf P$, $\theta\in\Theta$.

The proof of Proposition~\ref{P:inf_finite} shows that the limit processes $\mathbb G_{\mu^*}$ are uniformly $O_p(1)$, and the term $\lambda_n\eta_{\theta,\mu^*}(\phi)-\lambda_nT_{\mu^*}(\theta)$ which appears in $\Psi_{n,\mu^*,\theta}(\mathbb G_{\mu^*})$ is always nonpositive, so $\Psi_{n,\mu^*,\theta}(\mathbb G_{\mu^*})\le\sup_{\phi\in\Phi}\mathbb G_{\mu^*}(\phi)$, and
$$
C
\equiv
\sup_{n,\mu^*,\theta}
q_{n,\mu^*,1-\alpha+d}(\theta) < \infty.
$$
The bound~\eqref{E:boundedliptwo} guarantees
\begin{align}
\liminf_{n\ra\infty}
\inf_{\substack{\mu^*\in\mathbf P\\\theta\in\Theta}}
\PP{\mu^*}{
\hat c_{1-\alpha}(\theta)
\le
C+\ve/2
}
=
1.
\label{E:boundedlipboot2}
\end{align}
On the other hand, because the processes $\mathbb G_{n,\mu^*}$ are $O_p(1)$ uniformly, \eqref{E:etaconv} implies
$$
\sqrt nT_n(\theta)
\ge
\sqrt nT_{\mu^*}(\theta)
-
O_p(1)
$$
uniformly in $\mu^*$ and $\theta$.
Therefore, for any $\delta>0$, there exists some $\Delta$ large enough such that
\begin{align*}
\limsup_{n\ra\infty}
\sup_{\substack{\mu^*\in\mathbf P\\\theta\in\Theta}}
\PP{\mu^*}{
\sqrt n
(
T_n(\theta)-T_{\mu^*}(\theta)
)
\le
C+3\ve/2-\Delta
}
\le
\delta.
\end{align*}
By the definition of $C$ and~\eqref{E:boundedlipboot2}, the above is an upper bound for
\begin{align*}
\limsup_{n\ra\infty}
\sup_{\substack{
\mu^*\in\mathbf P\\
\theta\in\Theta_n^\Delta(\mu^*)
}}
\PP{\mu^*}{
\sqrt nT_n(\theta)
\le
\hat c_{1-\alpha}(\theta)+\ve
},
\end{align*}
which therefore converges to $0$ as $\Delta\to\infty$.
\end{proof}

\section{Two-player entry game}
\label{app:entry_game}

This appendix specializes Example~\ref{ex:entry_game_main} to bivariate standard normal errors, which is the entry game of Example~3.1 in \textcite{beresteanuSharpIdentificationRegions2011} (henceforth BMM), with the error distribution of their Section~3.4.
We show that the completed game has a LIMR (Subsection~\ref{app:entry_setup}), that with no restriction on the equilibrium selector its identified set is BMM's sharp identification region (Subsection~\ref{app:entry_equivalence}), and that its ADF has the same zero set as BMM's criterion and can be approximated arbitrarily well by finite linear programs (Subsection~\ref{app:entry_LP}).

\subsection{Model}
\label{app:entry_setup}

Firms $j=1,2$ choose $y_j\in\{0,1\}$, and firm $j$ earns $y_j(\delta_jy_{-j}+v_j)$ at error term realization $v$, with $\theta\equiv(\delta_1,\delta_2)\in\Theta\subset(-\infty,0)^2$.
Both firms observe $V\equiv(V_1,V_2)$ before they play, and its distribution $H\equiv N(0,1)^{\otimes2}$ is known.
The observable is the action profile $Z=Y\in\mathcal Y\equiv\{(0,0),(1,0),(0,1),(1,1)\}$, with true distribution $\mu^*\in\Pa(\mathcal Y)$.
We identify each $q\in\Pa(\mathcal Y)$ with its probability vector in $\R^4$, coordinates ordered as in $\mathcal Y$, and write $e_y$ for the point mass at $y$.

The solution concept is Nash equilibrium in mixed strategies, and the equilibria are those described by BMM \parencite[Example~3.1]{beresteanuSharpIdentificationRegions2011}.
Call $B_{1,\theta}\times B_{2,\theta}$, with $B_{j,\theta}\equiv[0,-\delta_j)$, the multiplicity cell, and let $\mathcal N_\theta\equiv\{v\in\R^2:v_j\in\{0,-\delta_j\}\text{ for some }j\}$ be the set of threshold points, so that $H(\mathcal N_\theta)=0$ because $H$ has a density.
For $v\notin\mathcal N_\theta$ outside the cell, the equilibrium is unique and pure, with action profile $y_\theta(v)$.
For $v\notin\mathcal N_\theta$ in the cell, there are three equilibria: the pure profiles $(1,0)$ and $(0,1)$, and a mixed equilibrium in which firm $j$ enters with probability $\sigma^m_{j,\theta}(v)\equiv v_{-j}/(-\delta_{-j})$.
The mixed equilibrium induces the outcome distribution
$$
q^{\mathrm{mix}}_\theta(v)
\equiv
\bigl((1-\sigma^m_{1,\theta})(1-\sigma^m_{2,\theta}),\ \sigma^m_{1,\theta}(1-\sigma^m_{2,\theta}),\ (1-\sigma^m_{1,\theta})\sigma^m_{2,\theta},\ \sigma^m_{1,\theta}\sigma^m_{2,\theta}\bigr),
$$
with each $\sigma^m_{j,\theta}$ evaluated at $v$.
Let $Q_\theta(v)\subseteq\Pa(\mathcal Y)$ be the set of outcome distributions of all Nash equilibria at $v$, as in BMM.
It is nonempty and compact \parencite[Remark~3.2]{beresteanuSharpIdentificationRegions2011}.
Off $\mathcal N_\theta$, it is $\{e_{y_\theta(v)}\}$ outside the cell and $\{e_{(1,0)},e_{(0,1)},q^{\mathrm{mix}}_\theta(v)\}$ in it.

We complete the game with an equilibrium selector $S\in\mathcal S\equiv\{1,2,3\}$, so the input is $W=(V,S)\in\mathcal W\equiv\R^2\times\mathcal S$.
The output kernel is the outcome distribution of the equilibrium that $s$ selects at $v$: for $v\notin\mathcal N_\theta$,
$$
K_\theta(\cdot\mid v,s)
\equiv
\begin{cases}
e_{y_\theta(v)}, & v\notin B_{1,\theta}\times B_{2,\theta}, \\[0.2em]
e_{(1,0)}, & v\in B_{1,\theta}\times B_{2,\theta},\ s=1, \\[0.2em]
e_{(0,1)}, & v\in B_{1,\theta}\times B_{2,\theta},\ s=2, \\[0.2em]
q^{\mathrm{mix}}_\theta(v), & v\in B_{1,\theta}\times B_{2,\theta},\ s=3.
\end{cases}
$$
On $\mathcal N_\theta$, let $K_\theta(\cdot\mid v,s)\equiv e_y$ for all $s\in\mathcal S$, where $y$ is the first profile in the order of $\mathcal Y$ that is a pure equilibrium at $v$.
The four pure-equilibrium regions are $\{v_1\le0,v_2\le0\}$, $\{v_1\ge0,v_2\le-\delta_2\}$, $\{v_1\le-\delta_1,v_2\ge0\}$, and $\{v_1\ge-\delta_1,v_2\ge-\delta_2\}$, respectively.
These closed sets cover $\R^2$, so a pure equilibrium always exists and choosing the first one in the stated order is measurable.
Thus $K_\theta(\cdot\mid v,s)\in Q_\theta(v)$ for every $v$ and $s$, and the output operator
\begin{equation}
\label{eq:entry_Loperator}
(\Lgen\gamma)(\{y\})
=
\int_{\R^2\times\mathcal S}
K_\theta(\{y\}\mid v,s)\,\d\gamma(v,s),
\qquad y\in\mathcal Y,
\end{equation}
is linear in $\gamma$.
The admissible set fixes the distribution of the error terms and leaves the conditional distribution of the selector unrestricted:
$$
\Gamma_\theta
\equiv
\left\{\gamma\in\Pa(\R^2\times\mathcal S):\EE{\gamma}{\one\{V\le c\}-H((-\infty,c])}=0\quad\text{for all }c\in\R^2\right\},
$$
where $V\le c$ holds componentwise, so that $\gamma_V=H$.
Assumptions~\ref{asm:measurable}--\ref{asm:domination} hold: $\R^2$ carries its Borel $\sigma$-algebra and $\mathcal S$ and $\mathcal Y$ are finite, $\Gamma_\theta$ is defined by moment equalities, $K_\theta$ is a probability kernel that does not depend on $\gamma$ (condition~\eqref{eq:probability_kernel} is vacuous without covariates), and counting measure on $\mathcal Y$ dominates.

\begin{remark}
\label{rem:entry_tamer_selection}
The completion spans the unrestricted equilibrium selection mechanisms of \textcite{Tamer2003} and \textcite{BerryTamer2007}.
Every $\gamma\in\Gamma_\theta$ disintegrates as $\d\gamma(v,s)=\nu_\gamma(s\mid v)\,H(\d v)$, with the conditional distribution $\nu_\gamma$ of the selector unrestricted.
Its pushforward $\lambda_\gamma(\cdot\mid v)$, the distribution of $K_\theta(\cdot\mid v,S)$ when $S$ has distribution $\nu_\gamma(\cdot\mid v)$, is a selection mechanism in their sense, a probability kernel from $\R^2$ to $\Pa(\mathcal Y)$ that puts mass one on $Q_\theta(v)$, and it generates $\Lgen\gamma$.
Conversely, every such mechanism $\lambda$ equals $\lambda_\gamma$, $H$-almost surely, for the $\gamma\in\Gamma_\theta$ with $\nu_\gamma(s\mid v)=\lambda(\{K_\theta(\cdot\mid v,s)\}\mid v)$ in the cell off $\mathcal N_\theta$, where the three elements of $Q_\theta(v)$ are distinct, and $\nu_\gamma(1\mid v)=1$ elsewhere.
Restrictions across $V$ and $S$, such as independence or symmetric selection between the two pure equilibria, are linear in $\gamma$ because $\gamma_V=H$ is fixed, so they fit the framework of Section~\ref{sec:computation}. Under such restrictions, however, \eqref{eq:entry_dual} below no longer holds as stated.
\end{remark}

\subsection{BMM equivalence}
\label{app:entry_equivalence}

BMM characterize their sharp identification region through the Aumann expectation of the random set $Q_\theta$:
$$
\EE{H}{Q_\theta}
\equiv
\left\{\int q(v)\,H(\d v): q \text{ measurable},\ q(v)\in Q_\theta(v)\text{ for all }v\in\R^2\right\},
$$
and their region is $\Thetaw^{\mathrm{BMM}}\equiv\{\theta\in\Theta:\mu^*\in\EE{H}{Q_\theta}\}$.
Such $q$ are the measurable selections of $Q_\theta$ \parencite[Definition~A.3]{beresteanuSharpIdentificationRegions2011}.
This is equivalent to BMM's definition.
Requiring $q(v)\in Q_\theta(v)$ for every $v$ rather than $H$-almost surely is without loss, since a selection can be redefined as $K_\theta(\cdot\mid v,1)$ on a null set.
No closure is needed, because $\EE{H}{Q_\theta}$ is closed \parencite[p.~1793]{beresteanuSharpIdentificationRegions2011}.
BMM allow selections that depend on randomness beyond $V$, but conditioning such a selection on $V$ gives a selection of $\operatorname{co}Q_\theta$ with the same integral,%
\footnote{For each $u\in\R^4$, $\inner{u,\EE{}{\tilde q\mid V}}=\EE{}{\inner{u,\tilde q}\mid V}\le\max_{q\in Q_\theta(V)}\inner{u,q}$ almost surely.
Taking a countable dense set of $u$ and using that $\operatorname{co}Q_\theta(V)$ is closed and convex gives $\EE{}{\tilde q\mid V}\in\operatorname{co}Q_\theta(V)$ almost surely.}
and $\EE{H}{\operatorname{co}Q_\theta}=\EE{H}{Q_\theta}$ by the convexification property used in the proof below.

\begin{prop}
\label{prop:entry_equivalence}
For every $\theta\in\Theta$, $\Lgen\Gamma_\theta=\EE{H}{Q_\theta}$, and this set is closed in total variation.
Hence $\Mtheta=\EE{H}{Q_\theta}$ and $\Thetaw=\Thetaw^{\mathrm{BMM}}$.
\end{prop}

\begin{proof}
The random set $Q_\theta$ is measurable, as in BMM's Proposition~3.1 and Remark~3.2, compact-valued, and integrably bounded because its values lie in the simplex \parencite[Chapters~1--2]{molchanov2017theory}.
Fix $\gamma\in\Gamma_\theta$.
Because $\gamma_V=H$ and $\mathcal S$ is finite, $\d\gamma(v,s)=\nu_\gamma(s\mid v)\,H(\d v)$ for a probability kernel $\nu_\gamma$, and~\eqref{eq:entry_Loperator} gives
$$
\Lgen\gamma=\int\sum_{s\in\mathcal S}\nu_\gamma(s\mid v)\,K_\theta(\cdot\mid v,s)\,H(\d v).
$$
The integrand lies in $\operatorname{co}Q_\theta(v)$ for every $v$, because $K_\theta(\cdot\mid v,s)\in Q_\theta(v)$.
So $\Lgen\gamma\in\EE{H}{\operatorname{co}Q_\theta}$, which equals $\EE{H}{Q_\theta}$ \parencite[p.~1793]{beresteanuSharpIdentificationRegions2011}.

Conversely, let $\mu=\int q(v)\,H(\d v)$ for a selection $q$.
Off $\mathcal N_\theta$, every element of $Q_\theta(v)$ equals $K_\theta(\cdot\mid v,s)$ for some $s$, so $s^*(v)\equiv\min\{s\in\mathcal S:K_\theta(\cdot\mid v,s)=q(v)\}$ is well defined there, and we set $s^*(v)\equiv1$ on $\mathcal N_\theta$.
The map $s^*$ is measurable because $q$ and $K_\theta$ are.
The distribution $\gamma$ of $(V,s^*(V))$ with $V\sim H$ lies in $\Gamma_\theta$, and $\Lgen\gamma=\int K_\theta(\cdot\mid v,s^*(v))\,H(\d v)=\mu$ because the two integrands agree off the null set $\mathcal N_\theta$.

Finally, on $\Pa(\mathcal Y)$ the total variation distance is half the $\ell^1$ distance in $\R^4$, and $\EE{H}{Q_\theta}$ is closed, as noted above.
Therefore $\Lgen\Gamma_\theta$ is closed in total variation and $\Mtheta=\Lgen\Gamma_\theta$, and the last claim follows from the definitions of $\Thetaw$ and $\Thetaw^{\mathrm{BMM}}$.
\end{proof}
With no restriction on the equilibrium selector, the identified set~\eqref{eq:identified_set} is therefore BMM's sharp identification region.

\subsection{Computation}
\label{app:entry_LP}

Since $\mathcal Y$ is finite, Assumption~\ref{asm:finite_Z} holds, test functions are vectors $\phi\in[0,1]^4$, and the input-space payoff is $q_{\phi,\theta}(v,s)=\inner{\phi,K_\theta(\cdot\mid v,s)}$.
Because the conditional distribution of the selector is unrestricted, the largest model expectation of $q_{\phi,\theta}$ is attained by choosing, at each $v$, the selector value with the largest payoff:
\begin{equation}
\label{eq:entry_dual}
\sup_{\gamma \in \Gamma_\theta} \EE{\gamma}{q_{\phi,\theta}(V,S)}
=
\int \max_{s \in \mathcal S} q_{\phi,\theta}(v,s)\, H(\d v).
\end{equation}
No conditional mixture of the three payoffs exceeds their maximum, and the smallest maximizing index is a measurable selector that attains it.
Off the multiplicity cell, the integrand does not depend on $s$ and integrates to $\inner{\phi,\pi(\theta)}$, where
$$
\pi(\theta)
\equiv
\Bigl(\tfrac14,\ \
\tfrac14+\bar\Phi(-\delta_1)\bigl[\Phi(-\delta_2)-\tfrac12\bigr],\ \
\tfrac14+\bigl[\Phi(-\delta_1)-\tfrac12\bigr]\bar\Phi(-\delta_2),\ \
\bar\Phi(-\delta_1)\,\bar\Phi(-\delta_2)\Bigr)
$$
is the vector contributed by the region of unique equilibrium, $\Phi$ is the standard normal CDF, and $\bar\Phi\equiv1-\Phi$.
Its coordinates sum to $1-H(B_{1,\theta}\times B_{2,\theta})$.
By~\eqref{eq:entry_Loperator}, $\EE{\Lgen\gamma}{\phi}=\EE{\gamma}{q_{\phi,\theta}(V,S)}$.
Combining this identity with~\eqref{eq:entry_dual} gives the following expression for the ADF of Section~\ref{sec:identification}:
\begin{equation}
\label{eq:entry_T_support}
T(\theta)
=
\sup_{\phi \in [0,1]^4} \Bigl\{ \inner{\phi, \mu^* - \pi(\theta)} - \int_{B_{1,\theta} \times B_{2,\theta}} \max\bigl\{ \phi_{(1,0)},\ \phi_{(0,1)},\ \inner{\phi, q^{\mathrm{mix}}_\theta(v)} \bigr\}\, H(\d v) \Bigr\}.
\end{equation}
Introducing a bounded measurable function $r\colon B_{1,\theta}\times B_{2,\theta}\to\R$ gives the equivalent linear program
\begin{equation}
\label{eq:entry_LP}
\begin{aligned}
T(\theta)=\max_{\phi\in[0,1]^4,\,r}\quad
&\inner{\phi,\mu^*-\pi(\theta)}
-\int_{B_{1,\theta}\times B_{2,\theta}}r(v)\,H(\d v)\\
\text{subject to}\quad
&r(v)\ge\phi_{(1,0)},\quad
r(v)\ge\phi_{(0,1)},\quad
r(v)\ge\inner{\phi,q^{\mathrm{mix}}_\theta(v)}\\
&\text{for all }v\in B_{1,\theta}\times B_{2,\theta}.
\end{aligned}
\end{equation}
For each test function, the smallest feasible $r$ is the pointwise maximum in~\eqref{eq:entry_T_support}, so the program is exact.
It is infinite-dimensional because $r$ is a function on the multiplicity cell.
Since $q^{\mathrm{mix}}_\theta$ is continuous on the closure of the cell, replacing $H$ on the cell by a finite-support measure of the same total mass, chosen independently of $\phi$, produces a finite linear program that approximates $T(\theta)$ arbitrarily closely.
The selector is maximized out in closed form, which gives the program without column-and-row generation.
By Theorem~\ref{thm:main_sharpness} and Proposition~\ref{prop:entry_equivalence}, the zeros of $T$ form BMM's sharp identification region.
BMM's criterion (their Theorem~3.2) is a different function of $\theta$ with the same zero set.

\printbibliography[heading=bibintoc]
\clearpage
\end{refsection}

\setcounter{page}{1}
\renewcommand{\thepage}{OA\arabic{page}}
\setcounter{section}{0}
\renewcommand{\thesection}{OA\Alph{section}}
\renewcommand{\theHsection}{OA\Alph{section}}
\begin{center}
  {\Large\bfseries Additional Appendix to\\[0.5ex]
  ``An Adversarial Approach to Identification, Computation, and Inference in Models with a Linear-in-Measures Representation''}\\[1ex]
  Irene Botosaru, Isaac Loh, and Chris Muris
\end{center}

\begin{refsection}
\section{Additional results: parametric binary choice}
\label{app:parametric_details}

This appendix supplies derivations, numerical details, and additional numerical results for Section~\ref{subsec:numerical_example1}, together with the computation and numerical results for Example~\ref{ex:binary_interval}.
Supplemental Appendix~\ref{app:computation_example1} derives the programs and row bounds, while Supplemental Appendix~\ref{app:algorithm} states the general algorithm.

\subsection{Computation for Figure~\ref{fig:numerical_example1_beta}}
\label{app:parametric_implementation}

Fix a candidate $\beta$ and let $(\phi^*,\zeta^*)$ be an optimizer of its auxiliary finite linear program.

\paragraph{Computing $r_{\mathrm{low}}$.}
Under the baseline and serial dependence specifications, we map $a\in\overline\R$ to $(1+e^{-a})^{-1}\in[0,1]$ and partition $[0,1]$ into finitely many intervals.
Under the baseline, for every covariate path $x$ and interval $J$, we compute a number $U_x(J)$ satisfying $U_x(J)\ge q_{\phi^*,\beta}(x,a)$ whenever the transformed value of $a$ lies in $J$.
The rule for $U_x(J)$ is given at the end of this subsection, and $q_{\phi,\beta}$ is defined in Supplemental Appendix~\ref{app:joint_parametric}.
We then set $q_{\mathrm{up}}(x)\equiv\max_J U_x(J)$, which bounds $\sup_{a\in\overline\R}q_{\phi^*,\beta}(x,a)$ from above because the transformed value of every $a$ belongs to one of the intervals.
We set $q_{\mathrm{up}}\equiv\max_x q_{\mathrm{up}}(x)$ and $r_{\mathrm{low}}\equiv\zeta^*-q_{\mathrm{up}}$.

Under serial dependence, the pathwise payoff is
$$
\widetilde q_{\phi^*,\beta}(x,a)
\equiv
\max_{\pi\in\mathcal F(x,a)}\sum_y\phi^*(y,x)\pi_y,
$$
where $\mathcal F(x,a)$ is the Fr\'echet polytope in Supplemental Appendix~\ref{app:marginal_parametric}.
Applying the same partition and interval bounds to $\widetilde q_{\phi^*,\beta}$ gives $q_{\mathrm{up}}(x)\ge\sup_{a\in\overline\R}\widetilde q_{\phi^*,\beta}(x,a)$.
The quantities $q_{\mathrm{up}}$ and $r_{\mathrm{low}}$ then follow as above.

Under fixed effect--error dependence, the finite program contains every reduced input row, so its row residual is zero and we take $r_{\mathrm{low}}=0$.

To tighten $q_{\mathrm{up}}(x)$, we split the interval $J$ with the largest $U_x(J)$ and compute new bounds for its two pieces.
After any number of splits, $\max_J U_x(J)$ remains a valid pathwise upper bound and therefore gives a valid $r_{\mathrm{low}}$ through the formulas above.
This repeated splitting is the interval branch and bound calculation.
If evaluating the payoff at a particular $(x,a)$ gives a value above $\zeta^*$, we add that pair as a row.
Thus, this procedure plays the role of RB in the algorithm described in Supplemental Appendix~\ref{app:algorithm}.

\paragraph{Bisection for $\beta$.}
For fixed $\beta$, let $\overline T(\beta)$ be the value $T_{\mathrm{LP}}(\beta)$ of the auxiliary program at a column-certified $(\mathcal J',\mathcal K')$, and let $\underline T(\beta)\equiv\max\{0,\overline T(\beta)+r_{\mathrm{low}}\}$.
We stop with an \textsc{in} verdict when $\overline T(\beta)=0$, and with an \textsc{out} verdict when $\underline T(\beta)>0$.
If neither condition holds, we tighten the branch and bound calculation or add a violated row.
If the computational budget is exhausted before either condition holds, the coefficient remains \textsc{undecided}.
A value of the auxiliary program is treated as zero only to linear programming accuracy, taken as $10^{-9}$, so neither verdict uses a separate membership threshold for $T(\beta)$.
That numerical zero does not bind at the reported resolution: every coefficient classified \textsc{in} has a value below $10^{-12}$, and the smallest value recorded above the zero exceeds $10^{-8}$.

A coarse grid first locates one or more \textsc{in} coefficients bracketed by \textsc{out} coefficients on both sides, and each endpoint is then bisected between its adjacent \textsc{in} and \textsc{out} values, stopping when the midpoint is \textsc{undecided}.
The final brackets in Figure~\ref{fig:numerical_example1_beta} are narrower than $2\times10^{-4}$.

\paragraph{Validity of the interval bounds.}
Fix an interval $J$ in the partition of $[0,1]$ and let $I=[a_\ell,a_u]\subseteq\overline\R$ be the corresponding interval of fixed effect values.
Write $p_t(a)\equiv H(x_t'\beta+a)$.
Monotonicity of $H$ gives $p_t(a)\in[p_t(a_\ell),p_t(a_u)]$ for every $a\in I$.
Write $\phi_{ij}\equiv\phi^*((i,j),x)$ for $i,j\in\{0,1\}$.
At $T=2$ the baseline payoff is $q\bigl(p_1(a),p_2(a)\bigr)$ with $q(p_1,p_2)=c_0+c_1p_1+c_2p_2+c_{12}p_1p_2$, where $c_0\equiv\phi_{00}$, $c_1\equiv\phi_{10}-\phi_{00}$, $c_2\equiv\phi_{01}-\phi_{00}$, and $c_{12}\equiv\phi_{00}-\phi_{10}-\phi_{01}+\phi_{11}$.
It is affine in each probability separately, so its maximum over the rectangle $[p_1(a_\ell),p_1(a_u)]\times[p_2(a_\ell),p_2(a_u)]$ is attained at a corner, and that corner value bounds $\sup_{a\in I}q_{\phi^*,\beta}(x,a)$ because the curve $a\mapsto\bigl(p_1(a),p_2(a)\bigr)$ stays in the rectangle.
Under serial dependence, the Fr\'echet maximization in closed form gives $\widetilde q_{\phi^*,\beta}(x,a)=\widetilde q\bigl(p_1(a),p_2(a)\bigr)$ with
$$
\widetilde q(p_1,p_2)
=
c_0+c_1p_1+c_2p_2+c_{12}g(p_1,p_2),
$$
where
$$
g(p_1,p_2)
=
\begin{cases}
\min\{p_1,p_2\}, & c_{12}\ge0,\\
\max\{0,p_1+p_2-1\}, & c_{12}<0.
\end{cases}
$$
The function $\widetilde q$ is affine on either side of a single kink line, $p_1=p_2$ in the first case and $p_1+p_2=1$ in the second, so its maximum over the rectangle is attained at a corner or at a point where the kink line meets the rectangle boundary.

On a constant covariate path, $p_1(a)=p_2(a)$, so the payoff is a quadratic or piecewise linear function of the single probability $p_1(a)$, and we set $U_x(J)$ equal to its exact maximum over $[p_1(a_\ell),p_1(a_u)]$.
On every other path, we use the maximum over the rectangle.

\subsection{The programs behind Figure~\ref{fig:numerical_example1_joint}}
\label{app:parametric_reductions}

Fix a coefficient $\beta$ in the identified set.
Throughout this subsection we let the fixed effect take the boundary values $\pm\infty$, with $H(-\infty)=0$ and $H(+\infty)=1$.
Let
$$
\Gamma_\beta(\mu^*)
\equiv
\left\{
\gamma\in\Gamma_\beta:
\Lbeta\gamma=\mu^*
\right\}
$$
be the set of admissible input measures that reproduce $\mu^*$, and let $h_\beta(w)$ denote the counterfactual integrand defining $\tau_{\mathrm{ASF}}$.
Under the baseline and serial dependence specifications, $h_\beta(w)=H(\bar x'\beta+a)$; under fixed effect--error dependence, $h_\beta(w)=\one\{\bar x'\beta+a-v_1\ge0\}$.
The ASF bounds at this $\beta$ are the values of the linear programs%
\footnote{With $\tau_{\mathrm{ASF}}$ fixed as well as $\beta$, an approximating sequence with finite fixed effects may reach the output only in the limit, but mixing it with vanishing weight on a finite fixed effect whose ASF value falls on the other side of $\tau_{\mathrm{ASF}}$ hits the target exactly without changing the limiting output, so the interval below is the section of the identified set at $\beta$.
Only interior targets arise here: $\bar x=1$ forces $\mu^*(Y_1=1,X_1=1)\le\tau_{\mathrm{ASF}}\le1-\mu^*(Y_1=0,X_1=1)$, and both probabilities are positive in the probit design.}
$$
\underline\tau_{\mathrm{ASF}}(\beta)
=
\inf_{\gamma\in\Gamma_\beta(\mu^*)}
\int h_\beta(w)\,\gamma(\d w),
\qquad
\overline\tau_{\mathrm{ASF}}(\beta)
=
\sup_{\gamma\in\Gamma_\beta(\mu^*)}
\int h_\beta(w)\,\gamma(\d w).
$$
They share the matching constraints and differ only in the direction of optimization, so no adversarial test function is needed.
This is the route behind Figure~\ref{fig:numerical_example1_joint}, rather than the enlarged parameter route of Supplemental Appendix~\ref{app:computation_example1}, in which the counterfactual enters as a moment restriction.
Each specification is solved through a reduction that removes the fixed effect.
Under the baseline and under serial dependence, on every covariate path with a period at the counterfactual covariate value, the path's counterfactual value equals the observed outcome probability in that period and is point identified, so the width comes from the remaining path alone.
On that path the observable content is the first $T$ moments of the distribution of the period probability under the baseline, giving a truncated Hausdorff moment problem that we solve through its polynomial majorant dual, and the distribution of the success count under serial dependence (the outcome distribution on that path is exchangeable in this design), giving the convex order constraint, which has a closed form for log-concave links.
Under fixed effect--error dependence the reduced input space is finite and each per-path section is a transportation program over orderings and regions.
The coefficient endpoints at three and four periods are located by bisection on the enclosure.
Under fixed effect--error dependence the finite program is exact.
Compactness and convexity of the pairs of output measure and ASF value attainable with the extended fixed effect imply that the counterfactual values form a closed interval with both endpoints attained.
The model preserves the finite observable $X$ and imposes no restriction across covariate paths, so both programs decompose by covariate path, and the bounds are the averages of the per-path bounds weighted by the observed path probabilities.

\subsection{The coefficient set along a probit--logit link mixture}
\label{app:parametric_mixture}

To examine the transition from probit to logit under the baseline, consider the link mixture $H_c=(1-c)\Phi+c\Lambda_s$, where $\Phi$ is the standard normal CDF, $\Lambda_s$ is the standardized logistic CDF, and $c\in[0,1]$.
For each $c$, we generate the probability measure from the baseline design of Section~\ref{subsec:numerical_example1} with link $H_c$ and compute the coefficient set under the same link.
The sets contract toward $\beta_0$ as the link moves from probit to logit, where the conditional logit argument point identifies the coefficient (Figure~\ref{fig:numerical_example1_mixture}).

\begin{figure}
    \centering
    \includegraphics[width=0.65\linewidth]{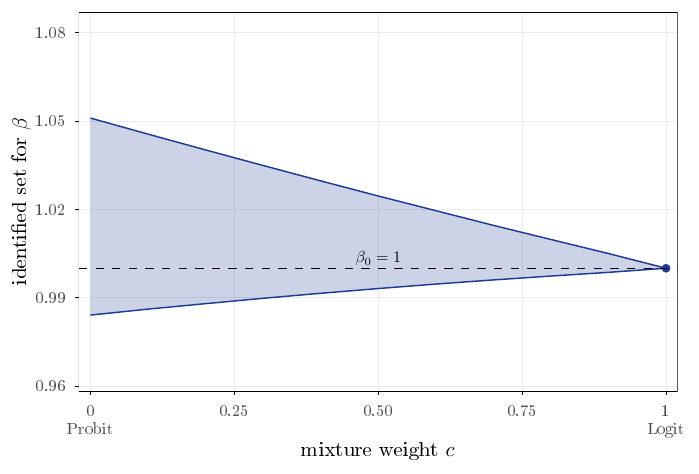}
    \caption{Baseline coefficient sets along the link mixture $H_c=(1-c)\Phi+c\Lambda_s$, from standardized probit ($c=0$) to standardized logit ($c=1$), in the $T=2$ design of Section~\ref{subsec:numerical_example1}.}
    \label{fig:numerical_example1_mixture}
\end{figure}

\subsection{Computation for Figure~\ref{fig:numerical_example1_power}}
\label{app:parametric_inference_implementation}

Little model-specific work remains for inference: the sample and bootstrap statistics are the linear programs of Supplemental Appendix~\ref{app:inference_computation}, and their only example-specific input is the row bound $r_{\mathrm{low}}$, computed using the interval branch and bound procedure of Subsection~\ref{app:parametric_implementation}, now at the optimizer of the sample program.
For each candidate coefficient and each Monte Carlo sample, column-and-row generation on the sample program therefore runs once and terminates with a retained row set $\mathcal W'$, the bound $r_{\mathrm{low},n}$, and the enclosure $\underline T_n\le T_n(\beta)\le T_{\mathrm{LP},n}(\beta)$.
Every bootstrap program is solved on this same $\mathcal W'$, with no further rows and no row bound calculations.
Under fixed effect--error dependence the reduced row family is finite and complete, so $r_{\mathrm{low}}=0$.

Monte Carlo samples are drawn once and shared across coefficients, penalties, and specifications, and the observation paths are nested across sample sizes so that the comparison in $n$ is paired.
The design uses $B=499$ bootstrap draws, $n\in\{10^3,\,10^4,\,10^5\}$, and $200$ Monte Carlo samples.
Critical values are empirical $(1-\alpha)$ quantiles from the $B$ bootstrap replicates, and the tolerance $\varepsilon$ of Corollary~\ref{C:inf_finite} is $10^{-6}$.
The penalty is the rule $\lambda_n=2\sqrt{\log n}$ of Section~\ref{subsec:numerical_example1}.
Figure~\ref{fig:numerical_example1_penalty} repeats the rejection curves of Figure~\ref{fig:numerical_example1_power} over the ladder $c\in\{1/2,1,2,5\}$, using $100$ Monte Carlo samples.
The constant matters most at the smallest sample size, and the spread across constants closes as the sample size grows.

\begin{figure}
    \centering
    \includegraphics[width=\linewidth]{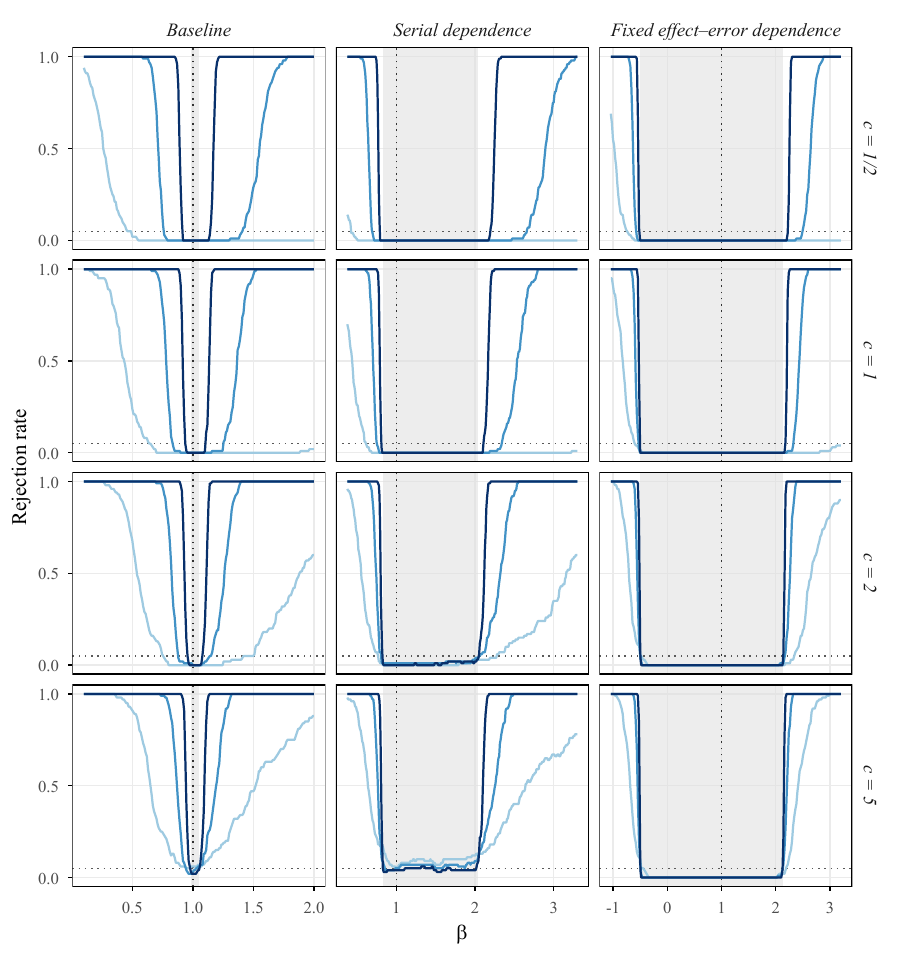}
    \caption{Rejection rates for the coefficient across the penalty ladder, at nominal level $0.05$ and $100$ Monte Carlo samples.
    Rows are the constant in $\lambda_n=c\sqrt{\log n}$, columns are the three specifications, and the curves are the sample sizes $n\in\{10^3,10^4,10^5\}$, light to dark.
    The shaded band is the identified set, and the dotted lines mark $\beta_0=1$ and the nominal level.}
    \label{fig:numerical_example1_penalty}
\end{figure}

\subsection{Interval-censored covariates}
\label{app:interval_parametric}

We now turn to the panel model of Example~\ref{ex:binary_interval}: the baseline specification of Example~\ref{ex:binary_parametric}, with the covariate path observed only through a bracket.
For each endpoint path $b=(x_L,x_U)$, the latent path $x^\star$ lies in the box $B_b\equiv\{x^\star:x_{L,t}\le x_t^\star\le x_{U,t}\ \text{componentwise for every }t\}$, so $\mathcal W=\bigsqcup_b(\{b\}\times B_b\times\R)$.
The payoff and auxiliary program are those of the baseline model, now with $x^\star$ ranging over $B_b$.
Because the covariate enters through the finite observable $b$ and no restriction links different endpoint paths, coefficient feasibility reduces to a separate check for each endpoint path.

The row residual problem differs from the baseline only in also maximizing over the latent path $x^\star\in B_b$.
At fixed $a$, the payoff is multilinear in the period probabilities $p_t\equiv H((x_t^\star)'\beta+a)$, and each period index $(x_t^\star)'\beta$ is minimized and maximized at a vertex of its own box.
The payoff is therefore maximized at one of the $2^T$ paths that take, in each period, either the index-minimizing or the index-maximizing vertex.
Collecting those paths in $C_b(\beta)$, the row residual for endpoint path $b$ is the baseline one-dimensional problem in $a$,
$$
r_b^*=\inf_{a\in\overline\R}\ \min_{x^\star\in C_b(\beta)}\bigl\{\zeta^*-q_{\phi^*,\beta}(b,x^\star,a)\bigr\}.
$$
Each endpoint path then gets the same full-line branch and bound as the baseline (Subsection~\ref{app:parametric_implementation} above).

\subsection{Interval-censored covariates: results}
\label{app:interval_results}

For the numerical illustration, we set $T=2$ and $\beta_0=0.25$, and let $A$ take the values $\{-0.5,0,0.5\}$ with probabilities $(0.25,0.5,0.25)$.
The latent covariates $X_{t}^\star\sim\operatorname{Uniform}[0,1]$ and error terms $V_{t}\sim\mathcal N(0,1)$ are i.i.d.\ over time, mutually independent, and independent of $A$.
The econometrician observes only which of $N$ equal-width bins contains each $X_{t}^\star$.
The model does not restrict $A$ to the three-point support of the design.

\begin{figure}
\centering
\includegraphics[width=0.7\linewidth]{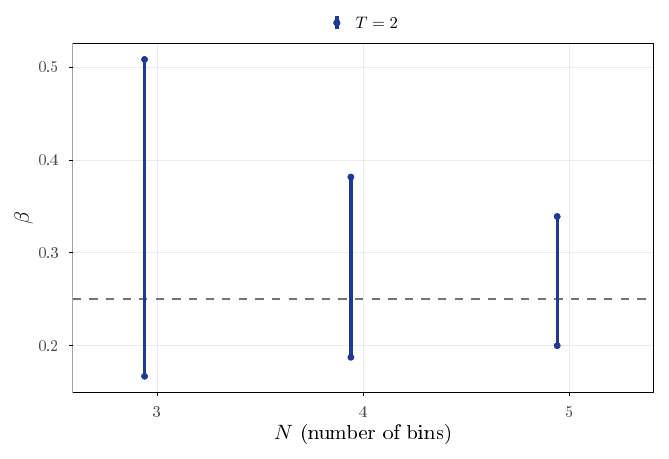}
\caption{Identified sets for $\beta$ with interval-censored covariates under the probit specification, at $T=2$ across bin counts $N\in\{3,4,5\}$. The dashed line marks $\beta_0=0.25$.}
\label{fig:numerical_interval_beta}
\end{figure}

With three bins the identified set is $[0.17,0.51]$ (Figure~\ref{fig:numerical_interval_beta}): the sign of $\beta$ is identified even though the latent covariate path is known only through coarse brackets.
Four bins narrow it to $[0.19,0.38]$ and five bins to $[0.20,0.34]$, a reduction of nearly $60\%$ in width relative to three bins.

\end{refsection}

\clearpage
\begin{refsection}

\section{Additional results: sequential exogeneity}
\label{app:se_details}

This appendix supports Supplemental Appendix~\ref{app:semiparametric_binary_choice}: Subsection~\ref{app:se_history_reduction} writes the row residual problem in terms of the observable histories, and Subsection~\ref{app:se_cells} shows that finitely many values of $\beta$ settle the identified sets of Section~\ref{subsec:numerical_semiparametric}.
Notation is that of Supplemental Appendix~\ref{app:se_cell_reduction} with $\theta=\beta$, so the counterfactual threshold $\bar c$ is dropped from $\mathcal C(\beta)$, and $M$ and the cells $C_j$ are read off the reduced threshold set.

\subsection{The row residual problem in terms of observable histories}
\label{app:se_history_reduction}

Fix $\nu\in\Delta$ and let $p(y,x)$ be a probability mass function on $\{0,1\}^T\times\mathcal X$.
For a history $x^t$, write $n_t(x^t;p)\equiv\Pr_p(X^t=x^t)$ and $m_t(x^t;p)\equiv\Pr_p(Y_t=1,X^t=x^t)$.
For each initial covariate value $x_1$, let $\bar F_{\nu,x_1}$ assign to every threshold $c\in\mathcal C(\beta)$ a number in $[0,1]$, weakly decreasing in $c$, to be read as the survival probability $\Pr(\tilde V_1\ge c\mid X_1=x_1)$ at this $\nu$.

\begin{prop}
\label{prop:se_history_reduction}
Fix $\beta$ and $\nu$.
A probability mass function $p$ of $(Y,X)$ is generated by a distribution of $(X,\tilde V)$ with covariate marginal $\nu$ that satisfies~\eqref{eq:se_fiber_se} if and only if its $X$ marginal is $\nu$ and there are weakly decreasing $\bar F_{\nu,x_1}$ with
\begin{equation}
\label{eq:se_history_factorization}
m_t(x^t;p)
=
\bar F_{\nu,x_1}\bigl(c(t,x_t;\beta)\bigr)\,n_t(x^t;p)
\end{equation}
for every $t$ and every history $x^t$ with initial value $x_1$.
Hence $\{\mu_\pi:\pi\in\Pi(\nu)\}$ is the set of such $p$, and Proposition~\ref{prop:se_finite_support} holds in these terms.
\end{prop}

\begin{proof}
\emph{Necessity.}
Under~\eqref{eq:se_fiber_se}, $\Pr(Y_t=1\mid X^t=x^t)=\Pr(\tilde V_t\ge c(t,x_t;\beta)\mid X^t=x^t)=\bar F_{\nu,x_1}(c(t,x_t;\beta))$ on every history of positive probability, and multiplying by the probability of the history gives~\eqref{eq:se_history_factorization}, with both sides zero on histories of probability zero.

\emph{Sufficiency.}
For each $x_1$ choose a distribution $F_{x_1}$ on $\R$ with survival probabilities $\bar F_{\nu,x_1}$ at the thresholds.
Draw $(Y,X)$ from $p$ and, given $(y,x)$, draw the $\tilde V_t$ independently from $F_{x_1}$ conditioned on $\tilde V_t\ge c(t,x_t;\beta)$ when $y_t=1$ and on $\tilde V_t<c(t,x_t;\beta)$ when $y_t=0$ (by~\eqref{eq:se_history_factorization}, an event of $F_{x_1}$-probability zero is conditioned on only on an event of $p$-probability zero).
Given $X^t=x^t$, the first conditional distribution is used with probability $m_t(x^t;p)/n_t(x^t;p)=F_{x_1}([c(t,x_t;\beta),+\infty))$, so $\tilde V_t$ given $X^t=x^t$ has distribution $F_{x_1}$ for every $t$, which is~\eqref{eq:se_fiber_se}, and the outcome equation returns $Y$.
The last claim is Lemma~\ref{lem:se_cell_reduction}.
\end{proof}

The restrictions~\eqref{eq:se_history_factorization} are products of a survival probability and a history probability, so the row residual problem~\eqref{eq:se_pricing} in these terms, with objective $\inner{\phi^*,p}-\zeta^*$, is bilinear rather than linear, but it has $|\mathcal X|2^T$ history probabilities and at most $|\mathcal X_1|\,M$ survival probabilities at each $\nu$, where $\mathcal X_1\equiv\{x_1:x\in\mathcal X\}$, in place of $|\mathcal X|(M+1)^T$ cell probabilities.
We bound its maximum from above by relaxing each product on pieces of the range of the survival probability \parencite{McCormick1976}, and that bound serves as $\delta_{\mathrm{ub}}$ in Supplemental Appendix~\ref{app:se_exact_pricing}.

\subsection{Coefficient cells}
\label{app:se_cells}

In the designs of Section~\ref{subsec:numerical_semiparametric} the thresholds $-(t-1)-\beta x_t$ with $x_t\in\{0,1,2\}$ are affine in $\beta$, and two of them coincide only at $\beta=(r-t)/(x-x')$ with $x\ne x'$, so the tie values are $\{-1,-1/2,0,1/2,1\}$ at two periods and $\{-2,-1,-1/2,0,1/2,1,2\}$ at three.

\begin{prop}
\label{prop:se_cells}
\begin{enumerate}
\item[(a)] The set $\mathcal M_\beta$, the ADF $T(\beta)$, and the auxiliary program~\eqref{eq:se_inner_lp} at any fixed list of values of $\nu$ depend on $\beta$ only through which thresholds in $\mathcal C(\beta)$ coincide and how the distinct ones are ordered, so in the scalar designs they are constant on each open interval between consecutive tie values and on the two unbounded intervals beyond them.
\item[(b)] In the scalar designs, let $b$ be a tie value and $I_-$ and $I_+$ the open intervals on either side of it.
Then $\mathcal M_b\subseteq\mathcal M_{I_-}\cap\mathcal M_{I_+}$, so $b$ is excluded from the identified set whenever $\mu^*$ lies outside $\mathcal M_{I_-}$ or outside $\mathcal M_{I_+}$.
\end{enumerate}
\end{prop}

\begin{proof}
\emph{Part~(a).}
The order determines $M$, the cells $C_j$, and the indicators $\one\{m_j\ge c\}$ with $c\in\mathcal C(\beta)$, and $\beta$ enters $\mathcal M_\beta$, $T(\beta)$, and the auxiliary program only through these (Lemma~\ref{lem:se_cell_reduction} and Proposition~\ref{prop:se_finite_support}).
In the scalar designs the order changes only at a tie value.

\emph{Part~(b).}
Fix $\beta'$ in $I_-$ or in $I_+$.
No tie value lies between $b$ and $\beta'$, so thresholds that differ at $b$ keep their order at $\beta'$, and thresholds that coincide at $b$ separate at $\beta'$ unless they are the same affine function.
Given $(p,\bar F)$ satisfying~\eqref{eq:se_history_factorization} at $b$, give each threshold at $\beta'$ the survival probability it had at $b$: the array is weakly decreasing, because every strict inequality between thresholds at $\beta'$ holds weakly at $b$, and~\eqref{eq:se_history_factorization} holds at $\beta'$ with the same $p$.
Hence $\mathcal M_b\subseteq\mathcal M_{\beta'}$, and the last claim follows from Theorem~\ref{thm:main_sharpness} because $\mathcal M_\beta$ is closed (Proposition~\ref{prop:se_finite_support}(b)).
\end{proof}

By Proposition~\ref{prop:se_cells}(a), one value of $\beta$ in each interval between tie values settles the interval.
It is included when the auxiliary program of Supplemental Appendix~\ref{app:se_exact_pricing} reaches value zero, and excluded when the lower end of the enclosure there is positive.
A tie value next to an excluded interval is excluded by Proposition~\ref{prop:se_cells}(b), and the remaining tie values are treated like the intervals.
The three-period identified set is contained in the two-period one, because dropping the last period of an admissible three-period input measure gives an admissible two-period input measure that generates the two-period marginal of $\mu^*$.
At three periods the list of values of $\nu$ in the auxiliary program is replaced by a list of pairs $(\nu,\pi)$ with $\pi\in\Pi(\nu)$, whose mixtures are output measures of the model by Proposition~\ref{prop:se_finite_support}(a), so the enclosure of Supplemental Appendix~\ref{app:se_exact_pricing} applies unchanged.

\printbibliography[heading=subbibliography]
\end{refsection}

\end{document}